\documentclass{book}

\usepackage{mathptmx}
\usepackage{helvet}
\usepackage{courier}
\usepackage{type1cm}

\usepackage{makeidx}         
\usepackage{graphicx}        
\usepackage{multicol}        
\usepackage[bottom]{footmisc}

\makeindex             

\usepackage{amsmath,amssymb,amsfonts,amsthm,latexsym, amssymb,mathrsfs}
\usepackage{hyperref}

\newenvironment{sproof}{%
  \proof}{\endproof}

\newtheorem{theorem}{Theorem}[section]

\newtheorem{proposition}{Proposition}[section]
\newtheorem{lemma}{Lemma}[section]
\theoremstyle{definition}
\newtheorem{definition}{Definition}[section]
\newtheorem{remark}{Remark}[section]

\begin{document}

\author{Claudio Dappiaggi, Valter Moretti and Nicola Pinamonti}
\title{Hadamard States From Light-like Hypersurfaces}
\maketitle

\frontmatter

%
%

\chapter*{Acknowledgements}

The authors are grateful to the department of mathematics of the University of Genoa, to that of the University of Trento as well as to the department of physics of the University of Pavia for the kind hospitality during the various phases of the realization of this project. Thanks are also due to Aldo Rampioni and to Kirsten Theunissen at Springer SBM NL for the fruitful and patient collaboration in the realization of this book.

\tableofcontents

\mainmatter

%
%
%
\chapter{Introduction}
\label{intro} 
%
%
%

\hyphenation{Min-kow-ski}

In the last century, our theoretical knowledge of key physical processes has experienced an impressive large and fast growth 
thanks to the birth and to the development of new theories. 
Prime examples are General Relativity, used to describe the gravitational interaction and Quantum Mechanics/Quantum Field Theory, which account for physical phenomena from the atomic to subnuclear scales. Recently, both theories have seen new spectacular experimental confirmations with the detection of gravitational waves and with the discovery of the Higgs Boson at the LHC.

The merge of the two viewpoints in a single unified theoretical body is still problematic. In spite of this hurdle, there are regimes where phenomena arising from the combination of these theories are expected to arise, most notably cosmology.  As a matter of fact, it is nowadays widely believed that the anisotropies, measured in the cosmic microwave background radiation, can be ascribed to the quantum fluctuations of the metric which originated at the time of their emission, that is when the universe was very hot and much smaller compared to its present status.
A thorough, complete theoretical description of such processes is possible only employing a theory of quantum gravity, which is not available yet. In spite of this lack, one can resort to using methods proper of quantum field theory on curved backgrounds to construct perturbatively interacting models which can describe these effects to a high degree of accuracy. 

Yet, it is not difficult to be convinced that quantizing fields over curved, classical spacetimes cannot be done following the same procedure as in Minkowski spacetime. In this case Poincar\'e invariance allows for exploiting several special properties, first and foremost the existence of a space of momenta accessible via Fourier transform. All these tools are no longer available once the underlying background has a non vanishing curvature. Hence one has to resort to considering different approaches and we will be focusing especially on the so-called Algebraic Quantum Field Theory. 
Based on the formulation first proposed by Haag and Kastler in the seminal paper \cite{HaagKastler}, this approach is divided in two steps. First of all one assigns a set of observables to a given physical system, fixing in particular their algebraic properties. These observables are associated to localized regions of the spacetime and such assignment follows a set of mild, physically motivated constraints, {\it e.g.}, inclusion of regions corresponds to inclusion of sets of corresponding observables, to causally separated regions are associated commuting observables, and so on and so forth. Secondly one identifies thereon a so-called (algebraic) state of the system, that is a positive, linear functional whose image can be interpreted as the expectation value of a given observable. While going through the first step is relatively easy, at least for field theories whose dynamics is ruled by a linear partial differential equation, the identification of a state with good physical properties is difficult even in the simplest scenarios. 

For those familiar with the standard approach to quantum field theory on Minkowski spacetime, this statement might look surprising at first glance. Yet one has to bear in mind the so-called Poincar\'e vacuum, which we are used to consider, enjoys a uniqueness property which is strongly tied to the underlying background being maximally symmetric. On a generic curved spacetime, obtaining a similar result is no longer conceivable and hence, a notion of preferred vacuum state ought to be selected by different physical principles, such as, for example, requiring the minimization either of the energy density or of other physically significant quantities. A typical scenario, where this procedure works, is that of a free field on a stationary spacetime, namely a background whose metric possesses a timelike Killing field. In this case, by defining a notion of frequency via the Fourier transform along a coordinate subordinated to such vector field, one can select a distinguished ground state whose modes contain only positive frequencies. 

The main goal of this book is to show that there exists a large class of spacetimes of physical interest and possessing suitable asymptotic geometric structures on top of which one can consider free field theories identifying explicitly physically sensible, distinguished states, enjoying at least asymptotically a uniqueness property. A notable aspect of this result lies in a non trivial use of a technique, often dubbed {\em bulk-to-boundary injection}, which calls for embedding injectively the algebra of observables of a given theory into a second ancillary counterpart which is intrinsically built on a codimension $1$-submanifold of the underlying background, which is usually the (conformal) boundary of the spacetime itself. The advantage of this procedure lies on the possibility of identifying a unique algebraic state on such boundary, whose counterpart in the bulk turns out to possess all desirable physical properties.

\section[Algebraic Approach]{Algebraic approach to quantum field theory on curved spacetimes}

The standard approach used to quantizing a free field theory in a flat spacetime is tied to the choice of a Fock space which is built over a unique vacuum state and over the corresponding one particle Hilbert space. These can be chosen either by requiring maximal symmetry with respect to the action of the Poincar\'e group or by minimizing the energy density measured by any inertial observer, the final result being the same.  

Yet, if the background on which the field propagates is no longer flat, this point of view cannot be applied so easily. As a matter of facts, since, on a generic spacetime there exist neither a preferred time nor a preferred notion of energy, nor a sufficiently large group of isometries, one cannot select a distinguished reference state.
A prototypical example of this difficulty consists of considering a free field theory on a curved background and selecting two states
whose behaviour is close to that of a vacuum state though at two different spacetime points. It is not difficult to show that one can give rise to inequivalent Hilbert space representations of the same algebra and there is no a priori reason to claim that one of the two choices is better than the other. 
  
The algebraic approach
provides a way out from this conundrum. It suggests to invert the standard point of view and instead to base the theory on the choice of a collection of observables together with their mutual relations, including thus information on the dynamics and on the canonical commutation relations. This defines a unital algebra $\mathcal{A}$, which needs to be decorated in addition with a $^*$-operation (an antilinear involution), a structure which allows to select the positive or self-adjoint elements in the algebra. This first step can be accomplished without having to invoke any particular choice of a Hilbert space representation, this being recovered only at a later stage once an algebraic state is chosen. It will turn out that the $^*$-involution corresponds at a level of Hilbert space to the Hermitian conjugation.

Focusing again on the construction of an algebra of observables, this is a well-understood procedure at least for free field field theories living on any globally hyperbolic spacetime $(M,g)$ and whose dynamics is ruled by an hyperbolic partial differential equation. Although there is a vast literature on this topic in which several different scenarios are thoroughly investigated, in this work we shall mainly focus on the special case of a real scalar field theory $\varphi:M\to\mathbb{R}$ where $M$ is four-dimensional, while the dynamics is ruled by the Klein-Gordon equation
\begin{equation}\label{eq:wave-equation}
\Box \varphi - \xi R\varphi - m^2\varphi = 0,
\end{equation}
where $\Box$ is the d'Alembert operator associated to the metric $g$,  $\xi\in\mathbb{R}$ stands for the coupling to the scalar curvature $R$ while $m$ is the mass. As a consequence of $(M,g)$ being globally hyperbolic, the solutions of \eqref{eq:wave-equation} can be constructed in terms of a corresponding classical Cauchy problem, that is assigning initial conditions on a suitable spacelike codimension $1$ submanifold. An equivalent description of the space of solutions of \eqref{eq:wave-equation} can be obtained observing that, on this class of spacetimes, fundamental solutions for the operator $\Box-\xi R-m^2$ exist. By imposing causal properties of the map from initial data to solutions, we can identify a unique pair of Green operators, known as the advanced and retarded fundamental solutions. Their difference yields the so-called ``causal propagator'' $G$, a bi-distribution of $M$ which has a twofold role. On the one hand it allows for a concrete and covariant characterization of the space of solutions of \eqref{eq:wave-equation} in terms of a suitable class of functions over $M$. On the other hand it is used at a quantum level to implement the canonical commutation relations of an  abstract linear self-adjoint {\em quantum} field $\phi(f)$ smeared with real functions  $f\in C^{\infty}_0(M)$.
In other words the antisymmetric part of the product of two such fields is fixed to be
\begin{equation}\label{eq:CCRintro}
[\phi(f),\phi(f^\prime)]=\phi(f)\phi(f^\prime) - \phi(f)\phi(f^\prime) = i\hbar G(f,f^\prime)\mathbb{I}.
\end{equation}
The abstract field operators  $\phi(f)$ are the {\em generators} of the complex unital $^*$-algebra {\em of observables} $\mathcal{A}$, 
which is thus made of finite complex linear combinations of finite products of smeared field operators. For a modern review see, e.g., \cite{Brunetti:2009pn, FredenhagenRejzner,Brunetti:2015vmh}.

The second step in the algebraic framework consists of identifying a state $\omega$, that is a linear functional over $\mathcal{A}$ which is normalized and positive. Notice that the standard probabilistic interpretation proper of quantum theories is recovered interpreting $\omega(A)$ as the mean expectation value of repeated measurements of the observable $A$ on the state $\omega$. Furthermore, once a state is chosen, the standard picture based on a Hilbert space and on linear operators acting thereon can be recovered from $(\mathcal{A},\omega)$ via the celebrated {\em Gelfand-Neimark-Segal construction}, \cite{GelfandNaimark, Segal}. More precisely, once a state $\omega$ over a $^*$-algebra $\mathcal{A}$ is chosen it is possible to construct a quadruple $(\mathcal{H}_\omega,\pi_\omega, \mathcal{D}_\omega, \Psi_\omega)$ where $\mathcal{H}_\omega$ is a Hilbert space, $\Psi_\omega\in \mathcal{D}_\omega$ is a normalized vector, $\pi_\omega$ is a $^*$-homomorphisms from $\mathcal{A}$ to an algebra of linear operators defined on the  common dense invariant subspace  $\mathcal{D}_\omega\subseteq \mathcal{H}_\omega$ with $\pi_\omega(\mathcal{A})\Psi_\omega = \mathcal{D}_\omega$, such that 
\[
\omega(a) = \left\langle \Psi_\omega| \pi_\omega(a) \Psi_\omega\right\rangle_\omega\:,  \quad \forall a \in \mathcal{A}
\]
where $\langle \cdot | \cdot\rangle_\omega$ is the scalar product in $\mathcal{H}_\omega$.
Any other such set $(\mathcal{H},\pi, \mathcal{D}, \Psi)$  such that $\omega(a) = \left\langle \Psi, \pi(a) \Psi\right\rangle$ for all $a\in \mathcal{A}$, is unitarily equivalent to $(\mathcal{H}_\omega, \pi_\omega, \mathcal{D}_\omega, \Psi_\omega)$.

The construction of the unital $^*$-algebra $\mathcal{A}(M)$ described so far is quite rigid, in the sense that, once a globally hyperbolic spacetime and an hyperbolic, linear partial differential operator are chosen, nothing more is needed to obtain $\mathcal{A}(M)$. Furthermore, if a globally hyperbolic spacetime $M_1$ is isometrically embedded in a second one $M_2$ and if we constructs the associated $^*$-algebras of observables $\mathcal{A}(M_1)$ and $\mathcal{A}(M_2)$ for the same field theory, there exists a $^*$-homomorphism which embeds $\mathcal{A}(M_1)$ into $\mathcal{A}(M_2)$. 
 Moreover causally separated subspacetimes imply commuting associated algebras.
This is the heart of the construction of a generally locally covariant theory \cite{BFV} based on the language of categories, which has the net advantage of allowing to identify observables among different theories and thus to compare experiments in a spacetime independent way.  

To conclude this short digression, we make a last observation. Notice that the commutation relations fix not only uniquely the antisymmetric part of the product $\phi(f_1)\phi(f_2)$, but also the product $\phi(f_1)\dots \phi(f_n)$ among $n$ linear fields up to the totally symmetric parts. Here $f_i\in C^\infty_0(M)$ for all $i=1,...,n$. One can push this comment one step further interpreting the product in $\mathcal{A}(M)$ as a formal deformation of a classical commutative product. This leads to analysing the algebraic approach in terms of a deformation quantization, see \cite{Brunetti:2009qc} and \cite{Waldmann} for further comments in this direction.

\section{States with good physical properties}

As already said, in the algebraic picture a {\em state} is a linear normalized positive functional $\omega:\mathcal{A}(M) \to \mathbb{C}$. Furthermore, the image of any $a\in\mathcal{A}(M)$ under the action of $\omega$ can be interpreted as the mean value of repeated measurements of the observable $a$ on the state $\omega$. As previously discussed, once a state $\omega$ is selected, via the GNS construction we may represent the observables in terms of operators on a suitable Hilbert space, where $\omega$ is represented by preferred normalized vector $\Psi_\omega$
restoring the standard picture proper of quantum theories.

The key notable difference between the choice of a state and the construction of $\mathcal{A}$ is that the latter is well-understood and uniquely defined for any free field theory on a globally hyperbolic spacetime. On the contrary the former is a subtle operation since there is no mean a priori to select a preferred $\omega$. Hence we are forced to a case-by-case analysis. As we will show later, if we consider the complex  unital $^*$-algebra $\mathcal{A}(M)$ generated by linear Hermitian fields $\phi(f)$, smeared with real functions $f\in C^\infty_0(M)$, fixing a state is tantamount to assigning suitable {\em $n-$point functions} $\omega_n \in \mathscr{D}'(M^n)$:
\begin{align}
&\omega(\phi(f_1)\dots \phi(f_n)) \doteq\nonumber\\
&  \omega_n(f_1,\ldots, f_n) \doteq \int_{M^n} \omega_n(x_1,\dots, x_n)   f_1(x_1)\dots f_n(x_n) \;d\mu_g(x_1)\cdots d\mu_g(x_n).\nonumber
\end{align}
Besides the compatibility with the dynamics, one has to keep in mind that only the totally symmetric part of $\omega_n$ is not constrained by the properties of the product of $\mathcal{A}(M)$, while the canonical commutation relations (\ref{eq:CCRintro}) fix the remaining ambiguities.

Among the plethora of all possible states, most of them have pathological physical behaviour and we need to find a selecting criterion to distinguish those which are acceptable. Already at a preliminary, heuristic level, we know that the Poincar\'e vacuum, which is used in the standard quantization picture of free field theories on Minkowski spacetime, possesses all wanted properties. It is thus reasonable to expect that a state could be considered physically acceptable if its ultraviolet behaviour mimics that of the Poincar\'e vacuum, an idea which we will make mathematically rigorous throughout the text. 
In addition, we observe that the Poincar\'e vacuum possesses another distinguishing feature, namely its $n-$point functions $\omega_n$ are completely determined by $\omega_2$. This is a prototypical feature of an important class of states, dubbed {\em quasi-free}, or equivalently {\em Gaussian}. More precisely a state is called quasi-free/Gaussian if its $n-$point functions with odd $n$ vanish, while for even $n$ $\omega_n$ is fully determined from $\omega_2\in\mathcal{D}'(M\times M)$ as follows
\[
\omega_{2n}(x_1,\dots, x_{2n}) = \sum_{\pi}\omega_2(x_{\pi(1)},x_{\pi(2)})\dots \omega_2(x_{\pi(n-1)},x_{\pi(2n)})
\]
where the sum is over all ordered permutations $\pi$ of $\{1,\dots, 2n\}$ such that  $\pi(2i-1)\leq \pi(2i)$ and
$\pi(2i-1)\leq \pi(2j-1)$ for every $i, j\in\{1,\dots, n\}$ with $i<j$. 

In a quasi-free state the antisymmetric part of the two-point function is proportional to the casual propagator $G$ which is the building block of the canonical commutation relations. At the same time, the symmetric part $\omega_{2,s}$ is constrained by the positivity of the state, which implies that
\[
\omega(\phi(f)^*\phi(f)) =  \omega_{2}(\overline{f},f)  \geq 0, \qquad f \in C^\infty_0(M;\mathbb{C}).
\]
This is equivalent to 
\[
\omega_{2,s}(f,f)\geq 0, \qquad   \frac{1}{4}|G(f,h)|^2 \leq \omega_{2,s}(f,f)\omega_{2,s}(h,h), \qquad f,h\in C^\infty_0(M). 
\]
Observe that, despite the strong constraints imposed by these inequalities on $\omega_{2,s}$, a large freedom in its choice is left.

A rather special scenario is that of a linear field theory, {\it e.g.} a Klein-Gordon field, living on a static globally hyperbolic spacetime, hence with a timelike Killing field, such as for example the Minkowski background. In this case one can construct a distinguished state, often called the vacuum of the theory, as the quasi-free state whose two-point function is obtained considering only the positive frequencies of the causal propagator. The above inequalities are automatically implemente since the positive and negative frequencies in the causal propagator appear, up to a sign, with the same weight thanks time reversal being an isometry of the underlying metric.

While this procedure works perfectly, it cannot be extended to a generic spacetime, since it is strongly tied to the existence of a specific isometry. Nonetheless, a close investigation of the two-point functions $\omega_2(x,x^\prime)$ constructed with the method outlined above unveils that their singular behaviour closely mimics that of the Poincar\'e vacuum, which in turn is controlled by kinematic and geometric data, such as the mass of the field and the geodesic distance between $x,x^\prime$. As a matter of fact, if we consider just the massless, scalar field on Minkowski spacetime, Poincar\'e invariance dictates that the integral kernel of $\omega_2$ is nothing but $\left[\eta(x-x^\prime,x-x^\prime)\right]^{-2}$, $\eta$ being the Minkowski metric.

While the example of static spacetime strengthens the na\"ive idea that the ultraviolet behaviour of a physically acceptable state should mimic that of the Poincar\'e vacuum, making this idea mathematically precise on a generic globally hyperbolic spacetime requires different and more advanced tools. As a matter of fact the correct framework to probe the singular behaviour of a distribution on a manifold is known to be microlocal analysis \cite{Hormander} which allows to translate such kind of information in terms of 
the {\em wave front set} of a given distribution. 
Microlocal techniques were introduced in the study of physically acceptable states in the late nineties leading to the formulation of  the {\em microlocal spectrum condition} \cite{Radzikowski, BFK}, according the which the wavefront set of the two-point function of a quasi-free state should be the same of that of the Poincar\'e vacuum, namely
\[
WF(\omega_2) =  \left\{ (x,x',k,k')\in T^*(M\times M)\setminus\{0\}, (x,k)\sim(x',-k') ,k\triangleright 0   \right\}
\]
where $(x,k)\sim(x^\prime,k^\prime)$ if $x$ and $x'$ are connected by a null geodesic $\gamma$, if the co-vector $k$ is co-parallel to $\gamma$ at $x$ and if the parallel transport of $k$ along $\gamma$ to $x'$ coincides with $k^\prime$.  Finally, $k\triangleright 0$ means that $k$ ought to be future directed.

Notice that the latter condition is a reminiscence of the requirement to consider only positive frequencies when constructing states for free field theories on a spacetime with a global timelike Killing field, though with the net advantage that the wavefront set prescription is manifestly covariant. It is again a remarkable fact that the microlocal spectrum condition yields a strong constraint on the form of the two-point function. More precisely, as proven by Radzikowski \cite{Radzikowski}, on any normal neighbourhood, the two-point function of a quasi-free state satisfying the microlocal spectrum condition and KG equation in both entries locally takes the {\em Hadamard form}
\[
\omega(x,x^\prime) = \lim_{\epsilon\to 0}\frac{1}{8\pi^2} \left ( \frac{u(x,x^\prime)}{\sigma_\epsilon(x,x^\prime)} +v(x,x^\prime) \log\left(\frac{\sigma_\epsilon(x,x^\prime)}{\lambda^2}\right) + w(x,x^\prime) \right)
\]
where $\sigma_\epsilon(x,x^\prime) = \sigma(x,x^\prime) +i \epsilon (T(x)-T(x^\prime)) + \epsilon^2$, $\sigma$ being the squared geodesic distance while $T$ is a temporal function. The other unknowns $u,v,w$ are smooth function on each geodesic neighbourhood, $u$ and $v$ being completely determined in terms of the metric and of the parameters in the equation of motion. The only freedom lies in the choice of $\lambda$, a reference squared length and of $w$, which is symmetric and constrained by the dynamics and by the requirement of positivity of the state. The Hadamard form was actually known much earlier then the microlocal characterization of states, see e.g.  \cite{DeWittBrehme, KW}. However, although it might look surprising at first glance, a concrete use of the Hadamard form is rather problematic on a generic spacetime since it would involve a control of the form of the two-point function in each geodesically convex neighbourhood. On the contrary the microlocal characterization is in many concrete scenarios a very effective and practical tool. To conclude this short excursus, we mention that the existence of Hadamard states, namely those satisfying the microlocal spectrum condition, is not questionable since it can be proven using a deformation argument, adapting to the case in hand an analysis of Fulling, Narcowich and Wald \cite{FNW}.

\section[Bulk to Boundary Correspondence]{Bulk to boundary correspondence and construction of Hadamard states}

Although the existence criterion, based on the deformation argument \cite{FNW}, is reassuring concerning the well-posedness of the microlocal spectrum condition, it is rather moot since it does not provide any mean to choose a preferred  Hadamard state. To infer any physical consequence from a given model, one needs to bypass this hurdle and therefore several physically meaningful construction schemes for Hadamard states 
in a given spacetime were devised in the past years.

Different approaches have been followed, many tied to the existence of specific symmetries such as those of cosmological spacetimes. More general schemes have been investigated only recently, a notable case being discussed in \cite{Gerard:2012wb}. Here states are constructed in terms of data on a given, yet arbitrary Cauchy surface $\Sigma$ embedded in a globally hyperbolic spacetime. A second approach, on which we shall focus, consists of considering, in place of a Cauchy surface, an initial characteristic surface, exploiting that the Goursat problem for a normally hyperbolic operator is well-posed \cite{Gerard:2014hla}. Although this procedure can be discussed in full generality, goal of this book is to show how it can be used in a large class of concrete examples which have studied in the past ten years \cite{Dappiaggi:2005ci, Moretti, Moretti2, Dappiaggi:2007mx, Dappiaggi:2008dk, Dappiaggi:2009fx}. These include specific, physically relevant scenarios, ranging from the construction of ground states on asymptotically flat and on cosmological spacetimes to the rigorous definition of the Unruh state on a Schwarzschild black hole.

To conclude the introduction, we sketch briefly and heuristically the idea which lies at the heart of the construction of Hadamard states exploiting the existence of a characteristic, initial surface. A mathematically rigorous analysis of what follows will be the content of the rest of this book. The starting point of our investigation will be to consider globally hyperbolic spacetimes $(M,g)$ possessing a (conformal) null boundary, which we indicate here for simplicity as $\mathcal{C}$. On top of $M$ we shall consider a free field theory whose dynamics is ruled by an hyperbolic partial differential equation, so that we can associate to such system a complex unital $^*$-algebra of observables $\mathcal{A}(M)$. The second step consists of focusing on $\mathcal{C}$ building on top of it a second and ancillary complex unital $^*$-algebra $\mathcal{A}(\mathcal{C})$ which is built in such a way of being able to prove the existence of an injective $^*$-homomorphism
\[
\iota:\mathcal{A}(M)\to\mathcal{A}(\mathcal{C}).
\]
From a physical point of view, this procedure can be interpreted as the existence of an embedding of the bulk observables into a boundary counterpart. For this reason one refers to $\iota$ as implementing a {\em bulk-to-boundary correspondence}. From a structural point of view, this map can be understood as an extension to characteristic surfaces of the time-slice axiom, see {\it e.g.} \cite{ChilianFredenhagen}. The latter says that $\mathcal{A}(\mathcal{O})$, the algebra of observables defined in a globally hyperbolic subregion $\mathcal{O}\subset M$ containing a whole Cauchy surfaces $\Sigma$ of $M$ is $^*$-isomorphic to $\mathcal{A}(M)$. Here this result is extended replacing $\mathcal{A}(\mathcal{O})$ with $\mathcal{A}(\mathcal{C})$, though the price to pay is the loss of the $^*$-isomorphism, which is replaced by an injective $^*$-homomorphism. It is possible to avoid such restriction by a careful choice of the boundary algebra as shown by G\'erard and Wrochna in \cite{Gerard:2014hla}.

The advantage of considering null hypersurfaces does not lie only in the loss of one dimension, but rather in their geometric structure. As a matter of fact, in all scenarios that we consider $\mathcal{C}$ turns out to be diffeomorphic to $\mathbb{R}\times\mathbb{S}^2$ and realized as the union of complete null geodesics. Hence one can identify on top of $\mathcal{C}$ a global null coordinate with respect to which one can define a Fourier-Plancherel transform. This peculiar feature is used to define a quasi-free state $\omega_\mathcal{C}$ for $\mathcal{A}(\mathcal{C})$ whose two-point function is constructed in terms of the positive frequencies identified out of the said null coordinate. Furthermore, in all cases that we shall consider, we will be able to prove that the ensuing state satisfies a uniqueness criterion on $\mathcal{C}$. As a last step we can combine $\omega_{\mathcal{C}}$ with $\iota$ to build $\omega_M \doteq  \iota^{*}\omega_{\mathcal{C}}$, namely,
\[
\omega_M(a) = \omega_{\mathcal{C}}(\iota(a))\quad \forall a \in \mathcal{A}(M)\:.
\]
The outcome is a state on $\mathcal{A}(M)$ which, on the one hand, can be interpreted as an asymptotic ground state enjoying in addition many natural geometric properties. On the other hand, it can be proven to be of Hadamard form. To show this last statement one needs to make a clever use of microlocal techniques, especially H\"ormander propagation of singularity theorem.  

To conclude this section we stress that, although in this book, we will only focus on real scalar fields, the methods devised can be used also to discuss other scenarios, such as Dirac fields \cite{DHP} free electromagnetism \cite {DappiaggiSiemssen, Siemssen:2011gma} or linearised gravity \cite{Benini:2014rya}.

%
%
%
\chapter{General geometric setup}
\label{geometry} 



Goal of this section is to discuss the geometric background on which our investigation is based. As we have already indicated in the introduction, we will not give a full-fledged analysis of all structures, but we will focus on the data essential to us. Yet we will make sure to point an interested reader to the relevant literature. We will also assume that the reader is familiar with the basic tools proper of differential geometry. Within this respect we will follow mainly the notations and conventions of \cite{Wald}.

\section{Globally hyperbolic spacetimes}
In this subsection we will follow mainly \cite{BGP,Wald}.
\begin{definition}
	A {\bf  spacetime} $M$ is a  connected Hausdorff second-countable orientable $4$-dimensional smooth manifold endowed with a smooth Lorentzian metric $g$ of signature $(-,+,+,+)$. 
\end{definition}

\noindent The requirement on the dimensionality of $M$ is only based on our desire to make a close contact with the examples that we will be discussing in detail. Most of the general results and of the constructive aspects of the theory is valid for Lorentzian manifolds with dimension greater or equal to $2$.

The Lorentzian character of the metric $g$ plays an important distinguishing role for the pair $(M,g)$, since we can define two additional concepts: {\it time orientability} and {\it causal structures}. We briefly recall that, for any point $p\in M$ using  the scalar product  $g_p$ in $T_pM$ induced by the metric $g$, we can divide the elements of $T_pM\setminus \{0\}$ into three distinct categories.  A vector $v\in T_pM\setminus \{0\}$ is called {\bf timelike} if $g_p(v,v)<0$, {\bf spacelike} if $g_p(v,v)>0$,  {\bf lightlike} (also said {\bf null}) if $g_p(v,v)=0$.  {\bf Causal vectors} are those either timelike or lightlike. Co-vectors are classified similarly.

A smooth $3$-dimensional embedded submanifold $S \subseteq M$, often referred to as a {\bf $3$-dimensional hypersurface}, 
is {\bf spacelike}, {\bf timelike}, {\bf lightlike} ({also said {\bf null}) if, respectively, its normal co-vector is  spacelike, timelike, lightlike.

In addition, for a fixed $p \in M$  we can define a two-folded {\bf light cone} $V_p\subseteq T_pM\setminus\{0\}$
made of all causal vectors and we have the freedom to call {\bf future-directed} the non-zero vectors lying in one of the two-folds. If such choice can be made smoothly varying  $p\in M$, we say that $(M,g)$ is {\bf time orientable}.
It turns out that if a spacetime is time orientable, being connected per assumption, there are only two inequivalent such choices each called {\bf time orientation}. 
\begin{remark}
 Henceforth all our spacetimes will be supposed to be {\bf time oriented}, {\em i.e.}, a choice of time orientation has been made.
\end{remark}
The subdivision of each $T_pM\setminus \{0\}$ into the said three subsets is at the heart of the definition of causal structures for $(M,g)$. Consider a {\em continuous} and {\em piecewise-smooth} curve $\gamma: I\to M$ where, indifferently,  
$I= (a,b)$ or $I= [a,b)$ or $I=(a,b]$ or $I=[a,b]$ and
we  henceforth suppose ${\mathbb R}\cup \{-\infty\} \ni a<b \in {\mathbb R}\cup \{+\infty\} $. Assuming  that its  tangent vector $\dot{\gamma}$ vanishes nowhere, we shall call $\gamma$ {\bf timelike} ({\em resp.} {\bf spacelike}, {\bf lightlike}) {\bf curve} if at each point of the image $\dot{\gamma}$ is timelike ({\em resp.}  spacelike,  lightlike). If the tangent vector to $\gamma$ is nowhere spacelike and is everywhere future ({\em resp.}  past) directed, we say that it is a {\bf future} ({\em resp.} {\bf past}) {\bf-directed}  {\bf causal} curve.  {\bf Future}   ({\em resp.}  {\bf past}) {\bf -directed} {\bf timelike}  curves are defined analogously. Each {\bf causal curve} is either future- or past-directed.

Any such curve, say $\gamma: I\to M$ with either $I=(a,b]$ or $I=(a,b)$
 is said to be  {\bf past inextensible}
if there is no  causal curve $\gamma':  I' \to M$ with $I'\supset I$ and $\inf I' < a$ such that 
$\gamma'|_{I}=\gamma$.  {\bf Future  inextensiblility} is analogously defined. Zorn's lemma assures that every causal curve can be completed into an inextensible causal curve.

We have all ingredients to introduce the defining building blocks of the causal structure of $(M,g)$. We call 
\begin{itemize}
	\item $J^\pm(p)$ the {\bf causal future (+) / past (-)} of $p\in M$, that is the union between $\{p\}$ and all points $q\in M$ such that there exists a future- (+) / past- (-) directed, causal curve $\gamma: [a,b]\to M$ with  $\gamma(a)=p$ and $\gamma(b)=q$,
	\item $I^\pm (p)$ the {\bf chronological future (+) / past (-)} of $p\in M$ that is the collection of all points $q\in M$ such that there exists a future- (+) / past- (-) directed, timelike curve $\gamma: [a,b]\to M$ with  $\gamma(a)=p$ and $\gamma(b)=q$.
\end{itemize}
For subsets $U\subseteq M$, we similarly define $J^\pm(U)\doteq\bigcup\limits_{p\in U}J^\pm(p)$ and $I^\pm(U)\doteq\bigcup\limits_{p\in U}I^\pm(p)$. 

$N,N' \subseteq M$ are said to be {\bf causally disjoint} if $(J^+(N) \cup J^-(N)) \cap N' = \emptyset$ which is equivalent to requiring  $(J^+(N') \cup J^-(N')) \cap N = \emptyset$.

\begin{remark}\label{rem:notation} \cite{Wald}  \cite{BEE} \cite{ONeill} $\null$\\
{\bf (1)} $I^\pm(N)$ are always {\em open} sets, even if $N\subseteq M$ is not. The topology of $J^\pm(N)$ is more complicated, however 
sufficiently close to $p\in M$, $\partial J^\pm(p)\setminus \{p\}$ is a null 3-dimensional hypersurface. General results for every  $N\subseteq M$ and points in $M$ are the following:

(a) $I^\pm(N) = \mbox{Int}(J^\pm(N))$ and $N \subseteq J^\pm(N) \subset \overline{I^\pm(N)}$, where $\mbox{Int}$ indicates the collection of interior points,

(b) if $p \in I^\pm(q)$ and $q\in J^\pm(r)$, then $p \in I^\pm(r)$, 

(c) if $p \in J^\pm(q)$ and $q\in I^\pm(r)$, then $p \in I^\pm(r)$, 

(d) if $p \in J^\pm(q)$ and $q\in J^\pm(r)$, then $p \in J^\pm(r)$.\\
{\bf (2)}	The notion of causal or chronological past/future  strongly depends on the choice of the underlying background. When a disambiguation will be necessary we will employ the more precise notation $J^\pm(U;N)$ where $U\subseteq N$ and $N$ is an open subset of the whole spacetime $M$ equipped with the restriction of the metric and the relevant causal curves defining  $J^\pm(U;N)$ are those completely included in  $M$.  We shall use a similar notation regarding chronological past/future. \\
{\bf (3)} A subset $O \subset M$  of a spacetime $(M,g)$ is said to be  {\bf causally convex} if every causal curve joining $p,q \in O$ has image completely enclosed in $O$.\\
{\bf (4)} A spacetime $(M,g)$ is said to be  {\bf strongly causal} if, for every $p\in M$ and  every open subset $U\ni p$, there is an open causally convex subset  $V$ with $p \in V\subseteq U$. In strong causal spacetimes, the family of open sets $I^+(p) \cap I^-(q)$, $p,q \in M$ is a topological basis of the (pre-existent) topology of $M$.
\end{remark}
Time orientation and causal structures open several possibilities for constructing spacetimes where closed timelike or causal curves exist, hence formalizing at a geometric level the fascinating idea of time travel.  From a physical point of view, we are instead interested in identifying a suitable class of spacetimes which, on the one hand, avoids any such pathology, while, on the other hand, it allows for existence and uniqueness  theorems  for solutions of hyperbolic partial differential equations, such as the scalar D'Alembert wave equation in terms of an initial value problem. The answer to these queries goes under the name of {\it globally hyperbolic spacetimes} (see, {\em e.g.},  \cite[Ch. 8]{Wald}).

We begin by introducing an {\bf achronal} subset of a spacetime $M$, namely a subset $\Sigma$ such that $I^+(\Sigma)\cap\Sigma=\emptyset$. Subsequently we associate to $\Sigma$ its  {\bf future domain of dependence} as the collection $D^+(\Sigma)$ of points $p\in M$ such that every past-inextensible causal curve passing through $p$ intersects $\Sigma$ somewhere.  One defines analogously  the {\bf past domain of dependence}  $D^-(\Sigma)$ of $\Sigma$ and its   {\bf domain of dependence} 
$D(\Sigma) \doteq D^+(\Sigma)\cup D^-(\Sigma)$.

\begin{definition}\label{globhyp}
A {\bf Cauchy hypersurface} of a spacetime $M$ is defined as a closed achronal subset $\Sigma \subseteq M$  such that $D(\Sigma) =M$.
A spacetime $M$ is called {\bf globally hyperbolic} if it possesses a Cauchy hypersurface.
\end{definition}
Globally hyperbolic spacetimes represent  the canonical class of backgrounds on which quantum field theories are constructed and $\Sigma$ is the natural candidate to play the role of the hypersurface on which initial data can be assigned, provided $\Sigma$ is 
sufficiently smooth. Yet, according to Definition \ref{globhyp}, only the existence of a single (hopefully smooth) Cauchy hypersurface is guaranteed. This is slightly disturbing since there is no {\em a priori} reason  why an initial value hypersurface for a certain partial differential equation should be distinguished. In addition Definition \ref{globhyp} does not provide any concrete criterion to establish whether a given spacetime $M$ with an assigned metric $g$ is globally hyperbolic or not. An important step forward in this direction is represented by the work of Bernal and Sanchez, 
\cite{Bernal, Bernal:2005qf}, who, by means of deformation arguments, gave an alternative, very informative additional characterization of globally hyperbolic spacetimes also proving the existence of {\em smooth spacelike} Cauchy surfaces. We shall report it essentially following the formulation given in Section 1.3 of \cite{BGP}:

\begin{proposition}\label{BS}
	Let $(M,g)$ be any time-oriented   spacetime. The following two statements are equivalent:
	\begin{enumerate}
		\item $(M,g)$ is globally hyperbolic;
		\item $(M,g)$ is isometric to $\mathbb{R}\times\Sigma$ with metric $-\beta dt \otimes dt +h_t$, where $\beta, h \in C^\infty(\mathbb{R}\times\Sigma)$, $\beta$ is strictly positive,  each $h_t$ is a smooth Riemannian metric on the smooth  manifold $\Sigma$, and each  $\Sigma_t \doteq \{t\}\times\Sigma$
identifies  to a  (smooth embedded co-dimension $1$) submanifold of $M$ which is a spacelike Cauchy surface. 
	\end{enumerate}
\end{proposition}
\begin{remark}\label{rem:propJ} The following properties are enjoyed by any globally hyperbolic spacetime $(M,g)$ \cite{Wald}:\\
{\bf (1)} It is strongly causal, hence item (4) in Remark \ref{rem:notation} holds true.\\
{\bf (2)} For every $p,q\in M$, $J^+(p) \cap J^-(q)$ is either empty or compact.\\
{\bf (3)} For every $p,q\in M$, $J^\pm(p) \cap J^{\mp}(\Sigma)$ and thus
also $J^\pm(p) \cap \Sigma$ a
are  compact if $\Sigma$ is a Cauchy surface in the future (+), resp., past (-) of $p$.\\
{\bf (4)}  If $S \subseteq M$ is compact (in particular $S=\{p\}$), then $J^\pm(S)= \overline{I^\pm(S)}$, $J^\pm(S)\setminus I^{\pm}(S) = \partial J^\pm(S)= \partial I^\pm(S)$.
\end{remark}
There are several known examples in the literature of globally hyperbolic spacetimes and most of the physically interesting scenarios are based on this notion. In this work we will consider some of these cases in detail and, hence, we will not discuss further this concept. A reader interested in a few concrete examples can consult either \cite{Wald} or the list given in \cite{Benini:2013fia}. It is important to stress that globally hyperbolic spacetimes do not exhaust the collection of spacetimes  used in physical applications, even within the algebraic approach of QFT, {\it e.g.}, the half Minkowski spacetime which appears in the investigation of the Casimir-Polder effect \cite{Dappiaggi:2014gea} or anti-de Sitter spacetime, the maximally symmetric solution of  Einstein's equations with negative cosmological constant \cite{Dappiaggi:2016fwc,Duetsch:2002hc,Duetsch:2002wy,Ribeiro:2007hv}. 
%
%

\section{Spacetimes with distinguished light-like $3$-surfaces}\label{Sec:asymptotically_flat0}
The next step consists of introducing the class of spacetimes which we will be considering throughout this text. As mentioned in the introduction we will focus on several specific scenarios where light-like $3$-surfaces play a central role.
\subsection{Asymptotically Flat Spacetimes}\label{Sec:asymptotically_flat}
The first scenario is at the heart of this section. Most notably we will present the spacetimes which are said to be {\em asymptotically flat at null infinity}. Heuristically speaking, they are those manifolds whose asymptotic behaviour far away along null directions mimics that of Minkowski spacetime. A rigorous mathematical characterizations of these backgrounds and a thorough analysis of the associated geometric properties has been subject of investigations starting from the early sixties. A reader who is interested in these aspects can consult \cite{Geroch}, \cite{AX},\cite{Friedrich, Friedrich2, Friedrich3},  and \cite[Ch. 8-11]{Wald}. Our presentation will follow mainly these references and we will review the structures and the related properties which play a key role in the next chapters.
\begin{definition}\label{Def:asymptotically_flat}
A ($4$-dimensional) spacetime $(M,g)$ (dubbed, {\em physical spacetime} or {\em bulk}),  is {\bf asymptotically flat at future null infinity} if the following objects  exist. 
	\begin{enumerate}
		\item[a)] A ($4$-dimensional) spacetime $(\widetilde{M},\widetilde{g})$ (dubbed, {\em unphysical spacetime}).
		\item[b)] A smooth embedding $\psi:M\to\psi[M]\subset\widetilde{M}$, $\psi[M]$ being an open subset of $\widetilde{M}$.
 \item[c)]  A smooth function $\Omega\in C^\infty(\psi[M])$ fulfilling $\Omega >0$ and
		$\widetilde{g}|_{\psi[M]}=\Omega^2 \:\psi^*g$.
	\end{enumerate}
	The following further five conditions must hold true.
	\begin{enumerate}
		\item $\psi[M]\subseteq \widetilde{M}$ is a manifold with boundary
$\Im^+ \doteq \partial \psi[M]$
 given by an embedded three-dimensional submanifold 
  of $\widetilde{M}$ which satisfies $\Im^+ \cap J^-(\psi[M], \widetilde{M})= \emptyset$.
		\item $(\widetilde{M},\widetilde{g})$ is 
		strongly causal in an open neighborhood of $\Im^+$.
		\item $\Omega$ extends (not uniquely in general) to a smooth function on the whole
		$\widetilde{M}$, still denoted by $\Omega$, such that $\Omega=0$  exactly  on $\Im^+$, whereas $d\Omega\neq 0$ at each point of $\Im^+$.
		\item Defining $n^a := \widetilde{g}^{ab} \partial_b \Omega$, there is a smooth function $\omega$ defined on $\widetilde{M}$ with $\omega >0$ on $M \cup \Im^+$ such that $\widetilde{\nabla}_a(\omega^4 n^a)=0$ on $\Im^+$ and the integral curves of $\omega^{-1}n$ are complete on $\Im^+$. This way $\Im^+$ results to be diffeomorphic to 
${\mathbb S}^2 \times {\mathbb R}$, the second factor being the range of the parameter of those integral curves.
\item Vacuum Einstein equations are assumed to be
fulfilled for $(M,g)$ in a neighbourhood of the boundary of $\psi[M]$  (or, more weakly, ``approaching'' it as discussed on p.278 of \cite{Wald}).
\end{enumerate}
We call	 {\bf future null infinity} of $(M,g)$ the set  $\Im^+$.
\end{definition}

\noindent Since the conformal embedding $\psi$ is a diffeomorphism from $M$ onto its image,  we will henceforth omit to indicate it explicitly writing  $M \subseteq \widetilde{M}$.  The  index notation has been used in the last part of the definition to make the statements less obscure. We will resort often to this policy, although all the structures that we  define and the results that we discuss could be presented all with an index-free notation, never being dependent on the choice of a local chart. 

\begin{remark}\label{Rem:alternative_asymptotically_flat}$\null$\\
{\bf (a)} Minkowski spacetime fulfils the given definition \cite{Wald} and in fact, the geometric structure of $\Im^+$ described in the definition above is just the one arising from the analysis of the geometry of null infinity of Minkowski spacetime. In this sense a spacetime is asymptotically flat at null infinity if it ``looks like Minkowski spacetime at null infinity''.\\
{\bf (b)}  More complicated 
completions of our basic definition exist and concern other types of infinity which may be added to $(M,g)$.  One could wish to add the {\em future time infinity}  of $M$.
This notion was first introduced by Friedrich in \cite{Friedrich, Friedrich2, Friedrich3}
to describe spacetimes resembling Minkowski space in the far timelike future, and used in a bulk-to-boundary context in \cite{Moretti}. Dealing with this class of backgrounds (which contains  Minkowski spacetime) we will refer to them as {\bf asymptotically flat (at future null infinity) with future time infinity}. In this setting in addition to Definition \ref{Def:asymptotically_flat}, 
\begin{enumerate}
\item[6.]$\widetilde{M}$ includes a preferred point $i^+$, {\bf the future time infinity}, and
the embedding of the bulk into the unphysical spacetime is assumed to satisfy $J^-(i^+;\widetilde M)\setminus\partial J^-(i^+;\widetilde M)=\psi[M]$ where now $J^-(i^+;\widetilde M)$ is supposed to be closed.
The future null infinity here satisfies $\Im^+= \partial J^-(i^+; \widetilde M)\setminus\{i^+\}$ so that $\partial\psi[M]=\Im^+\cup i^+$. The extended function $\Omega$
is again required to  satisfy $\Omega=0$ exactly on $\partial\psi[M]$, $d\Omega \neq 0$ on $\Im^+$, but $d\Omega(i^+)=0$ with $\widetilde{\nabla}_\mu \widetilde{\nabla}_\nu\Omega(i^+) = -2\widetilde{g}_{\mu\nu}(i^+)$.
\end{enumerate}
{\bf (c)} Analogously, one could work replacing $\Im^+$ and $i^+$ with the {\em past null infinity} $\Im^-$ and {\bf past time infinity}, $i^-$, respectively.  All constructions, properties and results being unchanged. 
Alternatively, one can define a spacetime which is asymptotically flat at null and {\em space} infinity, by adding a (further) point $i_0$ to $\widetilde{M}$, the {\bf space infinity} of $M$, such that around $i_0$ the spacetime ``looks like Minkowski spacetime at space infinity''. A detailed discussion on this sort of asymptotically flat spacetimes appears in \cite{Wald}. We only remark that the geometric structure close to $i_0$ is much more delicate than the one around $\Im^\pm$ or $i^\pm$ and non-trivial regularity issues arise for $\widetilde{g}$ at $i_0$. Assuming the existence of null and spacelike flat infinities, the embedding $\psi$ is supposed to satisfy  $\overline{J^+(i_0;\widetilde{M})}\cup\overline{J^-(i_0;\widetilde{M})}=\widetilde{M}\setminus\psi[M]$
and the null infinities are now defined as $\Im^\pm\doteq \partial J^\pm(i_0;\widetilde{M})\setminus\{i_0\}$.
\end{remark} 
\subsection{A compound example: Schwarzschild spacetime}\label{Sec:Example}
 Definition \ref{Def:asymptotically_flat} with the completion of item (b) in Remark \ref{Rem:alternative_asymptotically_flat} seems very hard to check in concrete examples of spacetimes. Yet, this is not really the case and, besides the obvious case of Minkowski spacetime, there are several other known backgrounds admitting null infinities and other types of infinities. The most famous one is the Schwarzschild solution of Einstein equations which also includes other types of light-like $3$-surfaces: {\em Killing horizons}. Since we will be discussing the construction of the Unruh state on this spacetime, it is worth highlighting it more in detail and, to this end, we will follow mainly the notations and conventions of \cite[Sec. 6.4]{Wald}. 
\begin{figure}\label{fig2}
	\centering
	\includegraphics[height=5cm]{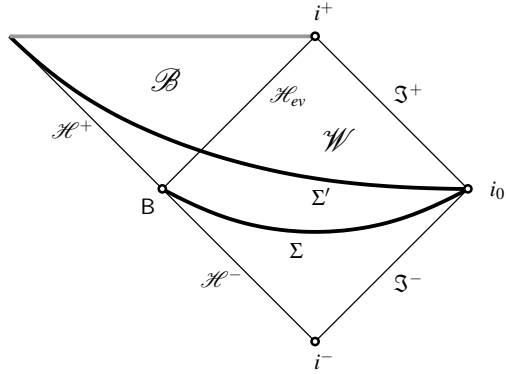}
	\caption{The overall picture represents $\mathcal{M}$.  The regions $\mathcal{W}$ and $\mathcal{B}$ correspond respectively to the
		regions $I$ and $III$ in fig 2. The thick horizontal line denotes the metric singularity at $r=0$, 
		$\Sigma$ is a smooth spacelike Cauchy surface for $\mathcal{M}$  while $\Sigma'$ is a smooth spacelike Cauchy surface for $\mathcal{W}$.}
\end{figure}

\noindent Our starting point is  Kruskal spacetime $\mathcal{K}$ and we are interested in the subspacetime $\mathcal{M}$ used to picture a black hole of mass $M>0$. Referring to Figure 2.1, $\mathcal{M}$ is made of the union of three pairwise disjoint parts, the {\bf Schwarzschild wedge} $\mathcal{W}$, the {\bf black hole} region  $\mathcal{B}$, and  their common boundary $\mathcal{H}_{ev}$ also known as the {\bf (future Killing)  event horizon}. Defining the {\bf Schwarzschild radius}  $r_S\doteq 2M$, the metric outside $\mathcal{H}_{ev}$ is best defined in terms of the {\bf Schwarzschild coordinates} 
$t,r,\theta,\phi$,
where $t\in \mathbb{R}$, $r\in (r_S,+\infty)$, $(\theta,\phi)$ are standard spherical coordinates over $\mathbb{S}^2$ in   ${\mathcal{W}}$,
whereas  $t\in \mathbb{R}$, $r\in (0,r_S)$, $(\theta,\phi)$ cover  $\mathbb{S}^2$ as before  in ${\mathcal{B}}$,
\begin{equation}\label{Schw}
g =-\left(1-\frac{2M}{r}\right) dt\otimes dt + \left(1-\frac{2M}{r}\right)^{-1} dr\otimes dr
+ r^2 d{\mathbb{S}^2}(\theta,\phi)\:,
\end{equation} 
where $d{\mathbb{S}^2} \doteq d\theta \otimes d\theta + \sin^2\theta d\phi\otimes d\phi$ is the standard metric on the unit $2$-sphere. We observe that as $r$ tends to $0$ we approach an {\em intrinsic} (curvature) singularity located at the horizontal upper boundary of $\mathcal{B}$ (Figure 2.1), whereas  $r=r_S$ defines an {\em apparent}  singularity on the event horizon $\mathcal{H}_{ev}$ which actually is just due to the bad choice of coordinates.
Another convenient chart over 
$\mathcal{W}\cup \mathcal{B}$
is that provided by  {\bf Eddington-Finkelstein coordinates} \cite{KW,Wald}: $u,v,\theta,\phi$, with $(\theta,\phi)$ standard spherical coordinates over $\mathbb{S}^2$  and
\begin{align}
&u \doteq  t-r^* \mbox{ in ${\mathcal{W}}$,}  \quad u\doteq-t-r^* \mbox{ in $\mathcal{B}$,}
 \label{one}
\\
&v \doteq t+r^* \mbox{ in ${\mathcal{W}}$,} \quad v\doteq t-r^* \mbox{ in $\mathcal{B}$,}
\label{two}
\\
&r^* \doteq r + 2M \ln \left|\frac{r}{2M}-1 \right| \in \mathbb{R}\:.\notag 
\end{align}
A third, related, chart yields the {\bf global null coordinates} $U,V,\theta,\phi$ which have the advantage of being defined on the whole 
$\mathcal{K}$ \cite{Wald} though here we restrict them to $\mathcal{M}$ only,
\begin{equation}
U  =  -e^{-u/(4M)}\:, \;
V  =  {e^{ v/(4M)}} \; \mbox{in ${\mathcal{W}}$,}\quad
U  =  e^{  u/(4M)}\:, \;
V  =  {e^{ v/(4M)}} \; \mbox{in $\mathcal{B}$}\:.\label{UV}
\end{equation}
In this frame,
\begin{gather*}
\mathcal{W}\equiv  \{ (U,V, \theta,\phi) \in \mathbb{R}^2 \times \mathbb{S}^2\:| \: U<0,V>0\}\:,\\
\mathcal{B} \equiv  \{ (U,V, \theta,\phi) \in \mathbb{R}^2 \times \mathbb{S}^2\:| \:UV<1 \:,  U, V > 0\}\:,\\
{\mathcal{M}} \doteq  \mathcal{W} \cup \mathcal{B} \cup \mathcal{H}_{ev} \equiv  \{ (U,V, \theta,\phi) \in \mathbb{R}^2 \times \mathbb{S}^2\:| \: UV<1\:,  V > 0\}\:. 
\end{gather*}

\begin{figure}\label{fig1}
	\centering
	\includegraphics[height=6.5cm]{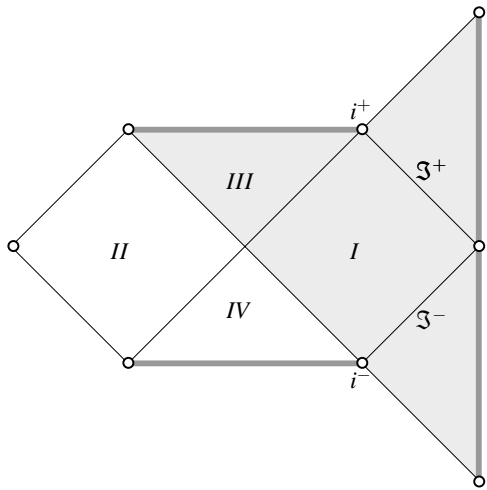}
	\caption{The Kruskal extension of Schwarzschild spacetime, where the physical region is shaded in grey. Following \cite{Ashtekar:1978zz} one can introduce {\bf spatial infinity} $i_0$, while $i^\pm$ are not part of the conformal diagram and we shall refer to them as (formal) {\em points at infinity}, often also known as {\bf future} and {\bf past} infinity respectively. 
	}
\end{figure}
\noindent Each of the three regions is a globally hyperbolic  subspacetime  of $\mathcal{K}$, and they represent the building blocks of our analysis together with  $\mathcal{H}_{ev}$ and with the complete {\bf past (Killing) horizon} $\mathcal{H}$ of $\mathcal{M}$ which is part of the boundary of $\mathcal{M}$ in the Kruskal manifold. These horizons are defined by:
\begin{eqnarray*}
\mathcal{H}_{ev} \equiv  \{ (U,V, \theta,\phi) \in \mathbb{R}^2 \times \mathbb{S}^2\:| \: U = 0, V>0\}
\:,  \\
\mathcal{H} \equiv  \{ (U,V, \theta,\phi) \in \mathbb{R}^2 \times \mathbb{S}^2\:| \: V = 0, U\in \mathbb{R}\}
\:.
\end{eqnarray*}
Following Figure \ref{fig2} and for future convenience, we decompose $\mathcal{H}$ into the disjoint union  $\mathcal{H} = \mathcal{H}^- \cup \mathsf{B} \cup \mathcal{H}^+$  where $\mathcal{H}^\pm$ are defined 
as the loci $U>0$ and $U <0$, while $\mathsf{B}$ is the {\bf bifurcation surface} at $U=0$. This is a  spacelike 
$2$-sphere with radius $2M$ where $\mathcal{H}$ meets the closure of $\mathcal{H}_{ev}$.\\
The metric of the whole Kruskal manifold and, per restriction, also of $\mathcal{M}$, is
\begin{equation}g = -\frac{32 M^3}{r} e^{-\frac{r}{2M}} dU\otimes dV + r^2 d{\mathbb{S}^2}(\theta,\phi)\:,
\label{g}
\end{equation}
where the only intrinsic (curvature) singularity of Kruskal spacetime at $r\to 0$, translates to $UV\to 1$ while the apparent singularity on $\mathcal{H}_{ev}$ has disappeared.\\
We outline succinctly the Killing vector structure since it will play a key role in the next sections. As one can infer directly from \eqref{Schw} or from \eqref{g} , besides the Killing fields associated to the spherical symmetry of the background, there is a further smooth Killing field $X$ which coincides with $\partial_t$. $X$ is timelike in ${\mathcal{W}}$ and orthogonal to $\partial_r,\partial_\theta, \partial_\phi$, which is thus a {\em static spacetime}, while $X$ is spacelike in ${\mathcal{B}}$. 
$X$ becomes lightlike and tangent to $\mathcal{H}$ as well as to both $\mathcal{H}_{ev}$ and to its completion in the Kruskal manifold, while it vanishes on $\mathsf{B}$, giving rise to the structure of a {\em bifurcate Killing horizon} \cite{KW}. In terms of the coordinates $u$ and $v$, it turns out that
\begin{equation*}
X = \mp \partial_u \mbox{ on $\mathcal{H}^\pm$,} \quad 
X = \partial_v \mbox{ on $\mathcal{H}_{ev}$.}
\end{equation*}
To conclude our short survey of Schwarzschild spacetime, we need to make contact with the analysis of the previous section, Definition \ref{Def:asymptotically_flat} in particular. One follows \cite{SW83} rescaling the metric \eqref{g} by a factor $r^{-2}$ after which one can notice that $(\mathcal{M}, r^{-2}g)$  admits a smooth 
larger extension $(\widetilde{\mathcal{M}},\widetilde{g})$ (see Figure 2.2) constructed in accordance with Definition \ref{Def:asymptotically_flat}. Here, the geometric singularity in $(\mathcal{M},g)$ is pushed at infinity in the sense that the non-null geodesics take an infinite amount of affine parameter to reach a point at $r=0$.
The extension of $(\widetilde{\mathcal{M}},\widetilde{g})$ obtained in this way does not cover the timelike and spacelike infinities  $i^\pm$ and $i_0$ in Figure 2.1 and Figure 2.2 (refer to item (b) in Remark \ref{Rem:alternative_asymptotically_flat}), though it includes the boundaries given by the {\bf future} and {\bf past null infinity} respectively ${\Im}^{\pm}$. These
null $3$-submanifolds of $\widetilde{\mathcal{M}}$ are formally localised at $r=+\infty$. (A finer extension which includes $i_0$ is described in great detail in the appendix of \cite{Ashtekar:1978zz} where it is more generally shown that all Kerr solutions of vacuum Einstein equations are asymptotically flat at null and spacelike infinity.)
The restriction of the rescaled extended metric $\widetilde{g}$
to the Killing horizon $\mathcal{H}$ as well as to the null infinities ${\Im}^\pm$ can be written explicitly. More precisely, in the first case,
$$ 
\widetilde{g}|_{\mathcal{H}} = 4M^2\left( - d\Omega \otimes 
dU - dU \otimes d\Omega +d{\mathbb{S}^2}(\theta,\phi)\right)\:,
$$
where $\Omega=2V$ vanishes at $\mathcal{H}$ since $V$ is defined as in \eqref{UV}. Barring the constant pre-factor $4M^2$ the metric is said to be in a {\bf geodetically complete Bondi form}, completeness being referred to the complete domain of the affine parameter $U \in (-\infty,+\infty)$ of the null geodesics forming $\mathcal{H}$.
The same structure occurs on $\Im^+$ and on $\Im^-$, respectively, formally defined by the limit-value of the Eddington-Finkelstein coordinates $v= +\infty$ and $u= -\infty$. The metric $\tilde{g}$ 
has still a {\bf geodetically complete Bondi form},
$$ \widetilde{g}|_{\Im^+} = - d\Omega \otimes du - du \otimes d\Omega + d\mathbb{S}^2(\theta,\phi)\:, 
$$
where $\Omega \doteq -2/v$ defines $\Im^+$ for $\Omega=0$. Similarly 
$$
\widetilde{g}|_{\Im^-} = - d\Omega \otimes dv - dv \otimes d\Omega + d\mathbb{S}^2(\theta,\phi)\:, 
$$
where $\Omega \doteq -2/u$ defines $\Im^-$ for $\Omega =0$. To conclude we observe that the $g$-Killing vector field $X$ coinciding in $\mathcal{W}$ and in $\mathcal{B}$ with $\partial_t$ is an affine Killing vector for 
$\widetilde{g}$ (as evidently $\widetilde{\nabla}_a X_b + \widetilde{\nabla}_b X_a = \mathcal{L}_X(\widetilde{g})_{ab}= X(\ln\Omega^2)\widetilde{g}_{ab}$ in the bulk) and it extends to an affine $\widetilde{g}$-Killing vector, still denoted by $X$, defined on $\widetilde{\mathcal{M}}$ with
$$ X =  \partial_u \mbox{ on $\Im^+$,} \quad 
X =   \partial_v \mbox{ on $\Im^-$}.$$

\subsection{Asymptotic Symmetries: The Bondi-Metzner-Sachs group}\label{Sec:BMS}

Our next goal is to consider an arbitrary asymptotically flat spacetime at future null infinity $(M,g)$ as in Definition \ref{Def:asymptotically_flat} further discussing the geometric properties of future null infinity $\Im^+$. For our purposes it suffices that only $\Im^+$ exists, but, whenever also past null infinity can be defined, a similar analysis for $\Im^-$ holds true.

As starting point, we remark that, in view of the definition of the unphysical spacetime $(\widetilde{M},\widetilde{g})$, the metric structure of $\Im^+$ enjoys a so-called {\em gauge freedom}. It stems from the observation that we can always rescale $\Omega \to \omega \Omega$ in a neighborhood of $\Im^+$ , where $\omega$ is a smooth and nowhere vanishing scalar function. Such freedom does not affect the topology of $\Im^+$, namely that of $ \mathbb{R} \times\mathbb{S}^2$, as well as the differentiable structure. 
Following Definition \ref{Def:asymptotically_flat}, once $\Omega$ is fixed, $\Im^+$  turns out to be
the union of  future-oriented  integral lines of the field 
$n^a \doteq\tilde{g}^{ab}\widetilde{\nabla}_b\Omega$.
This property is, in fact, invariant under gauge transformation, but the field $n$
depends on the gauge. All relevant information can be encoded in a triple of data $(\Im^+,\widetilde{h},n)$, where $\widetilde{h}$ is the degenerate metric induced by $\widetilde{g}$ on $\Im^+$. Under a gauge transformations $\Omega \to \omega \Omega$, these data transform as
\begin{equation}
\Im^+ \to \Im^+,\quad \tilde{h} \to \omega^2 \tilde{h}, \quad n \to \omega^{-1} n \label{gauge}\:.
\end{equation}
For a given asymptotically flat spacetime $(M,g)$, there is no general physical principle which may distinguish one of the above triple of data from another obtained via a gauge transformation.  Such property goes together with that of {\bf universality}: It turns out that the geometric structures at future null infinity of a pair of asymptotically flat spacetimes are always isomorphic up to gauge transformations. 
To wit, if $C_1$ and $C_2$ are two equivalence classes under gauge transformation of triples, associated respectively to two asymptotically flat spacetimes $(M_1,g_1)$ and $(M_2,g_2)$, there exists a diffeomorphism $\gamma: \Im^+_1 \to \Im^+_2$ such that, for suitable $(\Im^+_1,\tilde{h}_1, n_1)\in C_1$
and $(\Im^+_2,\tilde{h}_2, n_2)\in C_2$, 
\begin{equation}
\gamma(\Im^+_1) = \Im^+_2 \:,\:\:\:\:\: \gamma^* \tilde{h}_1=\tilde{h}_2 \:,\:\:\:\:\:\gamma^* n_1=n_2 \nonumber\:.
\end{equation}
The proof of this statement relies on the existence for every  asymptotically flat spacetime $(M,g)$ of a choice of the gauge such that the rescaled unphysical  metric (still indicated by  $\widetilde{g}$ instead of $\omega^2 \widetilde{g}$) reads
\begin{equation}\label{eq:metric_at_Scri}
\widetilde{g}|_{\Im^+}  = d\Omega \otimes du  -du \otimes d\Omega +  d{\mathbb{S}^2}(x_1,x_2)\:. 
\end{equation}
Indeed, Definition \ref{Def:asymptotically_flat} requires that for a given asymptotically flat spacetime $(M,g)$ with an initial $\Omega$, it is always possible to fix the gauge $\omega$ such that both
$\widetilde{\nabla}_a(\omega^4 n^a)=0$ and the integral curves of $\omega^{-1}n$ are complete, {\it i.e.}, their parameter ranges over the whole $\mathbb R$. The first condition  implies that these curves are ligthlike complete geodesics for the rescaled unphysical metric $ \omega^2 \widetilde{g}$. With this choice of $\omega$,
still denoting by $\Omega$ (instead of $\omega\Omega$) the gauge-transformed  conformal factor, with $n$ (instead of $\omega^{-1}n$) the gauge-transformed  tangent vector and with $\widetilde{g}$ (instead of $\omega^2 \widetilde{g}$) the gauge-transformed unphysical metric, 
 in a neighbourhood of $\Im^+$, there exists a coordinate system $(u, \Omega, x_1,x_2)$ so that $ d\mathbb{S}^2(x_1,x_2)$ is the standard metric on a unit $2$-sphere, $u \in \mathbb{R}$ is an affine parameter along the {\em complete} null geodesics forming $\Im^+$  with tangent vector $n= \partial/\partial u$. As in the special example of Schwarzschild spacetime, discussed in the previous section, these are also known as {\bf Bondi coordinates}.  $\Im^+$ is made of the points $(u,0,x_1,x_2)$ with $u \in \mathbb{R}$, $(x_1,x_2) \in \mathbb{S}^2$. Every asymptotically flat spacetime, using \eqref{gauge}, admits the triple of data
$(\Im^+,\tilde{h}_B, n_B) \doteq (\mathbb{R}\times \mathbb{S}^2, d\mathbb{S}^2, \partial/\partial u)$. For more details about the above structures see \cite{Geroch,Wald}.

We focus now on the main topic of this section, namely the so-called {\bf Bondi-Metzner-Sachs (BMS) group}, $G_{BMS}$ \cite{Penrose, Geroch, AS}. First introduced at the beginning of the sixties in \cite{Bondi}
\begin{definition}\label{defBMS}
$G_{BMS}$ is the group of diffeomorphisms $\gamma : \Im^+ \to \Im^+$ 
such that the triple $(\gamma(\Im^+),\gamma^*\tilde{h}, \gamma^*n)$
coincides with $(\Im^+,\tilde{h},n)$ up to a gauge transformation \eqref{gauge}.
\end{definition}
 In other words we are considering only those diffeomorphisms which generalize the notion of isometry 
to an asymptotic structure where the metric structures are equivalent up to gauge transformations. 
(The full group of diffeomorphisms $\gamma : \Im^+ \to \Im^+$  without restrictions is the so-called {\em Newman-Unti group} \cite{Newman:1962cia}.) The following characterization of the one-parameter subgroups of the BMS is rather useful \cite{Wald}:

\begin{proposition}\label{prop1} \em The smooth one-parameter group of diffeomorphisms which is generated by a smooth vector 
	field $\xi'$ on $\Im^+$ is a subgroup of $G_{BMS}$ if and only if 
	$\xi'$ can be smoothly extended  to a, possibly non unique, vector field $\xi$ defined in $M$ in some neighborhood of $\Im^+$ in
	such a way that $\Omega^2 \mathcal{L}_\xi g$ has a smooth extension at $\Im^+$ and vanishes thereon. 
\end{proposition}

\noindent Remembering that $\widetilde{g}$ is a smooth metric also on $\Im^+$ where $\Omega$ vanishes ($\widetilde{g}=\Omega^2 g$ is valid {\em inside} $M$), the condition that $\Omega^2 \mathcal{L}_\xi g$ is smooth and vanishes at $\Im^+$ can be seen as an intepretation of the heuristic idea of a vector field which becomes an exact symmetry only asymptotically. For this reason, taking the universality property into account, we can say that {\em the BMS group describes the  asymptotic Killing symmetries  all asymptotically flat vacuum spacetimes simultaneously}. Proposition \ref{prop1} characterizes also a subalgebra of the smooth vector fields on $\Im^+$, which can be seen as the {\em Lie algebra} of the BMS group since these vectors generate the smooth one-parameter subgroups. Since $G_{BMS}$ is not a finite-dimensional Lie group, it is by no means obvious that an exponential map from the Lie algebra to the whole group exists. Although we do not enter into the technical details of the proof, we stress that such exponential map exists in our case and it is a consequence of $\Im^+$ being generated by a whole family of complete integral curves built out of the vector field $n$.

In order to give an explicit representation of $G_{BMS}$, we fix the conformal rescaling, so to work in an already defined {\em  Bondi coordinate system}. Recall that this is constructed out of the affine parameter $u$ of the null integral curves forming $\Im^+$ together with two additional coordinates on the unit $2$-sphere. Starting form the standard ones $(\theta,\phi)$, we introduce via a stereographic projection the complex coordinates $(\zeta,\overline{\zeta})\in {\mathbb C}^2$ so that 
$\zeta= e^{i\phi}\cot(\theta/2)$. 
In this frame  the set $G_{BMS}$  is nothing but $SO(3,1)_\uparrow  \times C^\infty(\mathbb{S}^2)$, where $SO(3,1)_\uparrow$ is the proper orthocronous Lorentz group in four dimensions. A generic element $(\Lambda, f) \in SO(3,1)_\uparrow \times C^\infty(\mathbb{S}^2)$ acts on $p=(u,\zeta,\overline{\zeta})\in \Im^+$ as
\cite{sachsa}
\begin{eqnarray}
u &\to & u'\doteq K_\Lambda(\zeta,\overline{\zeta})(u + f(\zeta,\overline{\zeta}))\:,\label{u}\\
\zeta &\to & \zeta' \doteq\Lambda\zeta\doteq \frac{a_\Lambda\zeta + b_\Lambda}{c_\Lambda\zeta +d_\Lambda}\:, \:\:\:\:\:\:
\overline{\zeta} \: \to \: \overline{\zeta}' \doteq\Lambda\overline{\zeta} \doteq \frac{\overline{a_\Lambda}\overline{\zeta} + \overline{b_\Lambda}}{\overline{c_\Lambda}\overline{\zeta} +\overline{d_\Lambda}}\:,
\label{z}
\end{eqnarray}
where
\begin{eqnarray}
K_\Lambda(\zeta,\overline{\zeta}) \doteq  \frac{(1+\zeta\overline{\zeta})}{(a_\Lambda\zeta + b_\Lambda)(\overline{a_\Lambda}\overline{\zeta} + \overline{b_\Lambda}) +(c_\Lambda\zeta +d_\Lambda)(
	\overline{c_\Lambda}\overline{\zeta} +\overline{d_\Lambda})}
\label{K},
\end{eqnarray}
while
$$\Pi^{-1}(\Lambda)=\left[
\begin{array}{cc}
a_\Lambda & b_\Lambda\\
c_\Lambda & d_\Lambda 
\end{array}
\right],$$
where $a_\Lambda,b_\Lambda,c_\Lambda,d_\Lambda\in\mathbb{C}$ and $a_\Lambda d_\Lambda-b_\Lambda c_\Lambda=1$. $\Pi$ is the surjective covering homomorphism from $SL(2,\mathbb{C})$ to $SO(3,1)_\uparrow$ and thus the above matrix is unambiguously fixed up to a global sign, which plays ultimately no role.  (\ref{z}) and (\ref{K}) say that
$G_{BMS}$ has the structure of a semidirect product
$SO(3,1)_\uparrow\ltimes C^\infty(\mathbb{S}^2)$,
the elements of the Abelian subgroup $C^\infty(\mathbb{S}^2)$ being called {\bf supertranslations}.
In particular, if $\odot$ denotes the  product in $G_{BMS}$, $\circ$ the composition of functions, $\cdot$
the pointwise product of scalar functions
and $\Lambda$ acts on $(\zeta,\overline{\zeta})$ as in the right-hand side of (\ref{z}):
\begin{eqnarray}
K_{\Lambda'}(\Lambda(\zeta,\overline{\zeta})) K_\Lambda(\zeta,\overline{\zeta}) &=& K_{\Lambda' \Lambda}(\zeta,\overline{\zeta}) \label{KK}\:.\\
(\Lambda',f') \odot (\Lambda,f) &=& \left(\Lambda' \Lambda,\: f + (K_{\Lambda^{-1}} \circ \Lambda)\cdot (f'\circ \Lambda)  \right)\:.
\label{product2}
\end{eqnarray}

\begin{remark}\label{generatoremalditesta} 
	We underline that in the literature the factor $K_\Lambda$ does not always have
	the same definition. In particular, in \cite{Mc1, Mc2, Mc5, Girardello, Mc4} 
	$$K_\Lambda(\zeta,\overline{\zeta}) \doteq  \frac{(a_\Lambda\zeta + b_\Lambda)(\overline{a_\Lambda}\overline{\zeta} + \overline{b_\Lambda}) +(c_\Lambda\zeta +d_\Lambda)(
		\overline{c_\Lambda}\overline{\zeta} +\overline{d_\Lambda})}{(1+\zeta\overline{\zeta})},$$
	but in this paper we stick to the definition (\ref{K}) as in
	\cite{sachsa} accordingly adapting calculations and results from the above mentioned references.
\end{remark}
The following proposition arises from the definition of Bondi frame and from the equations above.

\begin{proposition}\label{exremark1} Let $(u,\zeta,\overline{\zeta})$ be a Bondi frame on $\Im^+$. The following holds.\\
	{\bf (a)} A global coordinate frame $(u',\zeta',\overline{\zeta}')$ on $\Im^+$ is a Bondi frame if and only if 
	\begin{eqnarray}
	u &=& u'+g(\zeta',\overline{\zeta}')\:,\label{transform1}\\
	\zeta &=& \frac{a_R\zeta' + b_R}{c_R\zeta' +d_R}\:, \:\:\:\:\:\:
	\overline{\zeta} \:\doteq\:\: \frac{\overline{a_R}\overline{\zeta}' + \overline{b_R}}{\overline{c_R}\overline{\zeta}' +\overline{d_R}}\:,\label{transform2}
	\end{eqnarray}
	for $g\in C^\infty(\mathbb{S}^2)$, while
	$R\in SO(3)$ 
	refers to the canonical inclusion $SO(3)\hookrightarrow SO(3,1)_\uparrow$
	(i.e. the canonical inclusion $SU(2)\hookrightarrow SL(2,\mathbb{C})$ for matrices of coefficients 
	$(a_\Lambda,b_\Lambda,c_\Lambda,d_\Lambda)$ in (\ref{K}).)\\ 
	{\bf (b)} The functions $K_\Lambda$ are smooth on the Riemann sphere $\mathbb{S}^2$.  Furthermore
	$K_\Lambda(\zeta,\overline{\zeta})=1$ for all $(\zeta,\overline{\zeta})$ if and only if $\Lambda \in SO(3)$.\\
	{\bf (c)} Let $(u',\zeta',\overline{\zeta}')$ be another Bondi frame as in (a).
	If $\gamma \in \mathrm{BMS}$ is represented by $(\Lambda,f)$ in $(u,\zeta,\overline{\zeta})$, the same $\gamma$ is represented
	by $(\Lambda',f')$ in $(u',\zeta',\overline{\zeta}')$ with
	\begin{equation} (\Lambda',f') = (R,g)^{-1}\odot (\Lambda,f)\odot (R,g)\:.\label{indep} \end{equation}
\end{proposition}
A last result, which will play a key role in our analysis, concerns the interplay of the isometries of an asymptotically flat spacetime and the algebra of vector fields on $\Im^+$ generating the BMS group. We strengthen the results of Proposition \ref{prop1} as follows.
\begin{proposition}\label{Prop:Extension_Isom}
For any asymptotically flat spacetime at future null infinity $(M,g)$, the following facts hold.
\begin{enumerate}
	\item[(a)] If $\xi$ is a Killing vector field of $(M, g)$, then it extends smoothly to a vector field $\widetilde{\xi}$ on the manifold $M\cup\Im^+$. The restriction to $\Im^+$ of $\widetilde{\xi}$ is uniquely determined by $\xi$, and it generates a one-parameter subgroup of $G_{BMS}$.
    \item[(b)] The linear map $\xi\to\widetilde{\xi}$ defined in (a) is injective and  if, for a fixed $\xi$ the one-parameter $G_{BMS}$-subgroup generated by $\widetilde{\xi}$ lies in $C^\infty(\mathbb{S}^2)$ then, more strictly, it must be a subgroup of 
    \begin{equation}\label{translations}
    T^4\doteq\left\{\alpha\in C^\infty(\mathbb{S}^2)\:\left|\:\alpha(\zeta,\overline{\zeta})=\sum\limits_{l=0}^1\sum\limits_{m=-l}^l\alpha_{lm}Y_{lm}(\zeta,\overline{\zeta})\:,\quad \alpha_{lm}\in\mathbb{C} \right.\right\},
    \end{equation}
    where $Y_{lm}(\zeta,\overline{\zeta})$ are the standard spherical harmonics.
\end{enumerate}
\end{proposition}
An explicit proof of item $(a)$ can be found in \cite{Geroch} while that of item $(b)$ in \cite{AX}. The symbol $T^4$ has been used on purpose since \eqref{translations} identifies a subgroup of the supertranslations which is isomorphic to the four-dimensional translation group. Two concluding comments are necessary.
\begin{enumerate}
	\item Although Proposition \ref{Prop:Extension_Isom} selects a subgroup of the supertranslations isomorphic to the ordinary four-dimensional translation group, we cannot conclude that we can extract a preferred Poincar\'e subgroup from the BMS group. As a matter of fact, starting from $\mathcal{P}=SL(2,\mathbb{C})\ltimes T^4\subseteq G_{BMS}$ and considering any element of the form $g=(\mathbb{I},Y_{lm}(\zeta,\overline{\zeta}))\in G_{BMS}$ with $l>1$, it turns out that $g^{-1}\odot\mathcal{P}\odot g$ is another subgroup of $G_{BMS}$  isomorphic to Poincar\'e group. 
	\item All our results rely crucially on $M$ being a four dimensional spacetime. The reasons are manifold, but it is important to notice that, for higher dimensions, the definition itself of an asymptotically flat spacetimes is rather subtle and it offers several difficulties -- see \cite{Hollands:2003ie} in particular. In addition, it is possible to impose stronger asymptotic conditions, so to reduce the BMS group to the Poincar\'e counterpart in higher dimensions -- see \cite{Hollands:2016oma}. Completely different is the situation for asymptotically flat, three dimensional spacetimes for which the BMS group has been studied, proving to be very different from the original four dimensional counterpart. We will not discuss this scenario in details and an interested reader should refer to \cite{Ashtekar:1996cd}.
\end{enumerate}
\subsection{Expanding universes with cosmological horizon}\label{Sec:Cosmo_spacetime}
In this section, we present a second class of backgrounds which are connected to the notion of asymptotic flatness and which will play a key role in our analysis. 
Our starting point are homogeneous and isotropic four dimensional manifolds $(M,g_{FRW})$, also known as {\bf Friedmann-Robertson-Walker} spacetimes. Following \cite{Wald}, their geometry is described by a product manifold $I\times\Sigma_\kappa$, where $I\subseteq\mathbb{R}$ is an open interval and where the metric  reads locally
\begin{equation}\label{metric}
g_{FRW} = -dt\otimes dt  +a^2(t)\left[ \frac{1}{1- \kappa r^2} dr\otimes dr+r^2d
\mathbb{S}^2(\theta,\phi)\right].
\end{equation}
Here $\kappa$ is a constant, which up to a normalization, can take the values $-1,0,1$. Depending on the choice of these values $\Sigma_\kappa$, equipped with the metric in the square bracket, is modelled respectively over one of the three simply-connected manifolds, the hyperbolic space $\mathbb{H}^3$, $\mathbb{R}^3$, $\mathbb{S}^3$,  or a non-simply-connected manifold constructed out of one of them via an identification under the action of a discrete group of isometries. 
In view of this remark and also, taking into account the present, more favoured model in cosmology, we will fix $\kappa=0$ and we will assume that $\Sigma_\kappa$ is isometric to  $\mathbb{R}^3$ with the standard Euclidean flat metric. 
The only unknown quantity in \eqref{metric} is $a(t)$, which is assumed to be a smooth and strictly positive function whose explicit form has to be determined via the Einstein equations. 
 The coordinate $t$ is referred as the {\bf proper time} of co-moving observers. By hypothesis,  $\partial_t$ 
defines the time orientation of $(M,g_{FRW})$. We can associate to \eqref{metric} two additional relevant structures: Consider a co-moving observer pictured by an  integral line $\gamma=\gamma(t)$, $t\in I$, of $\partial_t$.
\begin{enumerate}
	\item If $J^-(\gamma)$ does not cover the whole spacetime $M$, the observer cannot receive information from some points of $M$. Using the terminology of \cite{Rindler}, we call the three dimensional null hypersurfaces  $\partial J^-(\gamma)$ {\bf cosmological event horizon for $\gamma$}.
	\item If $J^+(\gamma)$ does not cover the whole spacetime $M$,
the observer cannot send information to some points of $M$. In this case we call the  three dimensional null hypersurfaces $\partial J^+(\gamma)$,  the {\bf cosmological particle horizon} for $\gamma$.
\end{enumerate}

\noindent Another representation of a Friedmann-Robertson-Walker metric  (\ref{metric}) for $\kappa=0$ is 
\begin{equation}\label{cosmo0}
g_{FRW}=a^2(\tau)\left[-d\tau^2 +dr\otimes dr+r^2 d\mathbb{S}^2(\theta,\phi)\right],
\end{equation}
where the {\bf conformal time} $\tau$ has been defined as
\begin{equation}
\label{tau}\tau(t)=  d + \int a^{-1}(t)dt\:,
\end{equation} 
 $d \in \mathbb{R}$ being any fixed constant. By construction  $\tau=\tau(t)$ is a diffeomorphism from $I$ onto a possibly unbounded interval $(\alpha,\beta)\ni \tau$.   $\partial_\tau$ is a conformal Killing vector field whose integral lines coincide  with those of $\partial_t$ up to re-parametrisation.  

In \eqref{cosmo0}, $a(\tau)$ plays the role of a conformal factor with respect to the Minkowki metric appearing in the square brackets. Since the causal structure is preserved under smooth  conformal transformations,  we can study $J^{\pm}(\gamma)$ referring to the  Minkwoski metric. 
A straightforward analysis establishes that $J^+(\gamma)$ and $J^-(\gamma)$ do not cover the whole $M$ respectively whenever $\alpha>-\infty$ and $\beta<+\infty$. In both cases the horizons $\partial J^-(\gamma)$ and $\partial J^+(\gamma)$ are null $3$-hypersurfaces 
diffeomorphic to $\mathbb{R}\times\mathbb{S}^2$, made of null geodesics of $g_{FRW}$. In some models $\alpha$ and $\beta$ can be interpreted, when they 
are finite, as the {\em big-bang} conformal time or the {\em big-crunch} conformal time  respectively. It is worth noticing that the cosmological horizons introduced above generally  depend on the fixed co-moving observer $\gamma$. Yet, it is important to bear in mind that the requirement on the finiteness of  $\alpha$ and $\beta$ are {\em sufficient} conditions for the existence of the cosmological horizons, but they are by no means necessary. 

Concerining consmological horizons, another more subtle and physically intriguing  possibility exists. Following the prototypical example of the cosmological de Sitter spacetime, it may happen that the manifold $(M, g_{FRW})$ can be realized as an open subset of a larger spacetime $(\widetilde{M},\widetilde{g})$ with physical meaning  so that  cosmological horizons may exist in $\widetilde{M}$. Furthermore they coincide with $\partial M$ which turns out to be a null $3$-surface in  $(\widetilde{M},\widetilde{g})$ similar to $\Im^\pm$ for asymptotically flat spacetimes. In these cases the horizons are independent from any choice of co-moving observer $\gamma$ in $M$. 

In the following, we shall focus on this type of cosmological horizons and our first goal is their characterization.  To this end,  let us make more precise the picture outlined above.  Starting from \eqref{cosmo0}, we rescale $g_{FRW}$ with the conformal factor $\Omega\doteq a(\tau)$ and we observe that $g\doteq \Omega^{-2} g_{FRW}$, is nothing but (a subset of) Minkowski spacetime which is asymptotically flat at null infinity so that $\Im^\pm$ can be defined. More precisely, if either $\alpha = -\infty$ or $\beta = +\infty$, 
 $(M,g)$ admits a corresponding past or future conformal completion $(\widetilde{M}, \widetilde{g})$ in accordance with  Definition \ref{Def:asymptotically_flat} where in particular
$\widetilde{g}|_M = \Omega^2 g = g_{FRW}$. This means that  $(\widetilde{M}, \widetilde{g})$ extends  $(M, g_{FRW})$ and includes one of the  null hypersurfces $\Im^\pm$ which is the boundary $\partial M$ of $M \subseteq \widetilde{M}$. 
Since $\Im^\pm \cap J^\mp(M; \widetilde{M}) = \emptyset$, this hypersurface can be viewed as a cosmological horizon in common with all observers co-moving with the metric $g_{FRW}$ in  $M$ itself.\\
The following theorem characterizes when the existence of the conformal boundary is guaranteed. We omit the long proof, which can be found in \cite{Dappiaggi:2007mx}:

\begin{theorem}\label{theorem1}
 Let $(M,g_{FRW})$ be a simply connected Friedmann-Robertson-Walker spacetime for $\kappa=0$, with
	\begin{equation*}
	M \simeq (\alpha,\beta) \times \mathbb{R}^3\:,  \quad g_{FRW}=a^2(\tau)\left[-d\tau\otimes d\tau + dr\otimes dr+r^2d
	\mathbb{S}^2(\theta,\phi)\right],\nonumber 
	\end{equation*}
	where $\tau \in (\alpha,\beta)$ and where $r,\theta,\varphi$ are the standard spherical coordinates on 
	$\mathbb{R}^3$. Suppose that there exists $\gamma\in \mathbb{R}$ with
	\begin{equation} \label{condag}
	a(\tau) = -\frac{1}{H\tau} + O\left(\frac{1}{H^2}\right)\:, \quad
	\frac{d a(\tau)}{d \tau} = \frac{1}{H\tau^2} + O\left(\frac{1}{\tau^3}\right)\:
	\end{equation}
	for either $(\alpha,\beta)\doteq (-\infty,0)$ and $H>0$, 
	or $(\alpha,\beta)\doteq (0,+\infty)$ and $H<0$. The above asymptotic values are meant to be taken as $\tau
	\to -\infty$ or $\tau \to +\infty$ respectively.
	The following  holds.
	\begin{enumerate}
		\item[(a)] The spacetime $(M,g_{FRW})$ extends smoothly 
		to a larger spacetime $(\widetilde{M},\widetilde{g})$, which is a past conformal completion of the asymptotically 
		flat spacetime at past, or future, null infinity,
		respectively,
		$(M,a^{-2}g_{FRW})$ with $\Omega=a$.
		\item[(b)] The manifold $M \cup \Im^\pm$ enjoys the following properties:
		\begin{enumerate}
			\item the vector field $\partial_\tau$ is a conformal Killing vector for $\widetilde{g}$ in $M$ 
			with conformal Killing equation  
			$$\mathcal{L}_{\partial_\tau}\widetilde{g} = -2 \partial_\tau(\ln a)  \: \widetilde{g}\:.$$
			where the right-hand side vanishes approaching $\Im^\pm$.
			\item $\partial_\tau$ tends to become tangent to 
			$\Im^\pm$ approaching it and coincides to $H^{-1} \widehat{\nabla}^b a$ thereon.
			\item The metric on $\Im^\pm$ takes the geodesically complete Bondi form up to the constant factor $H^{-2}\neq 0$:
			\begin{equation} \label{bondi2}
			\widehat{g}|_{\Im^\pm}   = H^{-2}(- du \otimes  d a -   d a \otimes du + d\mathbb{S}^2(\theta,\phi))\:,
			\end{equation}
			$u \in \mathbb{R}$ being the parameter of the integral lines of $n  \doteq  \nabla a$.
		\end{enumerate}
		\end{enumerate} 
 \end{theorem}

\begin{remark}  The statements (a),(b) hold true also if we change $g_{FRW}$ smoothly inside a region 
$M_0\simeq (\alpha,\beta)\times \Sigma_0$ for a compact $\Sigma_0\subseteq \mathbb{R}^3$. In this case
$\partial_\tau$ is a conformal Killing vector of the metric at least in $M\setminus M_0$. In the said hypotheses, one can find an open neighbourhood of $\Im^\pm$ where the above construction can be adapted.
\end{remark}
 As an example, consider the metric \eqref{metric} where 
\begin{equation}
a(t) = a_0\, e^{\alpha t} \qquad  \mbox{with $a_0, \alpha > 0$ constant.}
\end{equation}
This represents a conformally static subregion of {\em de Sitter spacetime} where the cosmological constant is $\Lambda = 3 \alpha^2$. This can be considered a realistic model of the observed universe assuming, as done nowadays, that the dark energy dominates among the various sources to gravitation in the framework of Einstein theory of gravity. In this case one can fix the integration constant 
in (\ref{tau}) so that
$\tau = -e^{-\alpha t}/(a_0\alpha)$ and thus $\tau \in (-\infty,0)$. In this case $a(\tau) = -c/\tau$, where $c=1/\alpha$, and 
thus we can use the obtained result
requiring that $(M, g_{FRW})$ admits a cosmological particle horizon in common with all the observers whose 
world lines are the integral curves of $\partial_t$, and that horizon coincides with $\Im^-$.\\
A more complicated example is a spacetime $(M,g)$ with metric
\begin{equation}
g(\tau,p) = - H^{-2}(p) dt\otimes dt +  a(t)^2 h(p)\:, \quad \mbox{for $(t,p)\in \mathbb{R} \times \Sigma \simeq M$.}
\label{static}
\end{equation}
where $a(t)$ is as in the previous examples, $\Sigma$ is diffeomorphic to $\mathbb{R}^3$ while $h = dr\otimes dr +r^2d
\mathbb{S}^2(\theta,\phi)$ and $H = -1$, outside a compact set in $S$. As in the previous examples, this spacetime $(M,g)$ is globally hyperbolic provided the Euclidean metric $h$ on $\Sigma $ is complete. This is due to $g$ being conformally equivalent to the metric whose line element reads 
\begin{equation}
g_0(\tau,p) = - H^{-2}(p) d\tau\otimes d\tau  +   h(p)\:, \quad \mbox{for $(\tau,p)\in \mathbb{R} \times \Sigma  \equiv M$.} \label{static2}
\end{equation}
Any spacetime $(M,g)$ with a static metric (\ref{static2}) is globally hyperbolic if there exist $c_1,c_2 \in \mathbb{R}$ with
$c_1 \geq (-H^{-1})(p) \geq c_2 >0$ for every $p\in \Sigma$ and if $(\Sigma,h)$ is complete \cite{Kay}.

\subsection{The Cosmological-Horizon Symmetry Group}
In view of the relation established between asymptotically flat spacetimes at null infinity and expanding FRW spacetimes with 
a suitable rate of expansion $a$, we can investigate further the geometric properties of the latter. In particular we want to establish which is the counterpart in this class of spacetimes of the BMS group which we introduced in Section \ref{Sec:BMS}. The best course of action is to consider 
the manifolds introduced in Theorem \ref{theorem1} as a specialization of the following more general class: 

\begin{definition} \label{defexp}
A globally hyperbolic spacetime $(M,g)$ equipped with a positive smooth function 
	$\Omega: M \to \mathbb{R}^+$ and a future-oriented timelike vector field $X$ on $M$, 
	will be called an {\bf expanding universe with geodesically complete cosmological particle horizon} if

	\begin{enumerate}
		
		\item {\bf Existence and causal properties of horizon}. $(M,g)$ can be embedded isometrically as the interior of a submanifold with boundary of a larger spacetime $(\widetilde{M}, \widetilde{g})$, 
		the boundary $\Im^- \doteq \partial M$  verifying $\Im^- \cap J^+(M; \widetilde{M}) = \emptyset$.

		\item {\bf $\Omega$-$\Im^-$ interplay}. $\Omega$ extends to a smooth function on $\widehat{M}$ such that $\Omega|_{\Im^-} =0$ while $d\Omega  \neq 0$ everywhere on $\Im^+$.

		\item {\bf  $X$-$\Omega$-$\widetilde{g}$-$\Im^-$ interplay}. $X$ is a conformal Killing vector for $\widetilde{g}$ in a neighbourhood of $\Im^-$ 
		in $M$, with
		\begin{equation} 
		\mathcal{L}_X(\widetilde{g})= - 2  X(\ln \Omega)\: \widetilde{g}\:, \label{confscri}
		\end{equation}
		where $X(\ln \Omega) \to 0$ approaching $\Im^-$  and where $X$ does not tend everywhere to the zero vector
		approaching $\Im^-$ .
		
		\item {\bf Bondi-form of the metric on $\Im^-$ and geodesic completeness}. $\Im^-$ is diffeomorphic to $\mathbb{R} \times \mathbb{S}^2$ and the metric $\widetilde{g}$ restricted thereon takes the Bondi form up to a possible constant factor 
		$\gamma^2> 0$, that is
		\begin{equation}
		\widetilde{g}|_\Im^- = \gamma^2 \left(-du \otimes d\Omega  -d\Omega \otimes du + d \mathbb{S}^2(\theta,\phi)\right)\:, \label{quasih}
		\end{equation}
		$d\mathbb{S}^2$ being the standard metric on the unit $2$-sphere, so that $\Im^-$ is a null $3$-submanifold, and   the  curves $\mathbb{R} \ni u \mapsto (u, \theta,\phi)$ 
		are complete null $\widetilde{g}$-geodesics.
	\end{enumerate} 
The manifold $\Im^-$ is called the {\bf cosmological (particle) horizon of $M$}. The integral parameter of $X$ is called the {\bf conformal cosmological time}.
\end{definition}	
 A similar definition can be given replacing past null infinity everywhere with $\Im^+$. As already mentioned we will not discuss this case any further.

\begin{remark}\label{remarkgeo} $\null$\\
{\bf (1)} In view of condition 3, the vector $X$ is a Killing vector of the metric $g_0\doteq \Omega^{-2} g$ in a 
neighbourhood of $\Im^-$ in $M$. In this neighbourhood (which may coincide with the whole $M$), one can
think of $\Omega$ as an expansion scale evolving with rate $X(\Omega)$ referred to the conformal 
cosmological time. \\
{\bf (2)} By standard properties of causal sets \cite{Wald}, it is possible to prove that $\Im^- \cap J^+(M; \widetilde{M}) = \emptyset$ entails $M= I^+(M;\widetilde{M})$ 
and $\Im^-=\partial M=\partial I^+(M;\widetilde{M})=\partial J^+(M;\widetilde{M})$, so that $\Im^-$ has the proper 
interpretation as a particle cosmological horizon in common for all the observers in $(M,g)$ evolving along 
the integral lines of $X$.\\
{\bf (3)} It is worth stressing that the spacetimes considered in the given definition are neither homogeneous nor isotropous in general; this is a relevant extension with respect to the family of FRW spacetimes.\\
{\bf (4)} Assuming Definition \ref{defexp}, the null geodesics in item ({\bf 4}) are the (complete) integral 
curves, $u$ is an affine parameter and
$\widetilde{\nabla}^a \Omega = -\gamma^{-2}\left( \partial_u\right)^a$ on $\Im^-$, which is totally geodesic.\\
{\bf (5)} In the following,  {\bf expanding universe with cosmological horizon} will mean
expanding universe with geodesically complete cosmological particle horizon.
\end{remark}

\noindent An important geometric property of the conformal Killing vector field $X$ in Definition \ref{defexp} is that it becomes tangent to past null infinity coinciding up a multiplicative factor with $\partial_u$. The following proposition establishes this fact and the proof can be found in \cite{Dappiaggi:2007mx}.

\begin{proposition}\label{X} If $(M,g, \Omega, X,\gamma)$ is an expanding universe with cosmological horizon, the following facts hold.\\

\noindent 	{\bf (a)}
		 $X$  extends smoothly to a unique smooth vector field $\widetilde{X}$ on $\Im^-$,  which 
		may vanish on a closed subset of $\Im^-$  with empty interior at most.
		The obtained extension of $X$ to $M\cup \Im^-$ fulfils the Killing equation on $\Im^-$ with respect to the metric $\widetilde{g}$.\\
		
		\noindent 	{\bf (b)} $\widetilde{X}=f\partial_u$, where $f$ depends only on the variables  on  $\mathbb{S}^2$, being smooth and nonnegative.
\end{proposition}

\noindent It is important to stress the role of the function $f$ in the preceding proposition, as $f=1$ on FRW spacetimes. Since \ref{defexp} encompasses a class of backgrounds larger than the cosmological ones, we can interpret a non constant function $f$ as a measure of the failure of homogeneity and isotropy.   

We can now look for a subgroup $SG_{\Im^-}$ of the isometries of $\Im^-$ 
with physical relevance. All the results that we will present are proven in \cite{Dappiaggi:2007mx}. We start from a preliminary, yet very useful proposition:

\begin{proposition} \label{togroup} 
	 If $(M,g, \Omega, X, \gamma)$ is an expanding universe with cosmological horizon and $Y$ is a Killing vector field of $(M,g)$,
	$Y$ can be extended to a smooth vector field $\widehat{Y}$ defined on $\widetilde{M}$. In addition the following facts hold true:
	
	\begin{enumerate}
	\item $\mathcal{L}_{\widehat{Y}} g =0$ on $M \cup \Im^-$;
	
	\item $\widetilde{Y}\doteq \widehat{Y}|_{\Im^-}$ is uniquely determined by $Y$ and it is tangent to $\Im^-$ if and only if 
	$g(Y,X)$ vanishes approaching $\Im^-$ from $M$.
	\end{enumerate} 
    
	If we restrict the attention to the linear space of Killing fields $Y$ on $(M,g)$ such that $g(Y,X)\to 0$ approaching $\Im^-$, the following additional facts hold true.
	
	\begin{itemize}
	\item If $\widetilde{Y}$  vanishes in some $A \subseteq \Im^-$ and $A \neq \emptyset$ is open 
	with respect to the topology of
	$\Im^-$, then $Y=0$ everywhere in $M$. Hence $\widehat{Y}$ vanishes in the whole $M \cup \Im^-$.  
	
	\item The linear map $Y \mapsto \widetilde{Y}$ is injective.
\end{itemize}
\end{proposition}	
The most relevant consequence off Proposition \ref{togroup} is that all Killing vectors $Y$ in $M$ with $g(Y,X) \to 0$ approaching $\Im^-$
extend to Killing 	vectors of $(\Im^-,h)$, $h$ being the degenerate metric on $\Im^-$ induced by $\widetilde{g}$. These Killing vectors of $(M,g)$ are represented on $\Im^-$  faithfully. In other words all the geometric symmetries of the bulk spacetimes $(M,g)$ are codified on the boundary $\Im^-$.

\begin{definition}\label{defpreservscri}
If $(M,g, \Omega, X, \gamma)$ is an expanding universe with cosmological horizon,  a Killing vector field of $(M,g)$, $Y$, is said to 
	{\bf to preserve $\Im^-$} if  $g(Y,X)\to 0$ approaching $\Im^-$. Similarly, the Killing isometries of the (local) one-parameter group generated by 
	$Y$ are said {\bf to preserve $\Im^-$}.
\end{definition}

\noindent In the rest of this section we shall consider the one-parameter group of isometries of $(\Im^-,h)$ generated by such Killing vectors $\widehat{Y}|_{\Im^-}$. These account only for part of the isometries of $(\Im^-,h)$ as one can infer considering the isometry constructed out of the coordinates $(u, \theta,\phi)\in \mathbb{R} \times \mathbb{S}^2 \equiv \Im^-$,
and of the smooth diffeomorphisms $f: \mathbb{R} \to \mathbb{R}$ generating
the transformation 
\begin{equation}
u \to f(u),\quad (\theta,\phi) \to (\theta,\phi). \label{sa}
\end{equation}
These isometries of $\Im^-$ play a role very similar to those of the elements of the Newman-Unti group in asymptotically flat spacetimes which are not part of the BMS group \cite{Newman:1962cia}. However only diffeomorphisms 
of the form $f(u) = a u + b$ with $a\neq 0$ can be isometries generated by 
the restriction $\widehat{Y}|_\Im^-$ to $\Im^-$ of extensions of Killing fields $Y$ of $(M,g)$ 
as in the proposition \ref{togroup}.  This is because those isometries 
are restrictions of isometries of the manifolds with boundary 
$(M\cup \Im^-,\: \widetilde{g}|_{M\cup \Im^-})$, and thus they {\em preserve the null} $\widetilde{g}$-{geodesics} 
in $\Im^-$.  The requirement that, for all constants 
$a,b\in \mathbb{R}$, $a\neq 0$, there exist constants $a',b'\in \mathbb{R}$, $a'\neq 0$ such that $f(au + b) = a'u + b'$
for all $u$ varying in a fixed nonempty interval $J$, is fulfilled only if $f$ is an affine transformation as said above. 
If we include also the transformations of angular coordinates, we end up studying the class $G_\Im^+$ of diffeomorphisms  of $\Im^-$ of the form
\begin{equation}
u \to u' \doteq f(u, \theta,\phi)\:, \quad (\theta,\phi) \to (\theta',\phi') \doteq g(u, \theta,\phi)
\label{diffgen}
\end{equation}
where $u \in \mathbb{R}$ and $(\theta,\phi) \in \mathbb{S}^2$ and where these transformations are isometries of the degenerate metric $h$ induced by $\widetilde{g}|_{\Im^-}$ (\ref{quasih}).
Following a lengthy analysis, discussed thoroughly in \cite{Dappiaggi:2007mx}, it is possible to classify all Killing  isometries of the degenerate metric $h$ on $\Im^-$ which are restrictions of possible Killing $\widetilde{g}$-isometries of $M\cup \Im^-$. The next definition summarizes the result: 

\begin{definition}\label{defSG} The {\bf horizon symmetry group} $SG_{\Im^-}$ is the group of all diffeomorphisms of $\mathbb{R} \times \mathbb{S}^2$,
	\begin{equation}\label{group}
	F_{(a,b,R)} : \mathbb{R} \times \mathbb{S}^2 \ni (u, \theta,\varphi) \mapsto \left(e^{a(\theta,\phi)}u + b(\theta,\phi), R(\theta,\phi)\right) \in \mathbb{R} \times \mathbb{S}^2
	\end{equation}
	with $u \in \mathbb{R}$ and $(\theta,\phi) \in \mathbb{S}^2$,	
	where $a,b \in C^\infty(\mathbb{S}^2)$ are arbitrary smooth functions and $R \in SO(3)$.\\
	The {\bf Horizon Lie algebra} $\mathfrak{g}_{\Im^-}$ is 
	the infinite-dimensional Lie algebra of smooth vector fields on $\mathbb{R} \times \mathbb{S}^2$ generated by the fields 
	$$ S_1\:, S_2\:, S_3\:,\: \beta \partial_u\:, \:
	u\alpha \partial_u\:, \quad\mbox{for all $\alpha,\beta \in C^\infty(\mathbb{S}^2)$.}$$
	$S_1,S_2,S_3$ indicate the three smooth vector fields on the unit sphere $\mathbb{S}^2$ generating rotations 
	about the orthogonal axes, respectively, $x$, $y$ and $z$. 
\end{definition}

\noindent  $SG_{\Im^-}$ depends on the geometry of $\Im^-$ but not on that of $(M,g)$. In this sense it enjoys the same properties of the BMS group, namely it is {\em universal} for the whole class of expanding spacetimes with cosmological horizon.  As a set 
$SG_{\Im^-}$ coincides with $SO(3) \times C^\infty(\mathbb{S}^2)\times 
C^\infty(\mathbb{S}^2)$. In other to unveil the group structure, consider an arbitrary $F_{a,b,R}$ and indicate it with the triple $(R,a,b)$. Per direct inspection of \eqref{group}, we see that the composition between elements in $SG_\Im^-$ can be defined as
\begin{eqnarray*}
(R, a,b) \odot (R', a',b') \doteq  \left(RR',\:\: a' + a \circ R' ,\:\: e^{a\circ R'} b' + b\circ R' \right) \\
\mbox{for all 
	$(R,a,b), (R', a',b') \in SO(3) \times C^\infty(\mathbb{S}^2)\times C^\infty(\mathbb{S}^2)$} \label{product}\:.  
\end{eqnarray*}
where $\circ$ denotes the usual composition of functions. \\

\noindent We state a few additional results aimed at characterizing the properties of $SG_{\Im^-}$. In the next proposition we emphasize instead that $\mathfrak{g}_{\Im^-}$  could be considered as the Lie algebra of $SG_{\Im^-}$, although a careful analysis of this point will not be discussed in this work. As usual, the proofs can be found in \cite{Dappiaggi:2007mx}:

\begin{proposition}\label{Gscrim} Referring to Definition \ref{defSG},
	the following facts hold.
	\begin{itemize}
		\item  Each vector field $Z \in \mathfrak{g}_{\Im^-}$ is complete and it generates a one-parameter group of diffeomorphisms 
	of $\mathbb{R} \times \mathbb{S}^2$, $\{\exp\{tZ\}\}_{t\in \mathbb{R}}$, subgroup of $SG_\Im^-$.
	
	\item For every $F \in SG_\Im^-$ there are  $Z_1,Z_2 \in \mathfrak{g}_{\Im^-}$ -- with, possibly, $Z_1=Z_2$ -- 
	such that  $F= \exp\{t_1Z_1\} \exp\{t_2Z_2\}$ for some real numbers $t_1,t_2$.
	\end{itemize}
\end{proposition}

\begin{theorem}\label{theorem2} 

Let  $(M,g, \Omega, X, \gamma)$ be an expanding universe with cosmological horizon and
	$Y$ a Killing  vector field of $(M,g)$ preserving $\Im^-$. Then
	\begin{enumerate}
		\item The unique smooth extension $\widetilde{Y}$ of $Y$ to $\Im^-$ belongs to $\mathfrak{g}_{\Im^-}$,
	
	\item $\exp\{t\widetilde{Y}\}_{t\in \mathbb{R}}$ is a subgroup of $SG_{\Im^-}$.
\end{enumerate}
\end{theorem}
As an example consider the expanding universe $M$ with cosmological horizon associated with the metric $g_{FRW}$ (\ref{cosmo0}) 
with $\kappa=1$ and $a$ as in (a) of Theorem \ref{theorem1}. In this case $X\doteq \partial_\tau$ and there are a lot of Killing vectors $Y$ of 
$(M,g_{FRW})$
satisfying $g_{FRW}(Y,X) \to 0$ approaching $\Im^-$. The main ones are those  of the surfaces at $\tau=$constant
with respect to the induced metric. They form a Lie algebra generated by $6$ independent Killing vectors $Y$ representing, respectively,
space translations and space rotations. In this case $g_{FRW}(Y,X)=0$ so that the associated  Killing vectors
$\widehat{Y}|_\Im^-$ belongs to $\mathfrak{g}_\Im^-$.\\
We state a last technical result whose proof is in \cite{Dappiaggi:2007mx}.

\begin{proposition}\label{positiveenergy}
	
	Let $(M,g, \Omega, X,\gamma)$ 
	be an expanding universe with cosmological horizon and $Y$  a smooth vector field of $(M,g)$
	which tends to the smooth field $\widetilde{Y} \in \mathfrak{g}_{\Im^-}$ pointwisely. If there exists an open set $A\subseteq \widetilde{M}$ with $A \supseteq \Im^-$ and such that $Y|_{A \cap M}$
	is timelike and future directed, then,  everywhere on $\Im^-$,
	\begin{equation}
	\widetilde{Y}(u,\theta,\phi) = f(\theta,\phi) \partial_u\:, \quad \mbox{for some $f\in C^\infty(\mathbb{S}^2)$, with $f(\theta,\phi)\geq 0$ on $\mathbb{S}^2$.}
	\end{equation}
\end{proposition}


%
%
%
\chapter[QFT in spacetimes with null surfaces]{Quantum Fields in spacetimes with null surfaces}
\label{Ch:qft-null-surfaces} 

The goal of this chapter is twofold. On the one hand we will review how the algebra of observables for a real scalar field on globally hyperbolic spacetimes is built. We will show in addition that a similar construction exists when one consider a suitable class of null manifolds of which future and past null infinity, discussed in the previous chapter, are the prototypes. On the other hand, we will highlight that, under suitable geometric hypotheses, all realized in the cases considered in spacetimes discussed in Chapter \ref{geometry}, these two kind of algebras can be closely intertwined. We will call such procedure {\em bulk-to-boundary projection}. 

\section[Algebra of observables]{Algebra of observables in globally hyperbolic spacetimes}

In this section, $C^\infty(M)$ and 
$C^\infty_0(M)$ respectively denote the {\em real} vector space of {\em real}-valued smooth functions  
and the {\em real} vector space of {\em real}-valued  smooth compactly-supported functions over $M$, where 
 $(M,g)$ is a generic globally hyperbolic spacetime of dimension $4$, although everything we will present can easily be generalized to $n>2$ dimensions. On top of $M$ we take a real smooth scalar field $\varphi\in C^\infty(M)$, which obeys the {\em Klein-Gordon equation}
\begin{equation}\label{eq:dynamics}
P\varphi=\left(\Box-m^2-\xi R\right)\varphi=0,
\end{equation} 
where $\Box=g^{\mu\nu}\nabla_\mu\nabla_\nu$ and $R$ are respectively the D'Alembert wave operator and the scalar curvature built out of the metric, $m^2\in \mathbb R$ is the squared mass parameter. Physical values are usually assumed to lie in the range $[0,+\infty)$ but many features of the theory survive when negative values of $m^2$ are considered. In addition $\xi\in \mathbb{R}$ is a coupling constant and notable values are $0$, known as {\bf minimal coupling} or $\xi=\frac{1}{6}$, known as {\bf conformal coupling} (in $4$ dimensions). We will discuss this last special case in the next sections. 

Our first goal is to review the properties of \eqref{eq:dynamics}, in particular, discussing how it is possible to assign a suitable algebra of quantum observables to such dynamical system. This is an overkilled topic which has been thoroughly analysed and presented by several authors in different contexts and frameworks. Recent reviews can be found in \cite{Benini:2013fia} and in \cite{Benini:2015bsa,Khavkine:2014mta}. to some extent our presentation will be based mainly on \cite{BGP}. To avoid to overextend this work, we will omit the proofs of many statements, but these can all be found in \cite{BGP}, unless stated otherwise.

Our starting point is the observation that the second order partial differential operator $P$ in \eqref{eq:dynamics} is {\bf normally hyperbolic}, that is its principal symbol $\sigma_P:T^*M\to C^\infty(M)$ is of metric type, $\sigma(P)(k)=-g^{\mu\nu}k_\mu k_\nu$, where $g$ has Lorentzian signature. This entails that we can associate to $P$ a distinguished pair of bi-distributions on $M$, {\it c.f.} \cite[Def. 3.4.1 \& Corol. 3.4.3]{BGP} as stated  in (a) below.  For the statement (b) see \cite{Wald}.

\begin{proposition}\label{prop:causal_propagator}
	Let $(M,g)$ be a globally hyperbolic spacetime and let $P$ be the Klein-Gordon operator as in \eqref{eq:dynamics}. The following facts hold.\\
{\bf (a)} There exist unique 
 linear operators $G^\pm: C^\infty_0(M)\to C^\infty(M)$  called 
{\bf advanced} $(-)$ and {\bf retarded $(+)$ Green operators} such that, if $\mathbb{I}$ is the identity on $C^\infty_0(M)$\footnote{In the rest of the book $P G^\pm$ and $G^\pm P$ and other similar expressions are understood as compositions of linear operators omitting, as is usual for linear operators, the composition symbol $\circ$.},
	\begin{enumerate}
		\item $P G^\pm = \mathbb{I}$  
		\item $G^\pm P|_{C^\infty_0(M)}=\mathbb{I}$,
		\item  $\textrm{supp}(G^\pm(f))\subseteq J^\pm(\textrm{supp}(f))$ for every $f\in C^\infty_0(M)$  
	\end{enumerate}
$G^\pm : C_0^\infty(M) \to C^\infty(M) \subset \mathcal{D}'(M)$ are continuous with respect to the standard distributional topologies \cite{Friedlander:2010eqa}.\\
{\bf (b)} If $\Sigma$ is a smooth spacelike Cauchy surface of $M$
  with unit normal future-pointing vector $n$, for every  $f,f^\prime \in C_0^\infty(\Sigma)$ there is a unique solution $\varphi \in C_0^\infty(M)$ of $P\varphi=0$
with  {\bf Cauchy data} $\varphi|_\Sigma=f$ and $n^\mu\partial_\mu\varphi|_\Sigma =f^\prime$. In this case the following inclusion of supports holds $\mbox{supp}(\varphi) \subseteq J^+(\mbox{supp}(f) \cup 
\mbox{supp}(f^\prime)) \cup J^-(\mbox{supp}(f) \cup 
\mbox{supp}(f^\prime))$.
\end{proposition}	
The bi-distribution $G=G^--G^+$ is called {\bf causal propagator} (also known in the literature as {\bf advanced-minus-retarded fundamental solution} or {\bf commutator} function). It satisfies the identities $P G= G P=0$ on $C_0^\infty(M)$ and it plays a relevant role in understanding the structure of the space of solutions of \eqref{eq:dynamics}. Before exploiting this fact we introduce an ancillary definition. 

\begin{definition}\label{Def:spacelike_compact}
	Let $(M,g)$ be a globally hyperbolic spacetime. We call $\varphi\in C^\infty(M)$,
	 {\bf spacelike compact} if, for every  smooth spacelike Cauchy surface $\Sigma$, both the restriction $\varphi|_\Sigma$ and the normal derivative $n^\mu \partial_\mu\varphi|_\Sigma$ have compact support. The real vector space of all such $\varphi$ is indicated with $C^\infty_{sc}(M)$.
\end{definition}

\noindent We use this definition and the existence of the causal propagator to characterize a distinguishable class of solutions of the Klein-Gordon equations, which represent the building block of the algebra of observables for such system:

\begin{proposition}\label{prop:initial_data}
	Let $(M,g)$ be a globally hyperbolic spacetime and let $P$ be the Klein-Gordon operator \eqref{eq:dynamics}. If  $\mathcal{S}(M)\doteq\{\varphi\in C^\infty_{sc}(M)\;|\;P\varphi=0\}$ the following facts hold:\\
{\bf (a)} The map $C^\infty_0(M) \ni f \mapsto G(f) \in \mathcal{S}(M)$ is a well-defined  surjective homomorphism of real vector spaces whose kernel is $P\left(C^\infty_0(M)\right)$.\\
{\bf (b)} The map
$$\frac{C^\infty_0(M)}{P\left(C^\infty_0(M)\right)}\ni[f]\mapsto G(f) \in \mathcal{S}(M)$$ is a vector space isomorphism.
\end{proposition}

\begin{proof} First of all, notice that $\mathcal{S}(M)$ is a real vector space.
(b) is a consequence of (a), so we prove (a).
	We start by establishing that the linear map $L: C^\infty_0(M) \ni f \mapsto G(f) \in \mathcal{S}(M)$ is a well-defined homomorphism of vector spaces.
Form Proposition \ref{prop:causal_propagator},  $G$ exists, $G(f)$ is smooth, satisfies $PG(f)=0$ for $f \in C_0^\infty(M)$ and the map $L: C^\infty_0(M) \ni f \mapsto G(f)$ is linear to the full vector space of smooth solutions of KG equation. We need to prove that, more strongly, $G(f) \in \mathcal{S}(M)$. From the definition of $G$ we have
$\mbox{supp}(G(f)) \subseteq J^+(\mbox{supp}(f))\cup J^-(\mbox{supp}(f))$. This fact implies that $G(f)$ is spacelike compact, i.e.,  $G(f) \in \mathcal{S}(M)$ as wanted. Indeed, if  $K \subseteq M$ is compact, using (1) Remark \ref{rem:propJ} compactness, 
(a) 1 Remark \ref{rem:notation},
one sees that 
there are a {\em finite} number of points $p_i$ such that $\cup_i J^+(p_i) \supseteq K$ so that $\cup_i J^+(p_i) \supseteq J^+(K)$ ((d)(1) Remark \ref{rem:notation}).
We can always suppose that all $p_i$ lie in the past of every fixed smooth spacelike Cauchy surface $\Sigma$, moving each $p_i$ along a past-directed  causal curve.
 From (3) Remark \ref{rem:propJ}, $\cup_i J^+(p_i) \cap \Sigma$ is compact and thus $J^+(K)\cap \Sigma$ is contained in a compact set. With a similar argument one establishes that $J^-(K)\cap \Sigma$ is included in a compact set. Taking $K = \mbox{supp}(f)$, from  $\mbox{supp}(G(f)) \subseteq J^+(\mbox{supp}(f))\cup J^-(\mbox{supp}(f))$, we have that the Cauchy data of $G(f)$ on $\Sigma$ must be compactly supported and thus $G(f) \in \mathcal{S}(M)$.\\ 
Let us prove that 
$Ker(L) = P\left(C^\infty_0(M)\right)$. Since $GP=0$ we have that $Ker(M) \subseteq P[C^\infty_0(M)]$, we want to establish the converse inclusion.
 Suppose that $G(f)=0$ for $f\in C^\infty_0(M)$. The very definition of $G$ implies $G^+(f)=G^-(f)$. Properties of $G^\pm$ immediately yield $\textrm{supp}(G(f))\subseteq J^+(\textrm{supp}(f))\cap J^-(\textrm{supp}(f))$.  Now, since $\textrm{supp}(f)$ is compact and thus,  using (1) Remark \ref{rem:propJ} compactness, 
(a) 1 Remark \ref{rem:notation}, we find that $J^+(\textrm{supp}(f))\cap J^-(\textrm{supp}(f))$ is inclosed in the compact set $\cup_{j} J^+(p_j)\cap J^-(q_j)$ for a  finite set of points $p_j,q_j$. Consequently $\textrm{supp}(G(f))$ is compact as well.
 Hence Proposition \ref{prop:causal_propagator} entails $f=Pg$ where $g = G^+(f)\in C^\infty_0(M)$ as wanted. \\ To conclude, we prove that $L$ is surjective. 
Decompose $M = \mathbb{R} \times \Sigma$ as in Proposition \ref{BS} where the smooth spacelike Cauchy surface $\Sigma$ is moved by $t\in \mathbb R$ in $M$.
Consider any $\varphi\in\mathcal{S}(M)$ and let $\chi^+\equiv\chi^+(t)$ be a smooth function such that there exists $t_0,t_1\in\mathbb{R}$ for which $\chi^+$ vanishes for all $t \leq t_0$, while it is equal to $1$ for all $t\geq t_1$. Since $\varphi$ is spacelike compact,
it has Cauchy data included in a compact $K \subseteq\Sigma_{t_0}$ and thus $\mbox{supp}(\chi^+\varphi) \subset J^+(K)$ due to (b) of Proposition \ref{prop:causal_propagator}. Since 
$P((\chi^+\varphi))= P\varphi=0$ above $\Sigma_{t_1}$, hence $\mbox{supp}(P(\chi^+\varphi)) \subset J^+(K)\cap J^-(\Sigma_{t_1})$. Exploiting an argument already used above 
relying on (3) Remark \ref{rem:propJ}, we see that $\mbox{supp}(P(\chi^+\varphi))$ is compact.
 So $f_\varphi\doteq -P(\chi^+\varphi)\in C^\infty_0(M)$ and, if $\chi^-=1-\chi^+$, then $P(\chi^-\varphi) = -P(\chi^+\varphi)\in C^\infty_0(M)$. Hence $G(f_\varphi)= G^-(P(\chi^-\varphi))-G^+(-P(\chi^+\varphi)) = (\chi^-+\chi^+) \varphi = \varphi$. \qed
\end{proof}

\noindent The construction of solutions to \eqref{eq:dynamics} discussed above is not antithetical to assigning initial data on a Cauchy surface, being actually equivalent to such procedure. It has the additional advantage of being manifestly covariant as no choice of a specific initial value surface is involved. Hence, from now on and in view of Proposition \ref{prop:initial_data}, the key object of our investigation will become $\frac{C^\infty_0(M)}{P\left(C^\infty_0(M)\right)}$. The following proposition shows that this space can be decorated with additional structures.

\begin{proposition}\label{prop:symplectic_form} Referring to Proposition \ref{prop:initial_data}, consider the real vector space
\begin{equation}\label{eq:observables}
\mathcal{E}^{obs}(M)\doteq\frac{C^\infty_0(M)}{P\left(C^\infty_0(M)\right)} \equiv \mathcal{S}(M),
\end{equation}
The following facts hold.\\

\noindent {\bf (a)} The map $\sigma_M:\mathcal{E}^{obs}(M)\times\mathcal{E}^{obs}(M)\to\mathbb{R}$
\begin{equation}\label{eq:symplectic_form}
\sigma_M([f],[f^\prime])\doteq 
 G(f, f^\prime) \doteq 
 \int\limits_{M} f(x)\,G(f^\prime)(x)d\mu_g(x) \:\quad \forall [f],[f^\prime]\in\mathcal{E}^{obs}(M)
\end{equation}
where $d\mu_g$ is the metric induced measure,
is a well-defined  symplectic form which is {\bf weakly non-degenerate}, i.e., 
$\sigma_M(\varphi,\varphi')=0$ for all $\varphi\in \mathcal{E}^{obs}(M)$ implies 
$\varphi'=0$.\\

\noindent {\bf (b)} If $\varphi_f \doteq G(f)$ and $\varphi_{f'} \doteq G(f')$ and $\Sigma$ is a spacelike smooth Cauchy surface of $M$ with future-directed normal unit covector $n$ and metric induced measure $d\Sigma_g$,
\begin{equation}\label{eq:idsigmasigma}
\int_{\Sigma} \left(\varphi_f  \nabla_n \varphi_{f'} - \varphi_{f'}  \nabla_n \varphi_f \right) d\Sigma_g = \sigma_M([f],[f'])\:. 
\end{equation}
\end{proposition}

\begin{proof}
(a)	As $P$ is formally self-adjoint and $P Gf = G Ph =0$ for $h\in C_0^\infty(M)$, the integral in (\ref{eq:symplectic_form}) 
does not depend on the choice of representatives of $[f]$ and $[f^\prime]$, so that
$\sigma_M$ is a well-defined bi-linear form on $\mathcal{E}^{obs}(M)$. 
To prove that $\sigma_M$ is antisymmetric take $f,f^\prime\in C^\infty_0(M)$ and, referring to Proposition  \ref{BS}, fix $\Sigma_{t_2}$ in the future of the compact set $\mbox{supp}(f) \cup \mbox{supp}(f^\prime)$
and $\Sigma_{t_1}$ in the past of the same set. Exploiting an argument similar to that used 
in the last part of the proof of Proposition \ref{prop:initial_data} it is easy to establish that
$J^+(\mbox{supp}(f))\cap J^-(\mbox{supp}(f^\prime))$ is included in a compact set $K \subset J^+(\Sigma_{t_1})\cap J^-(\Sigma_{t_2})$.
 Since $P$ is a formally self-adjoint, we have
	\begin{gather}
	\int\limits_{M} f(x)\,G^-(f^\prime)(x) d\mu_g =\int\limits_{M}P(G^+(f))(x)\,G^-(f^\prime)(x) d\mu_g=\nonumber \\
	=\int\limits_{M}G^+(f)(x)\,P(G^-(f^\prime))(x) d\mu_g\, =\int\limits_{M} G^+(f)(x)\,f^\prime(x) d\mu_g \label{eq:idGG}
	\end{gather}
where the intermediate passage is possible because $M \ni x \mapsto G^+(f)(x)G^-(f^\prime)(x)$ is smooth and compactly supported 
(the support lies in $K$). Since $G= G^--G^+$, (\ref{eq:symplectic_form}) implies $\sigma_M([f],[f^\prime])=-\sigma_M([f^\prime],[f])$ for all pairs $[f],[f^\prime]\in\mathcal{E}^{obs}(M)$. To prove that $\sigma_M$ is weakly non degenerate, assume that, for $[f^\prime]\in\mathcal{E}^{obs}(M)$, it holds  $\sigma_M([f],[f^\prime])=0$ for all $[f]\in\mathcal{E}^{obs}(M)$. Since the right hand side of \eqref{eq:symplectic_form} is nothing but the standard, non-degenerate pairing between $C^\infty_0(M)$ and $C^\infty(M)$, $G(f^\prime)$ must vanish. Hence $f^\prime\in\ker(G) = P\left(C^\infty_0(M)\right)$ 
which means $[f^\prime]=0$. (b) If $\Sigma$ is in the past of the supports of $f$ and $f'$, (\ref{eq:idsigmasigma}) easily arises from  (\ref{eq:idGG}) and Green's second identity for $\nabla$ and $\Box$, working in a region between $\Sigma$ and another Cauchy surface $\Sigma'$ in the future of the supports. Next the left-hand side of (\ref{eq:idsigmasigma}) can be proved to be independent from the choice of $\Sigma$ with an analogous argument based on  Green's second identity applied to $0= \varphi_f P \varphi_{f'}- \varphi_{f'} P \varphi_f =
\varphi_f \Box\varphi_{f'}- \varphi_{f'} \Box \varphi_f $ in the region between two Cauchy surfaces.\qed
\end{proof}

\noindent With these results, we have all ingredients necessary to construct the algebra of (quantum) observables for a Klein-Gordon field on a globally hyperbolic spacetime. As an intermediate step let us construct the complex unital associative  $^*$-algebra called {\bf universal tensor algebra} constructed upon the vector-space tensor product
\begin{equation}\label{eq:universal_tensor_algebra}
\mathcal{T}^{obs}(M)\doteq \mathbb{C}\otimes\bigoplus_{n=0}^\infty\left(\mathcal{E}^{obs}(M)\right)^{\otimes n},\quad(\mathcal{E}^{obs}(M))^{\otimes 0}\doteq \mathbb{R},
\end{equation}
whose elements are therefore {\em complex} linear combinations of the {\em definitively vanishing} sequences of real functions $F = (\varphi^{(0)},\varphi^{(1)}, \ldots)$ with $\varphi^{(k)} \in \left(\mathcal{E}^{obs}(M)\right)^{\otimes k}$. The associative non-commutative algebra product 
$$\mathcal{T}^{obs}(M)\times \mathcal{T}^{obs}(M)  \ni (F, G)\mapsto FG \in \mathcal{T}^{obs}(M)$$
reads for {\em real} sequences} 
 $F = (\varphi^{(0)},\varphi^{(1)}, \ldots)$ and $G = (\psi^{(0)},\psi^{(1)}, \ldots)$
$$
F G := \left( \sum_{k+h=0} \varphi^{(k)} \otimes \psi^{(h)},  \sum_{k+h=1}    \varphi^{(k)} \otimes \psi^{(h)},  \sum_{k+h=2}    \varphi^{(k)} \otimes \psi^{(h)}, \ldots \right).
$$
and this product is finally extended to the whole $\mathcal{T}^{obs}(M) \times \mathcal{T}^{obs}(M)$ by requiring complex bilinearity. The identity is $\mathbb{I} = (1,0,0, \ldots)$,  and the antlinear  involutive $^*$-operation is the unique {\em anti}linear extension to $\mathcal{T}^{obs}(M)$ of the maps, for $k=0,1,2,\ldots$ 
$$\left(\mathcal{E}^{obs}(M)\right)^{\otimes k} \ni   \varphi_1\otimes \cdots \otimes  \varphi_k \mapsto 
( \varphi_1\otimes \cdots \otimes  \varphi_k)^* \doteq  \varphi_k\otimes \cdots \otimes  \varphi_1 \in \left(\mathcal{E}^{obs}(M)\right)^{\otimes k}\:.$$ 
Observe that, while it encodes the dynamics \eqref{eq:dynamics} in the quotient defining $\mathcal{E}^{obs}(M)$, it bears no information on the so-called {\em canonical commutation relations}. In order to account also this datum, we need to construct a suitable quotient as manifest from the following definition:

\begin{definition}\label{def:algebra_of_observables}
We call $\mathcal{A}^{obs}(M)$ the {\bf algebra of observables} for a real Klein-Gordon field on a globally hyperbolic spacetime $(M,g)$, obeying the equation of motion \eqref{eq:dynamics}, the complex unital associative $^*$-algebra built as the vector-space quotient\begin{equation}\label{eq:algebrabservables}
\mathcal{A}^{obs}(M)\doteq\frac{\mathcal{T}^{obs}(M)}{\mathcal{I}^{obs}(M)},
\end{equation}
where  $\mathcal{I}^{obs}(M)$ is the two-sided $^*$-ideal of $\mathcal{T}^{obs}(M)$ (which therefore is  also a vector subspace of $\mathcal{T}^{obs}(M)$) finitely generated by the all elements of the form (assuming $\hbar=1$) $$\varphi\otimes \varphi'-\varphi'\otimes \varphi-i\sigma_M(\varphi,\varphi')\mathbb{I}\:, \quad \forall \varphi, \varphi'\in \mathcal{E}^{obs}(M)\:. $$ 
 The $^*$-operation and the algebra produt are the same as those of $\mathcal{T}^{obs}(M)$ which descend to the quotient:
$[F]^* \doteq [F^*]$ and $[F][G] \doteq [FG]$. The unit is $[\mathbb{I}]$ which we shall symply denote by $\mathbb{I}$.\\ An element $a \in \mathcal{A}^{obs}(M)$ is said to be an {\bf observable} if it is {\em self-adjoint}, {\em i.e.}, $a=a^*$.
\end{definition}
In other words two elements $F,G \in \mathcal{T}^{obs}(M)$ are equivalent 
when $F-G$ is a finite sum of finite products of elements of $\mathcal{T}^{obs}(M)$ and, in each product, at least one of the factors has the form $\varphi\otimes \varphi'-\varphi'\otimes \varphi-i\sigma_M(\varphi,\varphi')\mathbb{I}$ for some $\varphi, \varphi'\in \mathcal{E}^{obs}(M)$. The (double) equivalence classes $[[f]]$, $[[f']]$ in the quotient $\mathcal{A}^{obs}(M)$ --where  the external $[\:\cdot\:]$ refers to the quotient in (\ref{eq:algebrabservables}), the internal one to (\ref{eq:observables})-- satisfy the {\em canonical commutation rules} $$[[f]][[f']]- [[f']][[f]] =i\sigma_M([f],[f'])\mathbb{I} =iG(f,f')\mathbb{I}.$$
These rules are encapsulated in the product of the algebra $\mathcal{A}^{obs}(M)$ arising from the quotient procedure with respect to a two-sided ideal.
\begin{remark}\label{remgenconstr} It is worth stressing that all the described picture relies only on a 
real symplectic space, in our case $(\mathcal{E}^{obs}(M), \sigma_M)$, which could be replaced for a generic symplectic space $(\mathcal{S}, \sigma)$ with no intepretation in QFT. We will exploit this opportunity shortly.
\end{remark}
Explicitly referring to quantum fields, this construction of the algebra of observables has been thoroughly studied in the literature starting from the seminal paper of Dimock \cite{Dimock}.
There are well-known alternative and equivalent constructions of $\mathcal{A}^{obs}(M)$
in terms of  {\em field operators} \cite{Khavkine:2014mta}.  With the present approach the {\bf field operator} is  defined as the map
\begin{equation}\label{eq:fieldoperator}C_0^\infty(M) \ni f \mapsto \phi(f)  \doteq [[f]] \in \mathcal{A}^{obs}(M)\end{equation}
where  the external $[\:\cdot\:]$ refers to the quotient in (\ref{eq:algebrabservables}), the internal one to (\ref{eq:observables}).
\begin{proposition}\label{prop:fieldoperator} Let $\mathcal{A}^{obs}(M)$ be the algebra of observables for a real, Klein-Gordon field on a globally hyperbolic spacetime $(M,g)$, obeying \eqref{eq:dynamics}. The field operator (\ref{eq:fieldoperator}) satisfes if $f, f^\prime \in C^\infty_0(M)$ and $a,b \in \mathbb R$
\begin{itemize}
		\item {\bf $\mathbb R$-linearity}: $\phi(af + bf^\prime)= a\phi(f)+ b\phi(f^\prime)$, 
		\item {\bf self-adjointness}: $\phi(f)^* = \phi(f)$,
\item {\bf Klein-Gordon equation in a distributional sense}: $\phi(Pf)=0$,
\item {\bf Canonical Commutation Relations}: $[\phi(f), \phi(f^\prime)] = iG(f,f^\prime)\mathbb{I}$ (assuming $\hbar=1$).
	\end{itemize}
Furthermore, $\mathbb{I}$ and all $\phi(f)$ for all $f\in C^\infty_0(M)$ are {\bf generators} of $\mathcal{A}^{obs}(M)$, {\em i.e.}, every $a \in \mathcal{A}^{obs}(M)$ is a finite complex combination of finite products of those elements.
\end{proposition} 
\begin{proof} Everything immediately follows form (\ref{eq:fieldoperator}) and the definition of $\mathcal{A}^{obs}(M)$.\qed
\end{proof}
Sometimes it is desirable to work directly with fields smeared with complex functions as discussed in the introduction of the present book. The extension of $\phi(f)$ to a $\mathbb{C}$-linear map over complex valued compactly supported smooth functions is as follows: if $f$ is complex, $\phi(f) \doteq \phi(\mbox{Re}(f))+ i \phi(\mbox{Im}(f)) $. In this case, the self-adjointness property stated in the previous proposition becomes $\phi(f)^* = \phi(\overline{f})$. \\
We report here some of the properties of $\mathcal{A}^{obs}(M)$, which are more interesting for our analysis \cite{Dappiaggi:2016fwc,Khavkine:2014mta}.
 If $N \subseteq M$ is open, $\mathcal{A}^{obs}(N;M)$ denotes the unital sub $^*$-algebra of elements of $\mathcal{A}^{obs}(M)$ {\bf supported in $N$}, {\em i.e.}, linear combinations  of $\mathbb{I}$ and products of elements $\phi(f)$ with $\mbox{supp}(f) \subseteq N$.
\begin{proposition}\label{prop:propA}
	Let $\mathcal{A}^{obs}(M)$ be the algebra of observables 
as in Definition \ref{def:algebra_of_observables} 
for a real, Klein-Gordon field obeying \eqref{eq:dynamics} on a globally hyperbolic spacetime $(M,g)$. The following holds.\\
	{\bf (a)}  {\bf Causality}: Any two elements of $\mathcal{A}^{obs}(M)$ supported in two causally disjoint regions of $M$  commute.\\
	{\bf (b)} {\bf Time-slice axiom}: Let $N \subseteq M$ be a causally-convex ((3) Remark \ref{rem:notation}) open subset which is globally hyperbolic spacetime if endowed with the metric $g|_N$. If $N$ includes a Cauchy surface of $M$, then $\mathcal{A}^{obs}(M)=\mathcal{A}^{obs}(N;M)$.
\end{proposition} 
\begin{proof} (a) If $(J^+(N) \cup J^-(N)) \cap N' = \emptyset$
and $\mbox{supp}(f) \subseteq N$ while  $\mbox{supp}(f^\prime) \subseteq N'$,
	then $[\phi(f), \phi(f^\prime)]= iG(f,f^\prime) \mathbb{I}=0$ from the causal properties of $G^\pm$ established in Proposition \ref{prop:causal_propagator}. This fact extends to generic elements of $\mathcal{A}^{obs}(N;M)$ and $\mathcal{A}^{obs}(N';M)$ on account of the properties of the commutator. (b) Under the given hypotheses, one can prove that every Cauchy surface of $(N,g|_N)$ is a Cauchy surface of $M$. If $f \in C^\infty_0(M)$, the argument exploited at the end of the proof of Proposition \ref{prop:initial_data}, proves that there exists $h_f \in C^\infty_0(M)$ with support included between two Cauchy surfaces $\Sigma_{t_1}$ and $\Sigma_{t_2}$ of $N\subset M$ with $[f]=[h_f]$. Hence $\mathcal{A}^{obs}(N;M) \ni \phi(h_f)=\phi(f) \in \mathcal{A}^{obs}(M)$. Consequently $\mathcal{A}^{obs}(M) \subseteq \mathcal{A}^{obs}(N;M)$ which implies (b) because
$\mathcal{A}^{obs}(M) \supseteq \mathcal{A}^{obs}(N;M)$ by definition.\qed
\end{proof}

\begin{proposition}\label{lem:isometries}
	Let $\mathcal{A}^{obs}(M)$ be as in Proposition \ref{prop:propA} and let $\chi :M\to M$ be an isometry of $(M,g)$, {\it i.e.}, a diffeomorphism with $\chi^*g=g$. The following holds.\\
{\bf (a)} $\chi$ induces a $^*$-isomorphism $\alpha_{\chi}:\mathcal{A}^{obs}(M)\to\mathcal{A}^{obs}(M)$ which is completely specified by its action on the generators: For every $\phi(f) \in\mathcal{E}^{obs}(M)$, 
	$$\alpha_{\chi}\phi(f)\doteq \phi(\chi_*(f)),$$
	where $(\chi_*(f))(x)=f(\chi^{-1}(x))$.\\
{\bf (b)} If $\chi :M\to M$ is another isometry, we have $\alpha_{\chi} \circ \alpha_{\chi'}= \alpha_{\chi \circ \chi'}$ and $\alpha_{id}= id$.
\end{proposition}

\begin{proof}
	(a) In view of Proposition 5.2.20 in \cite{Khavkine:2014mta}, it suffices to prove that the vector space homomorphisms  $\Gamma_\chi:  [f] \mapsto [\chi_*(f)]$ and $ \Gamma_{\chi^{-1}}:  [f] \mapsto [(\chi^{-1})_*(f)]$, where $[f]
\in  C^\infty_0(M)/P(C^\infty_0(M))$, preserve the symplectic form and $\Gamma_\chi \circ \Gamma_{\chi^{-1}} = \Gamma_{\chi^{-1}} \circ \Gamma_\chi = id$. These fact are a by-product of $\chi$ and $\chi^{-1}$ being isometries while
$\sigma_M$ is isometry-invariant from (\ref{eq:symplectic_form}) since $\mu_g$ is such. (b) arises per direct inspection.\qed
\end{proof}

\section{Observables on null infinity}\label{sec:observables-on-null-infinity}

This section concerns a question, which is apparently only marginally related to the previous analysis, namely  whether it is possible to assign a $^*$-algebra of observables when, in place of a globally hyperbolic spacetime $(M,g)$, we consider a $3$-dimensional manifold $\Im\doteq \mathbb{R}\times\mathbb{S}^2$, endowed with a {\em degenerate} metric of the form $g= 0 du \otimes du + d\mathbb{S}^2$
as the restriction to $\Im^+$ of the Bondi metric (\ref{eq:metric_at_Scri}).
Here $u$ is the standard coordinate on $\mathbb{R}$  while $d\mathbb{S}^2$ is the standard metric on the unit $2$-sphere\footnote{As far as the analysis in this section is concerned, one could  more generally consider a $d$-sphere in place of $\mathbb{S}^2$ and all the results would still be valid. We avoid such degree of generality to make more manifest the connection with the preceding chapter.}. 
As the symbol $\Im$ suggests, we are considering a class of geometric structures including {\em null infinities} (Section \ref{Sec:asymptotically_flat}) or {\em null boundaries of cosmological spacetimes} (Section \ref{Sec:Cosmo_spacetime}), though here viewed as intrinsic structures on their own right. 
We are not interested in considering a dynamical evolution on $\Im$ like the one described by KG equation in $(M,g)$, but only a {\em kinematic structure} thereon. This is fully specified by a vector space of functions, whose choice is fine tuned by the specific case in hand. At this stage some definitions may seem mysterious, however a justification  will be provided in the next sections.
 
If $f \in \mathcal{S}'(\mathbb{R}\times\mathbb{S}^2)$ is a {\em Schwartz distribution} in the sense of \cite[Appendix C]{Moretti2}, $\widehat{f}  \in \mathcal{S}'(\mathbb{R}\times\mathbb{S}^2)$ henceforth stands for the distributional {\bf Fourier transform} in the $\mathbb R$-direction.  With the most natural distributional interpretation of symbols \cite[Appendix C]{Moretti2}, 
\begin{equation}\label{FTeq}
\widehat{f}(k,\theta,\phi)\doteq\int\limits_{\mathbb{R}}e^{iku}f(u,\theta,\phi) \frac{du}{\sqrt{2\pi}}\:.
\end{equation}
We define two real vector spaces of relevant {\em real}-valued functions $\psi$ on $\Im$, the first suitable for the null infinity scenario, the second for the cosmological one.
\begin{align}\label{eq:boundary_functions}
&\mathcal{S}(\Im)\doteq\{\psi\in C^\infty(\Im) \:|\:\psi, \partial_u\psi \in L^2(\Im)\},\\
&\mathcal{S}_{c}(\Im)\doteq\{\psi \in C^\infty(\Im)\cap L^\infty(\Im) \;|\;  k\widehat{\psi}\in L^\infty(\Im)\: \textrm{and}\;\partial_u\psi\:,\:\widehat{\psi}\in L^1(\Im) \},\label{eq:cosmological_boundary_functions}
\end{align}
where $L^p(\Im)$ refers to either the  measure $du\,d\mu_{\mathbb{S}^2}$ 
or $dk\,d\mu_{\mathbb{S}^2}$ ($d\mu_{\mathbb{S}^2}$ being the natural measure on $\mathbb{S}^2$), according to the variable along $\mathbb R$ of the considered function on $\Im \doteq \mathbb{R}\times\mathbb{S}^2$. Both spaces are not trivial because include $C^\infty_0(\Im)$. Furthermore $L^\infty(\Im, du\,d\mu_{\mathbb{S}^2})\subset \mathcal{S}'(\mathbb{R}\times\mathbb{S}^2)$, so that $\mathcal{S}_{c}(\Im)$ is well defined.

\begin{proposition}\label{propSSc} Referring to the real vector spaces  (\ref{eq:boundary_functions}) and (\ref{eq:cosmological_boundary_functions}), the following facts hold.\\
{\bf (a)}	The map  $\sigma_\Im:\mathcal{S}(\Im)\times\mathcal{S}(\Im)\to\mathbb{R}$, such that
	\begin{equation}\label{eq:boundary_symplectic_form}
	\sigma_\Im(\psi,\psi^\prime)=\int\limits_{\mathbb{R}\times\mathbb{S}^2} \left(\psi\partial_u\psi^\prime-\psi^\prime\partial_u\psi\right) \: du\,d\mu_{\mathbb{S}^2}\quad \forall \psi,\psi^\prime\in\mathcal{S}(\Im)
	\end{equation}
is a weakly non-degenerate simplectic form on $\mathcal{S}(\Im)$. \\
{\bf (b)} The same result holds on $\mathcal{S}_c(\Im)$ if defining $\sigma_{\Im_c}(\psi,\psi^\prime)$ as the right-hand side of (\ref{eq:boundary_symplectic_form}) for all $\psi,\psi^\prime\in\mathcal{S}_c(\Im)$.
\end{proposition}

\begin{sproof}
We prove the thesis for $\mathcal{S}(\Im)$ since the same line of reasoning can be applied to $\mathcal{S}_c(\Im)$ (See \cite{Dappiaggi:2008dk} for details). Per construction, $\mathcal{S}(\Im)$ is a linear space, while $\sigma_\Im$ is bilinear and antisymmetric if well defined. In fact, $\sigma_\Im$ is well-defined, since the integrand is a linear combination of products of two square-integrable functions. We need only to show that $\sigma_\Im$ is weakly non-degenerate. To this end let $\psi\in\mathcal{S}(\Im)$ such that $\psi\neq 0$ and $\sigma_\Im(\psi,\psi^\prime)=0$ for all $\psi^\prime\in\mathcal{S}(\Im)$. Since both $\psi$ and $\psi^\prime$ and their $u$-derivatives are smooth and square-integrable along $\Im$, we can conclude that $\lim\limits_{u\to\pm\infty}\psi^{(\prime)}=0$ (see footnote 7 in \cite{Moretti2}). Hence, integrating by parts, 
	$$\sigma_\Im(\psi,\psi^\prime)=-2\int\limits_{\mathbb{R}\times\mathbb{S}^2}\psi^\prime\partial_u\psi \: \:du\,d\mu_{\mathbb{S}^2}\:.$$
Take a sequence of $\{\psi^\prime_n\}_{n\in\mathbb{N}}\subseteq  C^\infty_0(\Im)$  such that $\lim_{n\to\infty}\psi^\prime_n=\partial_u\psi$ in the topology of $L^2(\Im)$. Since $\sigma_{\Im}(\psi,\psi^\prime_n)=\langle\psi^\prime_n,\partial_u\psi\rangle_{L^2(\Im)}=0$ for all $n$, by continuity of the scalar product $\langle,\rangle_{L^2(\Im)}$, we have $\|\partial_u\psi\|^2_{L^2(\Im)}=0$. Since $\partial_u\psi$ is continuous, it vanishes everywhere and thus $\psi$ is constant along $\mathbb R$. More strongly it vanishes as it vanishes for $u \to +\infty$.
$\null$	\qed
\end{sproof}

\noindent Since we have identified a symplectic space we can associate to $(\mathcal{S}(\Im),\sigma_\Im)$ following the same procedure which lead to Definition \ref{def:algebra_of_observables}. Hence let us define  the unital, universal tensor algebra 
on the complex vector space of definitively vanishing sequences
$$\mathcal{T}(\Im)=
{\mathbb C} \otimes \bigoplus_{n=0}^\infty\left(\mathcal{S}(\Im)\right)^{\otimes n},\quad\left(\mathcal{S}(\Im)\right)^{\otimes 0}\equiv\mathbb{R}.$$
endowed with an associative algebra product, a unit element $\mathbb I$, and an anti-linear involutive $^*$-operation constructed in exact analogy with the corresponding mathematical objects 
of $\mathcal{T}^{obs}(M)$.

\begin{definition}\label{Def:boundary_algebra}
	The {\bf boundary algebras} on $\Im$ respectively are  the unital $^*$-algebras 
	\begin{equation}\label{eq:quotients_boundary_algebras}\mathcal{A}(\Im)\doteq\frac{\mathcal{T}(\Im)}{\mathcal{I}(\Im)} \:, \quad \mathcal{A}_c(\Im)\doteq\frac{\mathcal{T}(\Im)}{\mathcal{I}_c(\Im)}\end{equation}
	$\mathcal{I}(\Im)$ and  $\mathcal{I}_c(\Im)$
	 being the two-sided $^*$-ideal of $\mathcal{T}(\Im)$, resp.,  $\mathcal{T}_c(\Im)$ generated by all elements  $\psi\otimes\psi^\prime-\psi^\prime\otimes\psi-i\sigma_{\Im}(\psi,\psi^\prime)\mathbb{I}$, resp., 
	 $\psi\otimes\psi^\prime-\psi^\prime\otimes\psi-i\sigma_{\Im_c}(\psi,\psi^\prime)\mathbb{I}$.
\end{definition}

\noindent  The classes $[\psi] \in \mathcal{A}(\Im)$ are self-adjoint generators of the unital $^*$-algebra $\mathcal{A}(\Im)$ and they play a role 
similar to the field operators $\phi(f) \in  \mathcal{A}^{obs}(M)$. 
However $\mathcal{A}(\Im)$ and $\mathcal{A}_c(\Im)$ are markedly different from the algebra of observables $\mathcal{A}^{obs}(M)$ since none of them bears dynamical information  associated with the Klein-Gordon equation for $\mathcal{A}^{obs}(M)$. A counterpart of the time-slice axiom does not exist here. This leads naturally to the question how $\mathcal{A}(\Im)$, $\mathcal{A}_c(\Im)$ and $\mathcal{A}^{obs}(M)$ are related.
This is the key question at the heart of this whole work. Before tackling this problem, we need to discuss  the interplay between $\mathcal{A}(\Im)$ or $\mathcal{A}_c(\Im)$ and the boundary symmetries represented by the BMS group $G_{BMS}$ defined in \eqref{u} and \eqref{z} on asymptotically flat spacetimes and the horizon symmetry group $SG_{\Im^-}$ introduced in Definition \eqref{group}. The next proposition addresses this issue and its proof can be either easily inferred from the previous discussions or it can be found in \cite[Th. 2.9]{Dappiaggi:2005ci} as far as the $G_{BMS}$ group is concerned or in \cite[Remark 4.1]{Dappiaggi:2007mx} for the group $SG_{\Im^-}$.

\begin{proposition}\label{prop:isomorphism_boundary_isometries}
	Consider $\Im\doteq  \mathbb{R}\times\mathbb{S}^2$ and let $\mathcal{A}(\Im)$ and $\mathcal{A}_c(\Im)$ be respectively the $^*$-algebra of Definition \ref{Def:boundary_algebra} starting from the symplectic spaces $(\mathcal{S}(\Im),\sigma_\Im)$ and $(\mathcal{S}_c(\Im),\sigma_{\Im_c})$. Then, the following holds.\\

\noindent {\bf (a)} If $g\in G_{BMS}$ and $\psi \in \mathcal{S}(\Im)$, define $(A_g\psi)(p)=K_\Lambda(g^{-1}p)^{-1}\psi(g^{-1}p)$ referring to (\ref{u})-(\ref{K}).
Then $A_g:\mathcal{S}(\Im)\to\mathcal{S}(\Im)$ is a vector space isomorphism  preserving $\sigma_\Im$.		\\

\noindent {\bf (b)} If $h\in SG_{\Im^-}$ and $\psi \in \mathcal{S}_c(\Im)$, define  $(\widetilde{A}_h\psi)(p)\doteq\psi(F_{h^{-1}}p)$ referring to \eqref{group}.
Then $\widetilde{A}_h:\mathcal{S}_c(\Im)\to\mathcal{S}_c(\Im)$  is a vector space isomorphism preserving $\sigma_{\Im_c}$.\\

\noindent{\bf (c)} If $g\in G_{BMS}$, there is a unique $^*$-automorphism $\alpha_g : \mathcal{A}(\Im) \to \mathcal{A}(\Im)$
satisfying $\alpha_{g}[\psi]=  [A_g \psi]$, for all $\psi \in \mathcal{S}(\Im)$. It holds
$\alpha_g\circ \alpha_h = \alpha_{g\odot h}$ if $h \in G_{BMS}$ and $\alpha_{id} = id$.\\

\noindent{\bf (d)} If  $h\in SG_{\Im^-}$, there is a unique $^*$-automorphism
 $\widetilde{\alpha}_g : \mathcal{A}_c(\Im) \to \mathcal{A}_c(\Im)$
satisfying 
$\widetilde{\alpha}_{h}[\psi] = [\widetilde{A}_h \psi]$ for all $\psi \in \mathcal{S}_c(\Im)$. It holds
$\widetilde{\alpha}_h\circ \widetilde{\alpha}_k = \widetilde{\alpha}_{h\odot k}$ if $k \in SG_{\Im^-}$ and $\widetilde{\alpha}_{id} = id$. 
\end{proposition}

\begin{remark}\label{remarkIDIM} The identification of $\Im^+$ and $\Im\doteq {\mathbb R}\times {\mathbb S}^2$ for {\em asymptotocally flat spacetimes} depends on the choice of a Bondi coordinate frame. Different choices produce different identifications. However, different Bondi coordinate frames are related by means of a transformation of the BMS group (Proposition \ref{exremark1}). Even changing the metrical structure of $\Im^+$ by means of an (always physically admitted)  {\em gauge transformation} (\ref{gauge}), the Bondi frames of the new metrical structure are related to the Bondi frames of the initial metrical structure just because gauge transformations are equivalent to BMS transformations (Definition \ref{defBMS}).
 A similar  picture arises regarding the identification of $\Im^-$ and $\Im\doteq {\mathbb R}\times {\mathbb S}^2$ for {\em cosmological spacetimes} where different Bondi frames compatible with the defintion of  expanding universe with cosmological horizon are connected by the subgroup of $SG_{\Im^-}$ of the transformations (\ref{group}) with vanishing $a$.
The definitions of $\mathcal{S}(\Im)$, $\mathcal{S}_c(\Im)$ are invariant under the action of the BMS group and 
$SG_{\Im^-}$ respectively, and the elements of these groups also preserve the associated symplectic forms $\sigma_{\Im}$ and  $\sigma_{\Im_c}$. Therefore, in view of Proposition \ref{prop:isomorphism_boundary_isometries}, the various definitions of 
$\mathcal{A}(\Im)$ and  $\mathcal{A}_c(\Im)$ respectively, based on different choices of Bondi frames 
are isomorphic. 
\end{remark}

\section{The bulk-to-boundary algebra injection}\label{Sec:bulk-to-boundary}

The goal of this section is to relate the algebra $\mathcal{A}^{obs}(M)$ with either $\mathcal{A}(\Im)$ or $\mathcal{A}_c(\Im)$ under suitable geometric assumptions on $(M,g)$. Heuristically, our approach consists of considering $4$-dimensional globally hyperbolic spacetimes $(M,g)$ which can be read, up to a conformal transformation, as an open submanifold of a larger globally hyperbolic spacetime. In addition $(M,g)$ must possess a (conformal) boundary identified with $\Im$. Since $\mathcal{A}^{obs}(M)$ is generated by field operators $\phi(f)$  labelled by functions $f \in C_0^\infty(M)$, we can construct via the causal propagator a unique smooth wave function $\psi = G(f)$ which can be restricted to $\Im$, identifying thereon a generator of $\mathcal{A}(\Im)$. This identification will uniquely extend  to unital $^*$-algebra embedding of the algebra of quantum observables $\mathcal{A}^{obs}(M)$ into the boundary algebras $\mathcal{A}(\Im)$ or $\mathcal{A}_c(\Im)$.

We explain this procedure for asymptotically flat spacetimes and cosmological spacetimes.  We finally  consider the case of Schwarzschild black hole with a separated discussion.

\subsection{Asymptotically Flat Spacetimes}
Let us consider $(M,g)$ to be a $4$-dimensional globally-hyperbolic spacetime which is {\em asymptotically flat at future null infinity with future time infinity} as per item (b) of Remark \ref{Rem:alternative_asymptotically_flat}. 
In addition and for later convenience we assume also that there exists an open subset $V$ of the unphysical spacetime $(\widetilde{M},\widetilde{g})$ such that $\overline{J^+(\Im^-;\widetilde{M})\cap M}\subseteq V$ and $(V,\widetilde{g}|_V)$ is  globally hyperbolic (notice that $i^+$ may not belong to $V$). 
On top of $(M,g)$ we consider a real scalar field $\varphi:M\to\mathbb{R}$, whose dynamics is ruled by the {\em conformally coupled wave equation}
\begin{equation}\label{eq:conformally_coupled}
P_0\varphi=0\:, \quad \mbox{where  $\quad P_0 \doteq \Box-\frac{R}{6}$\:.}
\end{equation}
Above, $\Box$ is the D'Alembert wave operator built out of $g$, while $R$ is the associated scalar curvature. This is nothing but a special instance of the Klein-Gordon equation considered in \eqref{eq:dynamics}, obtained by setting $m=0$ and $\xi=\frac{1}{6}$. The reason for such special choice of equation of motion can be found in the following proposition whose proof is direct (see, {\it e.g.}, \cite[App. D]{Wald}).
\begin{proposition}\label{prop:conformal_rescaling}
	Let $(M,g)$ be a $4$-dimensional spacetime and let $\Omega$ be a strictly positive scalar smooth function on $M$. Then, 
\begin{equation}\label{eq:addPP}
P_0=\Omega^3\widetilde{P}_0 \Omega^{-1}\:, \quad \mbox{where $\quad\widetilde{P}_0 \doteq \widetilde{\Box}-\frac{\widetilde{R}}{6}$.}
\end{equation}
Above, $\Omega^{-1}$ and $\Omega^3$ act as multiplicative operators, while
  $\widetilde{\Box}$ is the D'Alembert wave operator for the metric $\widetilde{g}\doteq\Omega^2 g$ with scalar curvature  $\widetilde{R}$.
\end{proposition}
We stress that such a behaviour is distinctive only of the conformally coupled wave equation, since adding a different coupling to scalar curvature or a mass term would alter drastically the form of the equation of motion under a conformal rescaling of the metric. In particular, the mass term $m^2$ would be mapped to $\frac{m^2}{\Omega^2}$. In the case of an asymptotically flat spacetime the conformal factor $\Omega$ vanishes on $\Im^+$  giving rise to pathologies. For this reason, we consider only \eqref{eq:conformally_coupled}.

 Proposition \ref{prop:causal_propagator} helps us translating Proposition \ref{prop:conformal_rescaling} into the language of Green operators. This problem was thoroughly investigated in \cite{Pinamonti:2008cx} but we report here just an elementary result, adapted to our framework. Observe that, for $\Omega >0$ smooth over $M$,  $(M,g)$ is globally hyperbolic if and only if $(M,\Omega^2g)$ is globally hyperbolic, since causal structures and, thus, Cauchy surfaces are preserved by smooth strictly-positive conformal transformations.

\begin{proposition}\label{prop:propagators_under_conformal}
	Let $(M,g)$ be a $4$-dimensional globally-hyperbolic spacetime, $\Omega>0$ a smooth scalar function on $M$ and $\widetilde{g}\doteq \Omega^2 g$. If $G^\pm_0: C^\infty_0(M)\to C^\infty(M)$ are the Green operators for the conformally coupled wave equation \eqref{eq:conformally_coupled} in $(M,g)$,  the corresponding Green operators for the same equation in $(M, \widetilde{g})$ satisfy 
	\begin{equation}\label{eq:propagators_rescaling}
		\widetilde{G}^\pm_0 = \Omega^{-1} G^\pm_0 \Omega^3\:, \quad \widetilde{G}_0 = \Omega^{-1} G_0 \Omega^3\:.
	\end{equation}
\end{proposition}	

\begin{proof}
	As $(M, \widetilde{g})$ is globally hyperbolic, Proposition \eqref{prop:causal_propagator} guarantees that the advanced and retarded fundamental solutions $\widetilde{G}^\pm_0$ of $\widetilde{P}_0$ exist and are uniquely determined by properties (1)-(3). Since the conformal structure is preserved by a conformal transformation and supp$(\Omega^3 f)=\mbox{supp}(f)$, the right-hand side in \eqref{eq:propagators_rescaling} satisfies the desired support properties (3) which are inherited from those of $G^\pm_0$. Hence, to satisfy (1) and (2) we need to prove that the right-hand side in \eqref{eq:propagators_rescaling} are left and right inverses of the equation of motion on $C_0^\infty(M)$. Indeed, (\ref{eq:addPP}) implies
	$\widetilde{P}_0(\Omega^{-1} G^\pm_0 \Omega^3)=\Omega^{-3}P_0 G^\pm_0 \Omega^3=\mathbb{I}$. An identical calculation shows also that $(\Omega^{-1} G^\pm_0 \Omega^3)\widetilde{P}_0|_{C_0^\infty(M)}=\mathbb{I}$, which concludes the proof establishing the first identity in (\ref{eq:propagators_rescaling}) from the uniqueness property.
The second one follows per definition of $G$.\qed
\end{proof}

\noindent The found results can be readily applied to the case in hand, since $M$ is assumed to be a globally hyperbolic spacetime both if endowed with $g$ or with $\Omega^2 g$. More importantly we have also assumed that there exists a second, globally hyperbolic spacetime $V \subseteq \widetilde{M}$ equipped with the metric $\widetilde{g}|_V$ which contains $M \cup \Im^+$. Without loss of generality we assume $V= \widetilde{M}$ so that $(\widetilde{M}, \widetilde{g})$ is globally hyperbolic as well and thus $\widetilde{G}^{(\pm)}_0$ exist and they associate smooth, compactly supported functions to solutions of $\widetilde{P}_0\phi=0$ {\em defined in the whole $\widetilde{M}$, though (\ref{eq:propagators_rescaling}) are only valid in $M$}.
We remind the reader that $\widetilde{g}|_M = \Omega^2 g$ where $\Omega>0$ is smooth and smoothly vanishes exactly on $\Im^+$.
Focus on the map (notice that supp$(f) \subsetneq M$, so that $\Omega^{-3} f\in C^\infty_0(M)$)
\begin{equation}\label{eq:bulk-to-boundary}
\Gamma:\mathcal{E}^{obs}_0(M)\ni [f]\mapsto\left. \widetilde{G}_0(\Omega^{-3} f)\right|_{\Im^+} \in C^\infty(\Im^+)\:,
\end{equation}
with $\mathcal{E}^{obs}_0(M)$ as in (\ref{eq:observables}) for $P$ specialised to $P_0$.
This map is well defined because any other $f' \in [f]$ yields $\widetilde{G}_{0}(\Omega^{-3} f')(x) = \widetilde{G}_0(\Omega^{-3} f)(x)$ when $x \in \Im^+$. As a matter of facts, (a) Proposition \ref{prop:initial_data} implies $f-f'= P_0 h$ for some $h\in C_0^\infty(M)$ so that, if $x\in M$, we have $\widetilde{G}_{0}(\Omega^{-3} (f-f')) (x)= \Omega^{-1}(x)(G_0 P_0 h)(x)= 0$
where we exploited (\ref{eq:propagators_rescaling}). By hypothesis $\widetilde{G}_0(\Omega^{-3} f)(x)$
is smooth for $x \in \Im^+$ and thus the result smoothly extends to $x \in \Im^+$. The map $\Gamma$ is evidently linear, but it also enjoys a crucial property. From a generalized version of (b) Proposition \ref{prop:symplectic_form},  using $\Sigma \doteq \Im^+\cup \{i^+\}$ as a {\em limit case} of a  Cauchy surface of $(M,g)$ in the fully extended $\widetilde{M}$ (exactly {\em here} the existence of $i^+ \in \widetilde{M}$ is  relevant, 
see \cite[Th. 4.1]{Moretti}),
one finds
\begin{equation}\label{eq:symplectomorphism}\sigma_0(\varphi,\varphi') = \sigma_{\Im}(\Gamma(\varphi), \Gamma(\varphi')) \quad \forall \varphi, \varphi' \in \mathcal{E}^{obs}(M)\:.\end{equation}
Here $\sigma_0$ is the symplectic from in \eqref{eq:symplectic_form} with $G_0$ playing the role of $G$ while
 $ \sigma_{\Im}$ is the one form in (\ref{eq:boundary_symplectic_form}) with $\Im \equiv \Im^+$ when identifying $\Im^+$ with $\Im \doteq \mathbb{R} \times \mathbb{S}^2$ by means of a Bondi coordinate frame (see \cite[Th. 4.1]{Moretti}). This identity also implies that $\Gamma(\mathcal{E}^{obs}_0(M)) \subseteq \mathcal{S}(\Im)$ defined in (\ref{eq:boundary_functions}). Eventually, since $\sigma_0$ is weakly non-degenerate, the linear map $\Gamma$ must be injectve: From (\ref{eq:symplectomorphism}), $\Gamma(\varphi)=0$ implies $\sigma_0(\varphi,\varphi')=0$ for every $\varphi' \in \mathcal{E}^{obs}(M)$ and thus $\varphi=0$. Putting all together we have the following result whose detailed proof is presented in \cite{Moretti}.
%

\begin{theorem}\label{Th:bulk_to_boundary}
	Let $(M,g)$ be a $4$-dimensional globally hyperbolic spacetime which is asymptotically flat at future null infinity with future time infinity 
((b) Remark \ref{Rem:alternative_asymptotically_flat})
and let $(\widetilde{M},\widetilde{g})$ be the associated unphysical spacetime. Assume that there exists a globally hyperbolic open subset  $V \subseteq \widetilde{M}$ such that $\overline{J^-(\Im^+,\widetilde{M})\cap M}\subseteq V$. Fixing the conformal factor $\Omega$ so that the metric on $\Im^+$  takes the Bondi form \eqref{eq:metric_at_Scri}, the following facts hold.\\

\noindent {\bf (a)}
A well-defined injective vector space homomorphism 
	$\Gamma:\mathcal{E}^{obs}_0(M)\to\mathcal{S}(\Im)$ exists defined by
 \eqref{eq:bulk-to-boundary},
	where $\mathcal{S}(\Im)$ is the space of functions \eqref{eq:boundary_functions} with the role of the null manifold $\Im$ played by future null infinity $\Im^+$.\\

\noindent {\bf (b)} $\Gamma$ preserves the natural symplectic forms of $\mathcal{E}^{obs}_0(M)$ and $\mathcal{S}(\Im)$ as in (\ref{eq:symplectomorphism})  
\end{theorem}

\noindent Recall that, from Proposition \ref{prop:fieldoperator}, any element $a \in \mathcal{A}^{obs}_0(M)$ of the algebra of observables in the bulk (Definition \ref{def:algebra_of_observables} where $P$ is specialised for $P_0$) has the
form 
\begin{equation}\label{eq:deca} a = c {\mathbb I} + \sum_{N=1}^\infty
\sum_{k_1, \ldots, k_N=1}^\infty c^{(N)}_{k_1k_2\cdots k_N}\phi(f^{(N)}_{k_1})\phi(f^{(N)}_{k_2})\cdots \phi(f^{(N)}_{k_N})
\end{equation}
for $c^{(N)}_{k_1k_2\cdots k_N}\in \mathbb C$ and $f^{(N)}_{k_j} \in C_0^\infty(M)$ depending on $a$ (but not uniquely fixed by it), such that only a {\em finite} number of them do not vanish.
At the same time, any element of the boundary algebra  $\mathcal{A}(\Im)$ (Definition \ref{Def:boundary_algebra} with $\Im$ identified to $\Im^+$) is similarly generated by elements $[\psi] \in \mathcal{S}(\Im)$.
Proposition 5.2.20 in \cite{Khavkine:2014mta} and Theorem \ref{Th:bulk_to_boundary} guarantee that the map associating the right-hand side of (\ref{eq:deca}) to (the external brackets in the right hand side referring to the quotient \ref{eq:quotients_boundary_algebras}) 
\begin{equation}\label{eq:deca2} c {\mathbb I} +\sum_{N=1}^\infty\sum_{k_1,\ldots, k_N=1}^\infty c^{(N)}_{k_1k_2\cdots k_N}[\Gamma_\Im ([f^{(N)}_{k_1}])] [\Gamma_\Im ([f^{(N)}_{k_2}])]\cdots[\Gamma_\Im([f^{(N)}_{k_N}])] \in \mathcal{A}(\Im^+)
\end{equation}
is an injective unital $^*$-algebra homomorphism
\begin{equation}\label{eq:bulk_to_boundary_homomorphism}
\iota:\mathcal{A}^{obs}_0(M)\to\mathcal{A}(\Im),
\end{equation}
which is completely specified by the action on the generators, namely $$\iota(\phi(f))=[\Gamma([f])] \quad \forall f\in C^\infty_0(M)\:.$$

\noindent If we consider now a bulk spacetime $(M,g)$ fulfilling the hypotheses of Theorem \ref{Th:bulk_to_boundary}, whose metric $g$ admits non trivial isometries, we know from Proposition \ref{lem:isometries} that there exists a corresponding $^*$-automorphism of $\mathcal{A}^{obs}_0(M)$. At the same time, to every element of the $G_{BMS}$ group, one can associate via Proposition \ref{prop:isomorphism_boundary_isometries} a $^*$-automorphism of $\mathcal{A}(\Im)$. Hence, in view of Theorem \ref{Th:bulk_to_boundary} and of Proposition \ref{Prop:Extension_Isom}, it is natural to wonder whether these $^*$-automorphisms are intertwined. The following proposition, ((a) was established in \cite[Prop. 3.4]{Moretti2}, (b) is proved below), answers to this question.

\begin{proposition}\label{prop:isometries_bulk_to_boundary_homomorphism}
	Let $(M,g)$ be a spacetime fulfilling the hypotheses of Theorem \ref{Th:bulk_to_boundary} and let $\xi$ be a complete Killing field generating a one-parameter group of isometries $\{\chi^\xi_t\}_{t \in \mathbb R}$. Let $\widetilde{\xi}$ be the unique extension of $\xi$ to $\Im^+$ as per Proposition \ref{Prop:Extension_Isom} generating  the one-parameter subgroup $\{\widetilde{\chi}^{\widetilde{\xi}}_t\}_{t \in \mathbb R}\subset G_{BMS}$. Then the following statements hold true:\\

\noindent {\bf (a)} 	Referring to the standard pull-back action of automorphisms ($\chi_*(h)\doteq h \circ \chi^{-1}$) 
	\begin{equation}\label{eq:interplay_isometries}
\widetilde{\chi}^{\widetilde{\xi}}_{t*}\circ \Gamma =\Gamma\circ \chi^\xi_{t*}\quad \forall t \in \mathbb R\:.
\end{equation}

\noindent {\bf (b)} 
Referring to the unital $^*$-algebra automorphism  $\alpha_{\chi^\xi_t}: \mathcal{A}_0^{obs}(M)\to\mathcal{A}_0^{obs}(M)$ 
defined in Proposition \ref{lem:isometries}, the $^*$-algebra automorphism  $\alpha_{\widetilde{\chi}^{\widetilde{\xi}}_t}:\mathcal{A}(\Im)\to\mathcal{A}(\Im)$ defined in Proposition \ref{prop:isomorphism_boundary_isometries},
and  to
the unital $^*$-algebra embedding 
$\iota:\mathcal{A}^{obs}_0(M)\to\mathcal{A}(\Im)$ as in (\ref{eq:bulk_to_boundary_homomorphism}),  (\ref{eq:interplay_isometries})
extends to
\begin{equation}\label{eq:interplay_algebras}
\iota\circ \alpha_{\chi^\xi_t}=\alpha_{\widetilde{\chi}^{\widetilde{\xi}}_t}\circ\iota, \quad \forall t \in {\mathbb R}\:.
\end{equation}
\end{proposition}

\noindent {\em Proof of (b)}.
Consider a generic element $a\in \mathcal{A}^{obs}_0(M)$ decomposed as in (\ref{eq:deca}).
$\iota_{\Im}:\mathcal{A}^{obs}_0(M)\to\mathcal{A}(\Im)$ 
is defined as the map transforming the said $a$ into the elment of $\mathcal{A}(\Im^+)$ with the form (\ref{eq:deca2}). Using this representation,
the fact that $\alpha_{\chi^\xi_t}$ and $\widetilde{\chi}^{\widetilde{\xi}}_t$ are 
$^*$-automorphisms and therefore preserve both the algebra products and their linear combinations and taking (\ref{eq:interplay_isometries}) into account,
 per direct inspection one obtains
$$\iota\left(\alpha_{\chi^\xi_t}(a)\right)=
\alpha_{\widetilde{\chi}^{\widetilde{\xi}}_t}\left(
\iota(a)\right)\:.$$
Arbitrariness of $a\in \mathcal{A}^{obs}_0(M)$ implies (\ref{eq:interplay_algebras}).

\subsection{Cosmological Spacetimes}

Let us consider $(M,g_{FRW})$, a four-dimensional, simply connected Friedamann-Robertson-Walker spacetime with flat spatial sections, fulfilling the hypotheses of Theorem \ref{theorem1}. This is globally hyperbolic and it can be extended to a larger, globally hyperbolic spacetime $(\widetilde{M},\widetilde{g})$ so that $\partial M\subset\widetilde{M}$ is the cosmological horizon of $M$ as per Definition \ref{defexp}. On top of $(M,g_{FRW})$ we can consider a generic Klein Gordon field $\varphi:M\to\mathbb{R}$ fulfilling \eqref{eq:dynamics}.
Following Proposition \ref{prop:initial_data} we build $\mathcal{S}(M)$, the space of smooth and spacelike-compact solutions of the Klein-Gordon equation $\varphi = G(f)$ for  $f\in C_0^\infty(M)$ and, due to Proposition \ref{prop:symplectic_form}, this space is isomorphic as symplectic space to the quotient (\ref{eq:observables}) $\mathcal{E}^{obs}(M)$. Here, $G=G^--G^+$ is the causal propagator associated to   the Klein-Gordon operator $P=\Box-m^2-\xi R(\tau)$
 built out of $g_{FRW}$. In turn, via Definition \ref{def:algebra_of_observables}, we associate to such dynamical system its algebra of observables $\mathcal{A}^{obs}(M)$. 

As for the case of an asymptotically flat spacetime, our next step consists of finding a way to relate $\mathcal{A}^{obs}(M)$ to an algebra living on the cosmological horizon, in this case $\mathcal{A}_c(\Im)$ as in Definition \ref{Def:boundary_algebra} with $\Im\doteq \mathbb{R}\times \mathbb{S}^2$ identified to $\Im^-$. Since $(\widetilde{M},\widetilde{g})$ is globally hyperbolic, we can construct the causal propagator of the Klein-Gordon equation in that larger spacetime which coincides with $G$ on $C_0^\infty(M)$ if restricting the functions in its image to $M$. For this reason and with a slight abuse of notation, we will use the symbol $G$ irrespectively from the underlying manifold being $M$ or $\widetilde{M}$. In particular, since the cosmological horizon is a submanifold of $\widetilde{M}$, the following real-linear map is meaningful
\begin{equation}\label{eq:cosmological_bulk_to_boundary}
\Gamma_c:\mathcal{E}^{obs}(M) \ni [f]\mapsto -H^{-1}G(f)|_{\Im^-} \in C^\infty(\Im^-)\:,
\end{equation}
 where $|_{\Im^-}$ stands for the restriction to $\Im^-$
while $H$ is the constant appearing in \eqref{condag}.

To investigate further the properties of the map $\Gamma_c$, a preliminary step  consists of a better understanding of the structural properties of the functions $\varphi = G(f)$. Since the metric has the form \eqref{cosmo0}, we know that, regardless of the choice of the scale factor $a(\tau)$, the isometry group of $(M,g)$ will include the three-dimensional Euclidean group $E(3)=SO(3)\ltimes\mathbb{R}^3$. Hence, every $\varphi\in\mathcal{S}(M)$ can be realized in terms of Fourier transform along the spatial directions,
\begin{equation}\label{eq:Fourier_expansion_solution}
\varphi(\tau,\vec{x})=\int\limits_{\mathbb{R}^3}\frac{d^3 k}{(2\pi)^{\frac{3}{2}}}\,\widetilde{\chi}(\vec{k},\tau)e^{i\vec{k}\cdot\vec{x}}+\textrm{c.c.},
\end{equation}
where c.c. stands for the complex conjugate since the field is real. In the last formula we switched from the spherical coordinates of \eqref{cosmo0} to the standard Euclidean coordinates $\vec{x}=(x,y,z)$ over $\mathbb{R}^3$. $\vec{k}=(k_x,k_y,k_z)$ and $\cdot$ is the three-dimensional Euclidean scalar product. By imposing \eqref{eq:dynamics}, it turns out that the functions $\widetilde{\chi}(\tau, \cdot)$ belong the Schwartz 
space over $\mathbb{R}^3$ at every $\tau$ because $\varphi(\tau, \cdot) \in C_0^\infty(\mathbb{R}^3)$ and they satisfy the following ODE corresponding to Klein-Gordon equation,
\begin{equation}\label{eq:ODE_modes}
\frac{d^2\chi_{\vec{k}}(\tau)}{d\tau^2}+\left(k^2+V(\tau)\right)\chi_{\vec{k}}(\tau)=0,\quad \mbox{ where }\quad \chi_{\vec{k}}(\tau)\doteq a(\tau)\widetilde{\chi}(\vec{k},\tau)\:.
\end{equation}
Above, $k^2\doteq |\vec{k}|^2$, while $V(\tau)\doteq a^2(\tau)[m^2+(\xi-\frac{1}{6})R(\tau)]$. The modes $\chi_{\vec{k}}(\tau)$ are chosen so to be normalized in such a way that 
$$W[\chi_{\vec{k}}(\tau),\overline{\chi_{\vec{k}}(\tau)}]=\frac{d\chi_{\vec{k}}(\tau)}{d\tau}\overline{\chi_{\vec{k}}(\tau)}-\chi_{\vec{k}}(\tau)\overline{\frac{d\chi_{\vec{k}}(\tau)}{d\tau}}=-i.$$
	The construction of the solutions of \eqref{eq:ODE_modes} can be found in \cite{Dappiaggi:2007mx,Dappiaggi:2008dk}. The use of (\ref{eq:ODE_modes}) turns out to be very suitable for understanding the asymptotic behaviour of \eqref{eq:Fourier_expansion_solution} at the cosmological horizon $\Im^-$. The rationale is the following one. Since asymptotically the metric \eqref{cosmo0} tends to the one of de Sitter spacetime, as per Definition \ref{defexp}, we expect that also the solutions of \eqref{eq:Fourier_expansion_solution} should mimic at $\Im^-$ the behaviour of those built on the cosmological de Sitter spacetime, obtained setting $a(\tau)=-\frac{1}{H\tau}$. 
	 In de Sitter space, the solutions of \eqref{eq:ODE_modes}
 are known \cite{SS} 
	\begin{equation}\label{eq:dS_modes}
	\chi^{dS}_{\vec{k}}(\tau)=\frac{\sqrt{-\pi\tau}}{2}e^{-\frac{i\pi\nu}{2}}\overline{H^{(2)}_\nu(-k\tau)},
	\end{equation}
	where $k^2=|\vec{k}|^2$, while $H^{(2)}$ stands for the Hankel function of second type. The parameter $\nu$ is defined as
	\begin{equation}\label{eq:nu}
	\nu=\sqrt{\frac{9}{4}-\left(\frac{m^2}{H^2}+12\xi\right)},
	\end{equation}
	where we assume that $\mbox{Re}(\nu)\geq 0$ and $\mbox{Im}(\nu)\geq 0$. The strategy for a more general spacetime consists of considering an arbitrary scale factor $a(\tau)$ though consistent with the hypotheses of Definition \ref{defexp} and looking for solutions of \eqref{eq:ODE_modes} which can be written as a Duhamel/perturbative series starting from \eqref{eq:dS_modes}. Although, each finite order in the series is well-defined, the problem of the convergence of the perturbative series remains. In \cite[Th. 4.5]{Dappiaggi:2007mx} and, subsequently, in \cite{Dappiaggi:2008dk}, it has been proven uniform convergence whenever $\mbox{Re}(\nu)<\frac{1}{2}$ and $V(\tau)-V_{dS}(\tau)=O(\tau^{-3})$ or $\mbox{Re}(\nu)<\frac{3}{2}$ and $V(\tau)-V_{dS}(\tau)=O(\tau^{-5})$. Here $V_{ds}(\tau)$ is the potential in \eqref{eq:ODE_modes} setting the metric to the de Sitter one.

\noindent Bearing in mind the above short digression on the explicit construction of the elements of $\mathcal{S}(M)$, we can investigate further the properties of the projection map $\Gamma_c$. We report here the results proven in \cite[Prop. 2.1 \& Th. 2.1]{Dappiaggi:2008dk}:

\begin{theorem}\label{th:cosmological_bulk_to_boundary_compatibility}
	Let $(M,g_{FRW})$ be a $4$-dimensional, simply connected Friedmann-Robertson-Walker spacetime with flat spatial sections, fulfilling the hypotheses of Definition \ref{defexp}.  If $\nu$ is as in \eqref{eq:nu} consider
	$$\Delta V(\tau) \doteq V(\tau)-V_{dS}(\tau)=a^2(\tau)\left[m^2+(\xi-\frac{1}{6})R(\tau)\right]-\frac{1}{\tau^2}\left[\frac{m^2}{H^2}+12(\xi-\frac{1}{6})\right]\:.$$ It turns that, 

(i) if $\mbox{Re}(\nu)<\frac{1}{2}$ and $\Delta V(\tau)=O(\tau^{-3})$ or,

(ii)  if $\mbox{Re}(\nu)<\frac{3}{2}$ and $\Delta V(\tau)=O(\tau^{-5})$, 

\noindent then the real-linear  map \eqref{eq:cosmological_bulk_to_boundary}
$\Gamma_c:\mathcal{E}^{obs}(M) \to C^\infty(\Im^-)$ satisfies the following facts.\\

\noindent	{\bf (a)} $\Gamma_c(\mathcal{S}(M))\subseteq\mathcal{S}_c(\Im)$, where $\mathcal{S}_c(\Im)$ is defined in \eqref{eq:boundary_functions} with $\Im$ identitified $\Im^-$.\\

\noindent {\bf (b)} $\Gamma_c$ preserves the symplectic forms: 
		$$\sigma_M(\varphi,\varphi')=\sigma_{\Im_c}(\Gamma_c(\varphi),\Gamma_c(\varphi')) \quad 
\forall \varphi,\varphi'\in\mathcal{E}^{obs}(M)\:,$$
		where $\sigma_M$ and $\sigma_{\Im}$ are respectively the symplectic forms \eqref{eq:symplectic_form} and \eqref{eq:boundary_symplectic_form}. In particular 
$\Gamma_c$ is injective since $\sigma_M$ is weakly non-degenerate.
\end{theorem}

\noindent With an argument identical to the one exploited for asymptotically flat spacetimes one proves the existence of the injective unital $^*$-algebra homomorphism 
\begin{equation}\label{eq:cosmological_bulk_to_boundary_homomorphism}
\iota_c:\mathcal{A}^{obs}(M)\to\mathcal{A}_c(\Im),
\end{equation}
 completely specified by the action on the generators of both algebras: $$\iota_c(\phi(f))=\left[\Gamma_c([f])\right]\:, \quad \forall f\in C^\infty_0(M)\:.$$
 Before concluding our analysis, we recall that the bulk spacetime $(M,g_{FRW})$ possesses a large isometry group and to each of its elements it corresponds a $^*$-automorphism of $\mathcal{A}^{obs}(M)$. At the same time, to every element of the $SG_{\Im^-}$ group, one can associate via Proposition \ref{prop:isomorphism_boundary_isometries} a $^*$-automorphism of $\mathcal{A}_c(\Im)$. Hence, in view of the Theorems \ref{th:cosmological_bulk_to_boundary_compatibility} and \eqref{theorem2} as well as of Proposition \ref{togroup}, it is natural to wonder whether these $^*$-automorphisms are intertwined. The following proposition summarizes the analysis in \cite{Dappiaggi:2007mx} about this problem. 
In particular, the proof of (b) is identical to the one of (b) Proposition \ref{prop:isometries_bulk_to_boundary_homomorphism}. 

\begin{proposition}\label{prop:cosmological_isometries_bulk_to_boundary_homomorphism}
Consider a $4$-dimensional, simply connected spacetime  $(M,g_{FRW})$ with flat spatial sections fulfilling the hypotheses of Definition \ref{defexp} and let $\xi$ be a complete Killing field, preserving $\Im^-$ as per Proposition \ref{togroup}, generating a one-parameter group of isometries $\{\chi^\xi_t\}_{t \in \mathbb R}$.
Let $\widetilde{\xi}$ be the unique extension of $\xi$ to $\Im^-$ as per Theorem \ref{theorem2} generating  the one-parameter subgroup $\{\widetilde{\chi}^{\widetilde{\xi}}_t\}_{t \in \mathbb R}\subset SG_{\Im^-}$. Then the following facts are true.\\

\noindent {\bf (a)} 	Referring to the standard pull-back action of automorphisms ($\chi_*(h)\doteq h \circ \chi^{-1}$) 
	\begin{equation}\label{eq:cosmological_interplay_isometries}
\widetilde{\chi}^{\widetilde{\xi}}_{t*}\circ \Gamma_c =\Gamma_c\circ \chi^\xi_{t*}\quad \forall t \in \mathbb R\:.
\end{equation}

\noindent {\bf (b)} 
Referring to the unital $^*$-algebra automorphism  $\alpha_{\chi^\xi_t}: \mathcal{A}^{obs}(M)\to\mathcal{A}^{obs}(M)$ 
defined in Proposition \ref{lem:isometries}, the $^*$-algebra automorphism  $\widetilde{\alpha}_{\widetilde{\chi}^{\widetilde{\xi}}_t}:\mathcal{A}_c(\Im)\to\mathcal{A}_c(\Im)$ defined in Proposition \ref{prop:isomorphism_boundary_isometries},
and  the unital $^*$-algebra embedding 
$\iota_{c}:\mathcal{A}^{obs}(M)\to\mathcal{A}_c(\Im)$ as in (\ref{eq:cosmological_bulk_to_boundary_homomorphism}),  (\ref{eq:cosmological_interplay_isometries})
extends to
\begin{equation}\label{eq:cosmological_interplay_algebras}
\iota_{c}\circ \alpha_{\chi^\xi_t}=\widetilde{\alpha}_{\widetilde{\chi}^{\widetilde{\xi}}_t}\circ\iota_{c}, \quad \forall t \in {\mathbb R}\:.
\end{equation}
\end{proposition}

\subsection{The case of Schwarzschild spacetime}\label{Sec:Schwarzschild}

In the last part of this chapter we focus on the specific case of a spherically symmetric black hole, solution of the vacuum Einstein equations with vanishing cosmological constant, namely the Schwarzschild spacetime and its Kruskal extension. This is an asymptotically flat spacetime at future null infinity in the sense of Definition \ref{Def:asymptotically_flat}, but, as one can infer readily from a close inspection of Figure \ref{fig2}, we cannot expect that we can define a bulk-to-boundary correspondence only by considering future or past null infinity. As a matter of fact, referring to the geometric structures introduced in Section \ref{Sec:Example}, the physical manifold, that we consider is $\mathcal{M}$ which is the union of $\mathcal{W}$, the static portion of Schwarzschild spacetime, $\mathcal{H}_{ev}$, the event horizon and $\mathcal{B}$, the black region. $\mathcal{M}$ is a globally hyperbolic spacetime and, as discussed in \cite{SW83}, it can be conformally embedded in a larger globally hyperbolic spacetime $\widetilde{\mathcal{M}}$ of which both $\Im^\pm$ and $\mathcal{H}$, the complete past horizon, are codimension $1$ submanifolds -- see Figure \ref{fig1}. Hence, also in view of the support properties of the solutions of the Klein-Gordon equation generated by smooth and compactly supported initial data, one expects that a full-fledged bulk-to-boundary correspondence, with properties similar to the one obtained in the first part of Section \ref{Sec:bulk-to-boundary}, can be set up only if one considers at the same time $\Im^-$ and $\mathcal{H}$. Actually, some further subtleties are present because the extended unphysical spacetime does not include future time infinity ($i^+$  appears just formally in the figures) whose nature is here much more complicated than the one described in (b) Remark \ref{Rem:alternative_asymptotically_flat} due to the presence of the curvature singularity of Schwarzschild spacetime.
For this reason this case has to be discussed separately -- see in particular \cite{Dappiaggi:2009fx}. 

The starting point is the same as in an asymptotically flat spacetime, namely a real, massless, conformally coupled scalar field $\varphi:\mathcal{M}\to\mathbb{R}$ whose dynamics is ruled by $P_0$ as in \eqref{eq:conformally_coupled}. The analysis of the classical theory in $\mathcal{M}$ and the construction of the associated algebra of observables are identical to the asymptotically flat case. Hence the building block is the real symplectic space $\left(\mathcal{E}^{obs}_0(\mathcal{M}),\sigma_0\right)$, while $\sigma_0$ is defined in \eqref{eq:symplectic_form} with $G$ replaced by $G_0$, being the unique causal propagator associated to $P_0$. The algebra of observables is thus nothing but $\mathcal{A}^{obs}_0(\mathcal{M})$ defined as in Definition \ref{def:algebra_of_observables} starting from $\mathcal{E}^{obs}_0(\mathcal{M})\doteq  \frac{C^\infty_0(\mathcal{M})}{P_0(C^\infty_0(\mathcal{M}))} \equiv \mathcal{S}(\mathcal{M})$. 

The key differences appear  at the level of definition of the symplectic space on the boundaries $\Im^-$ and $\mathcal{H}$, which need to be fine-tuned to the case in hand. Hence these spaces are rather different from \eqref{eq:boundary_functions}. We call

\begin{eqnarray}
\mathcal{S}(\Im^-)\doteq\left\{\psi\in \left. C^\infty(\Im^-)\:\right|\:\:\psi|_{\mathcal{O}_{i_0}}=0\quad\textrm{and}\quad \exists\, C_\psi,C^\prime_\psi\geq 0\right.\mbox{ such that } \notag\\
\left.|\psi(v,\omega)|\leq\frac{C_\psi}{\sqrt{1+|v|}},\;|\:\partial_v\psi(v,\omega)|\leq\frac{C^\prime_\psi}{1+|v|}\:, v \in \mathbb{R}\:, \omega \in \mathbb{S}^2 \right\}\label{eq:boundary_space_Im}
\end{eqnarray}
where $\mathcal{O}_{i_0}$ is a neighbourhood of spatial infinity $i_0$ whereas $v$ is the Eddington-Finkelstein coordinate \eqref{two}. In addition we define
\begin{eqnarray}
\mathcal{S}(\mathcal{H})\doteq\left\{\left.\Psi\in C^\infty(\mathcal{H})\;\right|\; \exists M_\Psi\geq 1\;\textrm{and}\; C_\Psi,C^\prime_\Psi\geq 0\right.\mbox{ such that }\notag\\
\!\!\!\!\!\!\!\!\!\!\!\!\!\!\left.|\Psi(U,\omega)|<\frac{C_\Psi}{\ln(U)},\;|\partial_U\Psi(U,\omega)|<\frac{C^\prime_\Psi}{U\ln(U)}\;\mbox{ if } U>M_\Psi\;\textrm{and}\;\omega \in \mathbb{S}^2\right\},\label{eq:boundary_space_H}
\end{eqnarray}
where $U$ is the global null coordinate in \eqref{UV}. Both $\mathcal{S}(\Im^-)$ and $\mathcal{S}(\mathcal{H})$ are symplectic spaces if endowed with the weakly-nondegenerate symplectic form \eqref{eq:boundary_symplectic_form}, which, to avoid possible disambiguation, will be indicated respectively with $\sigma_{\Im}$ and $\sigma_{\mathcal{H}}$. Observe that in both cases, the decay rate along the null coordinate guarantees that the integrand of \eqref{eq:boundary_symplectic_form} is an integrable function. As a next step, recalling the linear structure of both $\mathcal{S}(\Im^-)$ and $\mathcal{S}(\mathcal{H})$, we can construct their direct sum $\mathcal{S}_{\Im^-,\mathcal{H}}\doteq\mathcal{S}(\Im^-)\oplus\mathcal{S}(\mathcal{H})$ which is naturally endowed with the weakly non-degenerate symplectic form $\sigma_{\Im,\mathcal{H}}\equiv\sigma_{\Im}\oplus\sigma_{\mathcal{H}}$ defined as 
\begin{gather}\label{eq:SigmaIH}
\sigma_{\Im^-,\mathcal{H}}:\mathcal{S}_{\Im^-,\mathcal{H}}\times \mathcal{S}_{\Im^-,\mathcal{H}}\to\mathbb{C}\notag\\
\sigma_{\Im^-,\mathcal{H}}((\psi_1,\Psi_1),(\psi_2,\Psi_2))=\sigma_{\Im}(\psi_1,\psi_2)+\sigma_{\mathcal{H}}(\Psi_1,\Psi_2).
\end{gather}

Having identified a natural candidate for the space of kinematic configurations at the boundary, one needs to address the question on how to construct a map from $\mathcal{E}^{obs}_0(\mathcal{M})$ into it. Since we are dealing at the same time with $\Im^-$ and $\mathcal{H}$, two injection maps are needed. Starting from the former, we observe that, as already mentioned in Section \ref{Sec:Example}, we can realize $\mathcal{H}$ as part of the Kruskal manifold $\mathcal{K}$, which is a globally hyperbolic extension of $\mathcal{M}$, \cite{Wald}. Hence it is meaningful to define
\begin{equation}\label{eq:projection_on_horizon}
\Gamma_{\mathcal{H}}:\mathcal{E}^{obs}_0(\mathcal{M})\ni [f]\longmapsto \left.\widetilde{G}_0(f)\right|_{\mathcal{H}} \in C^\infty(\mathcal{H})\:,
\end{equation} 
where $\widetilde{G}_0$ is the unique causal propagator for $P_0$ on $\mathcal{K}$ which coincides with $G_0$ on $C_0^\infty(\mathcal{M})$ when restricting the functions in its image to $\mathcal{M}$. Focusing on $\Im^-$, the procedure to construct the injection map is identical to the one used in the case of asymptotically flat spacetimes, Proposition \ref{prop:conformal_rescaling} in particular. Without repeating the whole discussion we limit ourselves to recalling \eqref{eq:bulk-to-boundary} for the case in hand, where the role of $V$ is played by the manifold $\widetilde{\mathcal{M}}$:
\begin{equation}\label{eq:projection_on_scri}
\Gamma_{\Im^-}:\mathcal{E}^{obs}_0(\mathcal{M})\ni [f]\longmapsto\left.G_{\widetilde{\mathcal{M}}}(\Omega^{-3}f)\right|_{\Im^-} \in  C^\infty(\mathcal{\Im^-})\:, 
\end{equation}
where $G_{\widetilde{\mathcal{M}}}$ is the causal propagator of the massless, conformally coupled Klein-Gordon equation on $\widetilde{\mathcal{M}}$. The maps \eqref{eq:projection_on_horizon} and \eqref{eq:projection_on_scri} can be combined together and the following proposition characterizes the main properties of the ensuing map. The proof can be found in \cite{Dappiaggi:2009fx} and it relies strongly on the work of Dafermos and Rodnianski \cite{Dafermos}:

\begin{proposition}\label{prop:Schwarzschild_bulk_to_boundary}
	Let $\mathcal{M}$ denote the physical region of Schwarzschild spacetime and let $\Gamma_{\Im^-,\mathcal{H}}:\mathcal{E}^{obs}_0(\mathcal{M})\to  C^\infty(\Im^-)\oplus C^\infty(\mathcal{H})$ be 
	\begin{equation}\label{eq:Schwarzschild_full_proj}
	\Gamma_{\Im^-, \mathcal{H}} \doteq  \Gamma_{\Im^-}\oplus \Gamma_{\mathcal{H}}\:,
	\end{equation}
	where $\Gamma_{\mathcal{H}}$ and $\Gamma_{\Im^-}$ are defined respectively in \eqref{eq:projection_on_horizon} and \eqref{eq:projection_on_scri}. Then, the following facts hold.

\noindent {\bf (a)} $\Gamma_{\Im^-,\mathcal{H}}\left(\mathcal{E}^{obs}_0(\mathcal{M})\right)\subseteq \mathcal{S}(\Im^-)\oplus \mathcal{S}(\mathcal{H})$.\\
\noindent {\bf (b)} For every $\varphi,\varphi' \in\mathcal{E}^{obs}_0(\mathcal{M})$,
	$$\sigma_0(\varphi,\varphi')=\sigma_{\Im^-,\mathcal{H}}\left(\Gamma_{\Im^-, \mathcal{H}} (\varphi),\Gamma_{\Im^-, \mathcal{H}} (\varphi')\right),$$
	where the symplectic form $\sigma_{\Im^-,\mathcal{H}}$ is defined in  (\ref{eq:SigmaIH}).
As a consequence, since $\sigma_0$ is weakly non-degenerate, $\Gamma_{\Im^-,\mathcal{H}}$ is injective.
\end{proposition}

\noindent Having established the existence of a symplectomorphism between the generators of the bulk observables and a suitable boundary counterpart, we can extend this result at the level of the full algebra of observables. As a preliminary step we observe that, since we are considering both $\mathcal{S}(\Im^-)$ and $\mathcal{S}(\mathcal{H})$, the full boundary unital $^*$-algebra of observables is
\begin{equation}\label{eq:Schwarzschil_full_boundary_algebra}
\mathcal{A}_{\Im,\mathcal{H}}\doteq   \mathcal{A}(\Im^-)\otimes\mathcal{A}(\mathcal{H})\:,
\end{equation}
where $\mathcal{A}(\mathcal{H})$ and $\mathcal{A}(\Im^-)$ are the unital $^*$-algebras built respectively from \eqref{eq:boundary_space_H} and \eqref{eq:boundary_space_Im} according to Definition \ref{Def:boundary_algebra}. Starting from \eqref{eq:Schwarzschil_full_boundary_algebra} and from Proposition \ref{prop:Schwarzschild_bulk_to_boundary}, we can extend $\Gamma_{\Im^-, \mathcal{H}}$ to an injective $^*$-homomorphism
\begin{equation}\label{eq:Schwarzschild_*_homomorphism}
\iota_{\Im^-, \mathcal{H}}:\mathcal{A}^{obs}_0(\mathcal{M})\to\mathcal{A}_{\Im, \mathcal{H}},
\end{equation}
which is completely specified by the action on the generators, namely, if  $f \in C^\infty_0(\mathcal{M})$, it holds $\iota_{\Im^-, \mathcal{H}}(\phi(f))=[\Gamma_{\Im^-, \mathcal{H}}([f])]$. To conclude this section, we should study once more the interplay between the bulk isometries and $\Gamma_{\Im^-, \mathcal{H}}$. Yet, the whole analysis, both at $\mathcal{H}$ and $\Im^-$ would be a slavish repetition of the one for asymptotically flat spacetimes at future timelike infinity and hence we omit it. In particular this entails that the natural generalization to this scenario of Proposition \ref{prop:isometries_bulk_to_boundary_homomorphism} holds true. 


%
%
%
\chapter{Boundary-Induced Hadamard States}
\label{Ch:states} 

In the introduction to this work, we explained that the algebraic approach to quantum field theory is a two-step procedure. The first consists of the assignment to a physical system of a suitable complex  unital $^*$-algebra of observables which encompasses structural properties ranging from dynamics, to causality and locality. In the previous chapter we have shown, not only how to construct concretely such algebra for a free scalar field on a globally hyperbolic spacetime, but also how to embed it into a second auxiliary algebra, defined in terms of a kinematic space of functions whose domain is a null manifold representing the (conformal) boundary of the underlying spacetime. 

The second step of the algebraic approach consists of assigning to the algebra of observables a state, out of which it is possible to recover the canonical probabilistic interpretation proper of quantum theories. Goal of this chapter is to focus in detail on this specific aspect. We will emphasize that, while it is possible to construct a plethora of states, the most part of it is not physically acceptable. Hence a selection criterion is necessary and it goes under the name of {\em Hadamard condition}. Without entering at this stage into the details of the definition and of the motivations leading to such result, we stress that the explicit construction and identification of Hadamard states has been a rather daunting task for many years. Our main goal will be to show that the bulk-to-boundary correspondence for algebras, explained in the previous chapter, translates naturally at a level of states, thus providing a concrete mechanism to construct Hadamard states in a large class of interesting scenarios. 

\section{Algebraic States and the Hadamard Condition}\label{SecHadamardCond}

In this section we review the basic definitions and properties of algebraic states and we will discuss the Hadamard condition. This is a topic which has been thoroughly analysed by many authors and we do not claim to try to even get close to providing an omni-comprehensive review, rather we will focus on the more important aspects. If a reader is interested in further details, we recommend especially \cite{Brunetti:2015vmh,Khavkine:2014mta}.
\subsection{Algebraic states, GNS construction, and all that}
\begin{definition}\label{Def:states}
	Let $\mathcal{A}$ be a complex $^*$-algebra with unit $\mathbb{I}$. An (algebraic) {\bf state} is a $\mathbb{C}$-linear map $\omega:\mathcal{A}\to\mathbb{C}$ which is {\bf normalized} and {\bf positive}, {\em i.e.}, respectively,
	$$\omega(\mathbb{I})=1\;\;,\qquad\omega(a^*a)\geq 0\:\: \mbox{if $a\in\mathcal{A}$}. $$
\end{definition}	

\noindent The relevance of that notion of state is all encoded in the celebrated GNS theorem, which we state in the version valid for $^*$-algebras instead of $C^*$-algebras \cite{Haag:1992hx}\cite{Bratteli:1996xq} without giving a proof. An interested reader can find it in \cite[Th. 1]{Khavkine:2014mta}. If $\mathcal{D}$ is a subspace of a Hilbert space $\mathcal{H}$, the symbol $\mathcal{L}_{\mathcal{H}}(\mathcal{D})$ denotes the linear space of operators $A : \mathcal{D} \to \mathcal{H}$ which leaves $\mathcal{D}$ invariant:
$A(\mathcal{D}) \subset \mathcal{D}$.

\begin{theorem}[Gelfand-Naimark-Segal]\label{Th:GNS}
 Let $\mathcal{A}$ be a complex, unital $^*$-algebra and let $\omega:\mathcal{A}\to\mathbb{C}$ be a state thereon. Then there exists a quadruple of data  $(\mathcal{H}_\omega,\mathcal{D}_\omega,\pi_\omega,\Omega_\omega)$ where $\mathcal{D}_\omega$ is a dense subspace of a complex Hilbert space $\mathcal{H}_\omega$, while  $\pi_\omega:\mathcal{A}\to\mathcal{L}_{\mathcal{H}_\omega}(\mathcal{D}_\omega)$ is a $^*$-representation, that is, a linear map such that,
 $$\pi(\mathbb{I}) =I|_{\mathcal{D}_\omega}, \quad \pi_\omega(a^*)=\pi_\omega(a)^*|_{\mathcal{D}_\omega},\quad\pi_\omega(ab)=\pi_\omega(a)\pi_\omega(b) \quad \mbox{for all $a,b\in\mathcal{A}$,}$$
where the symbol $^*$ on the right-hand side denoting the adjoint with respect to $\mathcal{H}_\omega$.
 Furthermore $\Omega_\omega\in \mathcal{D}_\omega$ is a unit-norm vector such that $\mathcal{D}_\omega=\pi_\omega\left(\mathcal{A}\right)\Omega_\omega$ and  $$\omega(a)=\langle\Omega_\omega|\pi_\omega(a)\Omega_\omega\rangle_{\mathcal{H}_\omega} \quad \mbox{for all $a\in\mathcal{A}$.}$$ The GNS quadruple is unique up to unitary equivalences, that is, if there exists a second quadruple $(\mathcal{H}_\omega^\prime, \mathcal{D}^\prime_\omega,\pi^\prime_\omega,\Omega_\omega^\prime)$ satisfying the same properties as the first one, then there exists also a unitary operator $U:\mathcal{H}_\omega\to\mathcal{H}^\prime_\omega$ such that 
 $$U\Omega_\omega=\Omega^\prime_\omega,\quad U(\mathcal{D}_\omega)=\mathcal{D}_\omega^\prime,\quad  \pi^\prime_\omega(a)=U\pi_\omega(a)U^{-1} \quad \mbox{for all $a\in\mathcal{A}$.}$$
\end{theorem}

\noindent Notice that as $\mathcal{D}^\prime$  is dense, we also have that $\Omega_\omega$ is {\bf cyclic for} $\pi_\omega$, namely it holds $\overline{\pi\left(\mathcal{A}\right)\Omega_\omega}= \mathcal{H}_\omega$. Furthermore, every operator $\pi_\omega(a)$ is closable, with closure $\pi_\omega(a)^{**}$, becasue $\pi_\omega(a)^*$ has dense domain, it including $\mathcal{D}_\omega$. In addition it is symmetric at least if $a=a^*$.\\
In view of the uniqueness of the GNS quadruple up to unitary equivalences, we will omit henceforth the subscript $\omega$, unless there is a source of possible confusion. In between the many structural properties enjoyed by states, a notable one in relation to our analysis is the interplay with the $^*$-automorphism of the algebra, which is summarized in the following proposition whose proof can be found in \cite[Prop. 3]{Khavkine:2014mta}. A similar statement holds for antilinear $^*$-automorphism giving rise to antiunitary operators.

\begin{proposition}\label{Prop:automorphisms_and_GNS}
	Let $\mathcal{A}$ be a unital $^*$-algebra and let $\omega$ be a state thereon with associated GNS quadruple $(\mathcal{H}_\omega,\mathcal{D}_\omega,\pi_\omega,\Omega_\omega)$.\\ For every $^*$-automorphism $\alpha:\mathcal{A}\to\mathcal{A}$ such that $\omega\circ\alpha=\omega$, there exists a unique unitary operator $V^{(\alpha)}:\mathcal{H}\to\mathcal{H}$ such that 
	$$V^{(\alpha)}(\mathcal{D}_\omega)=\mathcal{D}_\omega,\quad V^{(\alpha)}\Omega_\omega=\Omega_\omega,\quad V^{(\alpha)}\pi_\omega(a) V^{(\alpha)-1}=\pi_\omega(\alpha(a))\quad \mbox{for all $a\in\mathcal{A}$.}$$
\end{proposition}

\begin{remark}\label{remagggKMS}
If a  state $\omega$ is invariant under a one-parameter group $\{\alpha_t\}_{t \in \mathbb R}$ of $^*$-automorphisms of $\mathcal{A}$, according to Proposition \ref{Prop:automorphisms_and_GNS}, the group is implemented in the GNS representation of $\omega$ by a one-parameter group of unitaries $\{U_t\}_{t \in \mathbb R}$ in ${\mathcal H}_\omega$ leaving the cyclic vector $\Omega_\omega$ invariant. If $U$ is strongly continuous and its self-adjoint generator $H$ has non-negative spectrum, then $\omega$ is said to be a {\bf ground state} with respect to $\alpha$. Another class of
$\{\alpha_t\}_{t \in \mathbb R}$-invariant states are the so-called {\bf KMS-states} \cite{Haag:1992hx}\cite{Bratteli:1996xq}, where the condition on the spectral properties of $H$ is replaced by the so-called {\bf KMS condition}. They describe {\em equilibrium thermal states} with respect to the {\em temporal evolution} $\{\alpha_t\}_{t \in \mathbb R}$ at a given temeperature $T$ (and possibly also at a non vanishing chemical potential).
\end{remark}

\noindent Having established the notion of state on a generic unital $^*$-algebra, we focus our attention on the specific case of $\mathcal{A}^{obs}(M)$, the algebra of observables for a real, Klein-Gordon field on a globally hyperbolic spacetime, constructed in Definition \ref{def:algebra_of_observables}.
From Proposition \ref{prop:fieldoperator}, we know that every element of $\mathcal{A}^{obs}(M)$ can be written 
\begin{equation}\label{eq:deca2b} a = c {\mathbb I} + \sum_{N=1}^\infty
\sum_{k_1, \ldots, k_N=1}^\infty c^{(N)}_{k_1k_2\cdots k_N}\phi(f^{(N)}_{k_1})\phi(f^{(N)}_{k_2})\cdots \phi(f^{(N)}_{k_N})
\end{equation}
for $c^{(N)}_{k_1k_2\cdots k_N}\in \mathbb C$ and $f^{(N)}_{k_j} \in C_0^\infty(M)$ depending on $a$ (but not uniquely fixed by it), such that only a {\em finite} number of them do not vanish. 
The field operator $\phi(f) := [[f]]$, $f \in C_0^\infty(M)$ has been introduced in (\ref{eq:fieldoperator}) as a double-quotient equivalence class.
It takes into account both the dynamics, encoded in the first quotient (\ref{eq:observables}) referred to the equation of motion, and the bosonic canonical commutation relations, encoded in the second  quotient (\ref{eq:algebrabservables}).
Since a state  $\omega : \mathcal{A}^{obs}(M) \to \mathbb{C}$  is  $\mathbb C$-linear, it is  therefore completely determined by the class of complex numbers  $$\widetilde{\omega}_n([[f_1]],\ldots, [[f_n]]) \doteq \omega(\phi(f_1)\cdots \phi(f_n))\:,\quad \mbox{ with $n= 1, 2,\ldots$ and $f_k \in C_0^\infty(M)$.}$$

These $\widetilde\omega_n$ are also often referred to as {\em n-point (correlation) functions}. For all practical purposes, it is preferable not to start from the double classes $[[f]] \in \mathcal{A}^{obs}(M)$, but  directly from the functions in $C^\infty_0(M)$, hence looking for {\bf n-point (correlation) functions} as multi-linear maps 
\begin{equation}\label{eq:correlation_functions}
\omega_n:\underbrace{C^\infty_0(M)\times\cdots \times C^\infty_0(M)}_{n}\to\mathbb{C},\quad\omega_n(f_1,\ldots,f_n)\doteq\omega(\phi(f_1)\cdots \phi(f_n)).
\end{equation}
Observe that $\omega_n$ is constructed to  be compatible with the quotient generated by the elements of the form $Ph$, $h\in C^\infty_0(M)$, {\em i.e.}, $\omega_n(f_1,...,f_n)=0$ if one of its arguments has the form $f_k = Ph$, $h\in C^\infty_0(M)$. Similarly it agrees with commutation relations:
\begin{align}&\omega_n(f_1,\ldots, f_k,f_{k+1}, \ldots, f_n) = \omega_n(f_1,\ldots, f_{k+1},f_k, \ldots, f_n) \nonumber  \\
& + iG(f_k,f_{k+1}) \omega_{n-2}(f_1,\ldots, f_{k-1},f_{k+1}, \ldots, f_n)\:,\end{align}
for $n=2,3,\ldots$ and $f_k \in C^\infty_0(M)$. It is easy to see that, taking this property into account, the subclass of the completely symmetrized correlation functions is enough to fix  the remaining ones.\\ An  assignment of a class of $\mathbb{C}$-multilinear maps 
$\omega_n:C^\infty_0(M)\times\cdots \times C^\infty_0(M)\to\mathbb{C}$, for  $n=1,2,\ldots$,
 satisfying the two constraints above defines a state over $\mathcal{A}^{obs}(M)$ {\em up to the positivity requirement}.
Let us focus our attention on  a  distinguished subclass of states whose $n$-point functions are determined by $\omega_2$.

\begin{definition}\label{def:quasifree_state}
  Let $\mathcal{A}^{obs}(M)$ be the unital $^*$-algebra as per Definition \ref{def:algebra_of_observables}. A state $\omega:\mathcal{A}^{obs}(M)\to\mathbb{C}$ is called {\bf quasi-free/Gaussian} if the associated $n$-point correlation functions are such that $\omega_{n}$ vanishes for $n$ odd, whereas, if $n>0$ is even and  $f_k\in C_0^\infty(M)$, for all $k=1,...,n$,
  \begin{equation}\label{eq:quasifree}
 {\omega}_{n}(f_1,...,f_{n})=\sum\limits_{\mbox{\small partitions of $\{1,\ldots,n\}$}}{\omega}_2(f_{i_1}, f_{i_2})\cdots {\omega}_2(f_{i_{n-1}}, f_{i_{n}})
  \end{equation}
 and
the partitions refers to the class of all possible decomposition
of set $\{1, 2,\ldots, n\}$ into $n/2$ pairwise disjoint subsets of $2$ elements
$\{i_1, i_2\}$, $\{i_3, i_4\}$, $\ldots$, $\{i_{n-1}, i_{n}\}$
with $i_{2k-1} < i_{2k}$ for $k = 1, 2, \ldots , n/2$ and $i_{2k-1}<i_{2k+1}$ for $k = 1, 2, \ldots , n/2-1$.
\end{definition}

\noindent From now on, we will focus only on quasi-free states. Therefore the first step consists of understanding what are the constraints to be imposed on the two-point function in order for it to identify a quasi-free state according to Definition \ref{def:quasifree_state}. The following proposition starts by assuming that a state on $\mathcal{A}^{obs}(M)$ is given, thus deducing some necessary conditions that the two-point function must obey:

\begin{proposition}\label{propomega2}
	Let $\omega:\mathcal{A}^{obs}(M)\to\mathbb{C}$ be a state. Then the $2$-point function $\omega_2$ is a complex-valued bi-linear map on $C_0^\infty(M)$ such that for all $f,f^\prime\in C^\infty_0(M)$,
	\begin{enumerate}
		\item $\omega_2(Pf,f^\prime)=\omega_2(f,Pf^\prime)=0$  where $P$ is the Klein-Gordon operator \eqref{eq:dynamics};
		\item $\omega_2(f,f^\prime)-\omega_2(f^\prime,f)=iG(f,f^\prime)$ where the right hand side is defined in \eqref{eq:symplectic_form};
		\item $\mbox{\em Im}(\omega_2(f,f^\prime))=\frac{1}{2}G(f,f^\prime)$ 
\item $\mbox{\em Re}(\omega_2(f,f^\prime))=\omega_{2,s}(f,f^\prime)\doteq \frac{1}{2}(\omega_{2}(f,f^\prime)+ \omega_{2}(f^\prime,f))$,
		\item $\frac{1}{4}\left|G(f,f^\prime)\right|^2\leq\omega_2(f,f)\omega_2(f^\prime,f^\prime)$.
	\end{enumerate}
\end{proposition}

\begin{proof}
	The first and the second identities are particular cases of already discussed properties of $\omega_n$ when $n=2$. The third one arises form te GNS theorem observing that
$\phi(h)^*= \phi(h)$ and thus $\overline{\omega(\phi(f)\phi(f'))} = \overline{\langle \Omega_\omega \:|\:\pi_\omega(\phi(f))
\pi_\omega(\phi(f'))\Omega_\omega\rangle}=$
$$\langle \pi_\omega(\phi(f))
\pi_\omega(\phi(f')) \Omega_\omega \:|\:\Omega_\omega\rangle = \langle \Omega_\omega\:|\: \pi_\omega(\phi(f')^*)\pi_\omega(\phi(f)^*)\Omega_\omega\rangle = \omega(\phi(f')\phi(f))$$ and eventually applying 2. 
Property 4 is now a trivial consequece of 2 and 3.
Only the last property remains to be proven. By applying again the GNS theorem to $\omega$, it descends, that
	\begin{gather*}
	|\omega_2(f,f^\prime)|^2= |\langle\Omega_\omega \:|\:\pi_\omega(\phi(f))\pi_\omega(\phi(f^\prime))\Omega_\omega\rangle|^2\leq\\ \leq\|\pi_\omega(\phi(f))\Omega_\omega\|^2_{\mathcal{H}_\omega}\,\|\pi_\omega(\phi(f^\prime))\Omega_\omega\|^2_{\mathcal{H}_\omega}
	\end{gather*}
	where we used the Cauchy-Schwarz inequality. Since $\omega_2(h,h)=\|\pi_\omega(\phi(h))\Omega_\omega\|^2_{\mathcal{H}_\omega}$ we conclude that 
	$$\frac{1}{4}\left|G(f,f^\prime)\right|^2 =[\mbox{ Im}(\omega_2(f,f^\prime))]^2\leq|\omega_2(f,f^\prime)|^2\leq\omega_2(f,f)\,\omega_2(f^\prime,f^\prime)\:,$$
which entails item 5. \qed
\end{proof}

\noindent Having established which properties are enjoyed by the two-point function of a given state, we can ask ourselves how we can characterize and construct quasi-free states.
The crucial observation is that, given a state $\omega$, the  GNS construction permits to construct a preferred closed subspace $\mathcal{H}^{(1)}_\omega \subset \mathcal{H}_\omega$ as 
$$\mathcal{H}^{(1)}_\omega  := \overline{\mbox{Span}_{\mathbb C}\{ \pi_\omega(\phi(f)) \Omega_\omega\:|\: f \in C_0^\infty(M)\}}$$
where the bar denotes the closure. From Proposition \ref{propomega2}, the GNS scalar product restricted to $\mathcal{H}^{(1)}_\omega $
satisfies, 
$$\langle \pi_\omega(\phi(f))| \pi_\omega(\phi(f')) \rangle_{\mathcal{H}_\omega} = \omega_{2,s}(f,f') + \frac{i}{2}G(f,f')\quad f,f' \in C_0^\infty(M)\:,$$
this relation being completely fixed it by Hermitian linearity.
If $\omega$ is quasifree, it turns out that $\mathcal{H}^{(1)}_\omega $ is enough to re-costruct 
 $\mathcal{H}_\omega$ as it is nothing but the {\em symmetrized Fock space} with {\em one-particle Hilbert space} $\mathcal{H}^{(1)}_\omega $ and {\em vacuum} $\Omega_\omega$. Hence, as it can be grasped from the identity above, only the symmetric part of $\omega_2$ really matters in the construction of the GNS representation of the quasifree state $\omega$. Such symmetric component is a real scalar product on the real symplectic vector space $\mathcal{E}^{obs}(M)$ of classes $[f]$, the symplectic form $\sigma_M$ being instead defined by $G$ as discussed in Proposition \ref{prop:symplectic_form}: $\sigma([f],[f'])= G(f,f')$. It seem plausible that real scalar products on $\mathcal{E}^{obs}(M)$ define Gaussian states on $\mathcal{A}^{obs}(M)$.  The following abstract proposition is a key step in this direction. Its proof can be found in \cite[Prop. 3.1]{Kay}.
\begin{proposition}\label{Prop:One_particle_Hilbert}
	Let $(\mathcal{S},\sigma)$ be a real symplectic space, whose symplectic form is weakly non-degenerate. Then the following holds true:
	\begin{enumerate}
		\item If $\mu:\mathcal{S} \times \mathcal{S}\to\mathbb{R}$ is a real scalar product on $\mathcal{S}$ such that 
		\begin{equation}\label{ineqs}\frac{1}{4}|\sigma(x,x^\prime)|^2\leq\mu(x,x)\mu(x^\prime,x^\prime),\quad\forall x,x^\prime\in \mathcal{S}\:,\end{equation}
		then there exists a so-called {\bf one-particle structure} $(K,\cal{H})$, where $\cal{H}$ is a complex Hilbert space with scalar product $\langle\cdot|\cdot \rangle$, while $K:\mathcal{S}\to \cal{H}$ is such that
		\begin{enumerate}
			\item $K$ is a $\mathbb{R}$-linear map such that $K(\mathcal{S})+iK(\mathcal{S})$ is dense in $\cal{H}$,
			\item $\langle Kx | Kx^\prime\rangle =\mu(x,x^\prime)+\frac{i}{2}\sigma(x,x^\prime)$, for all $x,x^\prime\in \mathcal{S}$;
		\end{enumerate}
	\item if $(K^\prime,\cal{H}^\prime)$ is a second pair satisfying conditions a. and b., then there exists an isometric, surjective $\mathbb C$-linear operator $V:\cal{H}\to \cal{H}^\prime$ such that $VK=K^\prime$.
	\end{enumerate}
\end{proposition}

\noindent This proposition is at the heart of the following result on the characterization of quasi-free states for the algebra of observables of a real Klein-Gordon field. It appears in various but equivalent forms in the literature, see Proposition B1 in \cite{Dappiaggi:2009fx} and references therein, \cite{KW} especially.

\begin{proposition}\label{Prop:characterization_quasifree_states}
	Let $\mathcal{A}^{obs}(M)$ be the algebra of observables as in  Definition \ref{def:algebra_of_observables} and let $\mathcal{E}^{obs}(M)$ be the real vector space equipped with the symplectic form $\sigma_M$ determined by $G$ as in Definition \ref{prop:symplectic_form}.\\
If  $\mu: \mathcal{E}^{obs}(M) \times \mathcal{E}^{obs}(M)\to \mathbb{bR}$ is a real scalar product satisfying (\ref{ineqs}) with respect to   $\sigma_M$, then there exists a (unique) quasi-free state $\omega:\mathcal{A}^{obs}(M)\to\mathbb{C}$ such that its two-point function reads
	$$\omega_2(f,f^\prime)=\mu([f],[f^\prime])+\frac{i}{2}\sigma([f],[f^\prime]),\quad\forall f,f^\prime\in C^\infty_0(M)\:.$$
	Moreover, up to unitary equivalence, the associated GNS quadruple $(\mathcal{H}_\omega,\mathcal{D}_\omega,\pi_\omega,\Omega_\omega)$ is such that
	\begin{enumerate}
		\item $\mathcal{H}_\omega$ is the Bosonic (symmetrized) Fock space whose one-particle Hilbert space is $\mathcal{H}^{(1)}_\omega = \cal{H}$ as per Proposition \ref{Prop:One_particle_Hilbert} with $\sigma_M$ in place of $\sigma$,
		\item $\Omega_\omega$ is the vacuum of the Fock space $\mathcal{H}_\omega$,
		\item $\mathcal{D}_\omega$ coincides with the dense subspace of the finite complex linear combinations of $\Omega_\omega$ and of the vectors of the form $a^\dagger(K[f_1])...a^\dagger(K[f_n])\Omega_\omega$ for $n\in\mathbb{N}$ and for $f_k\in C^\infty_0(M)$, $k=1,...,n$, where $a^\dagger(K[f_k])$ is the standard creation operator\footnote{For the definition and for the analysis of the properties of creation and annihilation on Fock spaces, refer to \cite{Bratteli:1979tw,Bratteli:1996xq}.} associated to $f_k$ and $K :\mathcal{E}^{obs}(M)\to \mathcal{H}^{(1)}_\omega $ is the $\mathbb{R}$-linear map defined in  Proposition \ref{Prop:One_particle_Hilbert} with $\mathcal{E}^{obs}(M)$ in place of $\mathcal{S}$,
		\item the representation $\pi$ is completely determined by the creation ($a^\dagger$) and by the annihilation ($a$) operators via the identity $\pi_\omega(\phi(f))=a(K[f])+a^\dagger(K[f])$. 
\item The state $\omega$ is {\bf regular}, that is the field operator $\pi_\omega(\phi(f))$ is essentially self-adjoint on $\mathcal{D}_\omega$ for every $f\in C_0^\infty(M)$.
	\end{enumerate}
\end{proposition}

\begin{remark}\label{Rem:obvious} $\null$\\
{\bf (1)} Taking Remark \ref{remgenconstr} into account, everything we found can be automatically translated if we replace the algebra of observables $\mathcal{A}^{obs}(M)$ with $\mathcal{A}(\Im)$ constructed in Definition \ref{Def:boundary_algebra}. Also in this case we will keep the terminology of quasi-free/Gaussian state. The motivation for using this name might appear obscure and it is better understood in the context of Weyl $C^*$-algebras (see \cite{Bratteli:1979tw,Bratteli:1996xq}). We only say that in Minkowski spacetime the Poincar\'e invariant state of a free field theory is Gaussian.\\
{\bf (2)} Consider a quasifree state $\omega$ which, respectively, is a ground state/KMS state referring to a one-parameter group of $^*$-automorphisms $\{\alpha_t\}_{t \in \mathbb R}$ (see Remark \ref{remagggKMS}) representing a group of isometries of $M$ as in Proposition \ref{lem:isometries}. In this case, the one-particle space $\mathcal{H}^{(1)}_\omega$ is invariant under  the unitary group $\{e^{-itH}\}_{t \in \mathbb R}$ 
implementing $\{\alpha_t\}_{t \in \mathbb R}$ as in  Remark \ref{remagggKMS}. 
An important result due to Kay and used in \cite{KW} establishes that such $\omega$ is unique if the restriction $H|_{\mathcal{H}^{(1)}_\omega}$ has no zero eigenvalues. \\
{\bf (3)} If the two-point function $\omega_2$ viewed as a linear map $\omega_2: C_0^\infty(M) \to \mathcal{D}'(M)$ is sequentially continuous in the natural topologies, it defines an one-to-one associated  {\bf Schwartz kernel} $\omega_2 \in \mathcal{D}'(M\times M)$ such that $$\omega_2(f,f') = \int_{M\times M}\omega_2(x,x')f(x)f'(x) d\mu_g(x)d\mu_g(x') \quad \forall f,f' \in C_0^\infty(M)\:,$$ in distributional sense in view of {\em Schwartz kernel theorem} \cite{Friedlander:2010eqa,Hormander}.
\end{remark}

\subsection{Hadamard States}

As one can infer from the definition itself, algebraic states play a key role in the analysis of the physical properties of a given quantum system. Yet, even in the simple case of a real, massive scalar field, it is not hard to imagine that there are infinite possible choices for states $\omega$ on $\mathcal{A}^{obs}(M)$. 
A rather natural criterion for restricting the possibilities is to look for those $\omega$ which are compatible with the underlying background symmetries, in the sense that every isometry induces a $^*$-automorphism of $\mathcal{A}^{obs}(M)$ fulfilling the hypotheses of Proposition \ref{Prop:automorphisms_and_GNS}. For globally hyperbolic spacetimes which are maximally symmetric, under mild hypotheses \cite{AJ86} this condition leads to uniqueness result for quasifree states. In Minkowski spacetime, for instance,  we end up with the Poincar\'e vacuum \cite{Haag:1992hx}, while in de Sitter spacetime, we select the so-called Bunch-Davies state \cite{BD}. Dropping maximal symmetry of the spacetime uniqueness may fail to be true. However,  further uniqueness results exists for invariant ground and KMS states at given temperature when referring to Killing symmetries (typically time evolutions) with positive self-adjoint generators when unitarily implemented in the GNS Hilbert space ((2) Remark \ref{Rem:obvious}).
In the absence of global symmetries of the spacetime, no preferred choice of a quantum state seems to exist.
It is, therefore, natural to wonder whether it is possible to establish a further selection criterion to decide whether a given state is physically reasonable. This question has been thoroughly investigated and debated in the past decades leading to a general consensus in formulating the so-called {\em Hadamard condition}. In the case in hand the rationale behind it is the following: A minimal requirement for a state on $\mathcal{A}^{obs}(M)$ to be physically acceptable is {\em 1)}, to yield finite quantum fluctuations of all observables, {\em 2)} to have the same ultraviolet behaviour of the Poincar\'e vacuum, {\em 3)} to allow for a covariant construction of an algebra of Wick polynomials. This last point in particular is vital to allow for discussing a time-ordered product, hence accounting for interactions with a perturbative scheme. A discussion on the role of Hadamard states in the construction of Wick polynomials will enter Chapter \ref{Ch:wicks}, whereas a thorough analysis can be found in \cite[Sec. 3.1]{Khavkine:2014mta} or in \cite{Hollands:2001fb,Hollands:2001nf}.

In the remaining part of this subsection we shall discuss the mathematical formulation of the Hadamard condition. Besides having in mind $\mathcal{A}^{obs}(M)$, we will restrict our attention to quasi-free states as per Definition \ref{def:quasifree_state}. This choice is only due to a matter of convenience, although this assumption is not necessary for all results that we will be reporting. A proof of this assertion has been given in \cite{Sanders:2009sw}. 

Let us recall that we are considering a four-dimensional globally hyperbolic spacetime $(M,g)$ and a real scalar field thereon obeying \eqref{eq:dynamics}. Let $\mathcal{O}\subset M$ be a geodesically convex open neighbourhood \cite{BEE,ONeill}. We define the integral kernel of the so-called {\bf Hadamard parametrix} $H\in\mathcal{D}^\prime(\mathcal{O}\times\mathcal{O})$, a bidistribution, local approximate solution of the Klein-Gordon equation which encodes the same short-distance behaviour of the Poincar\'e vacuum:
\begin{equation}\label{eq:Hadamard_parametrix}
H(x,y)=\lim_{\epsilon\to 0^+}\frac{U(x,y)}{\sigma_\epsilon(x,y)}+V(x,y)\ln\frac{\sigma_\epsilon(x,y)}{\lambda^2},
\end{equation}
where the limit has to be interpreted in a distributional sense and where
\begin{enumerate}
	\item $\sigma_\epsilon(x,y)\doteq\sigma(x,y)+ 2i\epsilon(T(x)-T(y))+\epsilon^2$. Here $\epsilon>0$ is a regularizing parameter, $\sigma(x,y)$ is the halved, squared geodesic distance between $x$ and $y$ (well defined in geodesically convex sets), while $T:\mathcal{O}\to\mathbb{R}$ is a local time function, increasing towards the future,
	\item $\lambda$ is an arbitrary scale length, which makes the argument of the logarithm dimensionless,
	\item $U$ and $V$ are smooth scalar functions on $\mathcal{O}\times\mathcal{O}$, which depend only on the geometry of the background and on \eqref{eq:dynamics}. They can be determined by imposing that, up to smooth terms, $H$ solves in both entries the equation of motion. It turns out that $U(x,y)$ is the square root of the Van-Vleck-Morette determinant -- see \cite{Moretti:1999ez}, while $V(x,y)$ can be expressed as a series in powers of $\sigma$, that is $V(x,y)=\sum\limits_{n=0}^\infty v_n(x,y)\sigma^n(x,y)$. The unknowns $v_n\in C^\infty(\mathcal{O}\times\mathcal{O})$ are the so-called {\em Hadamard coefficients} which can be determined recursively from the equation of motion by imposing consistency with the Poincar\'e vacuum in the ultraviolet regime as initial condition \cite{Moretti:2001qh}. Observe that the convergence of the series is asymptotic, unless a suitable set of cut-off functions is introduced, so to obtain pointwise convergence \cite{Hollands:2001nf}. Different choices of these cut-off functions change $H$ by adding a smooth term which does not affect the rest of the discussion.
\end{enumerate}

\begin{remark}\label{Rem:dimension_dependence}
	While the notion and existence of a Hadamard parametrix does not depend on the dimension of the underlying background, its form \eqref{eq:Hadamard_parametrix} instead is critically based on $\dim M$ being $4$. In \cite{Moretti:2001qh} a reader can find how \eqref{eq:Hadamard_parametrix} changes by varying the dimension of $M$.
\end{remark}

\noindent The following is a local characterization of Hadamard states (a global one appear in \cite{KW} but it is equivalent to the local one as established by Radzikowski in the papers we quote below).

\begin{definition}\label{Def:local_Hadamard_form}
	Let $\omega:\mathcal{A}^{obs}(M)\to\mathbb{C}$ be a quasi-free state as in Definition \ref{def:quasifree_state} whose associated two-point function as in \eqref{eq:correlation_functions} descends from a bi-distribution, that is $\omega_2\in\mathcal{D}^\prime(M\times M)$ with associated integral kernel $\omega_2(x,y)$. We say that $\omega$ is of {\bf Hadamard form} if, in every geodesically convex neighbourhood $\mathcal{O}$, $\omega_2(x,y)-H(x,y)\in C^\infty(\mathcal{O}\times\mathcal{O})$ where $H(x,y)$ is the Hadamard parametrix \eqref{eq:Hadamard_parametrix}.	
\end{definition}

\noindent It is not difficult to imagine that it is very hard to use Definition \ref{Def:local_Hadamard_form} to verify in concrete cases whether a given state is of Hadamard form, unless one considers special backgrounds for which it is necessary to check the form of the integral kernel of the two-point function in a small number of geodesic neighbourhoods. Notable instances are de Sitter and Minkowski spacetime. It is thus necessary to find an alternative and possibly more efficient way to characterize Hadamard states.

\subsubsection{The microlocal definition of Hadamard states}

 The question just raised has been open for many years and it has been solved in the seminal papers of Radzikowski \cite{Radzikowski,Radzikowski2}. The key point lies in the observation that, if  the two-point function of a Hadamard state is a bidistribution, the most convenient way to describe its singular structure is to use the techniques proper of microlocal analysis and the associated notion of wavefront set. We will review succinctly the main definitions and properties lying at the heart of this concept. Notation and conventions will follow those used in the textbook \cite{Hormander} or in \cite[Sec. 3.3]{Khavkine:2014mta}.
 
Let  ${\cal D}(M)\doteq C_0^\infty(M, \mathbb C)$ be the set of complex-vaòued smooth and compactly-supported functions on a spacetime $(M,g)$, which we may assume to be globally-hyperbolic only for consistency with the assumptions in this book.

If $M \doteq U$ is an open set of $\mathbb{R}^n$, independently from $g$,  $\mathcal{D}^\prime(U)$  is the space of {\bf distributions} over $U$, {\em i.e.}, $\mathbb{C}$-linear maps $u: {\cal D}(U) \to \mathbb{C}$ such that $u(f_n) \to 0$ if  ${\cal D}(U) \ni f_n \to 0$ uniformly together with the series of every fixed order of derivatives and there is a compact set $K \subset U$ including all supports of the functions $f_n$. 

Coming back to the general case,
 with $\mathcal{D}^\prime(M)$ we indicate the space of  $\mathbb{C}$-linear functionals
$u : {\cal D}(M) \to \mathbb C$ such that, for every coordinate patch $(U,\psi)$ over $M$
there is distribution $u_\psi \in {\cal D}'(U)$ such that\footnote{$\det{g_\psi}$ is the determinant of the metric in coordinates $\psi$ so that $|\det{g_\psi}|^{1/2} dx^1\ldots dx^n$ is the standard measure on $(M,g)$ written in local coordinates $\psi$ over $U$.}
$u(f) = u_\psi( |\det{g_\psi}|^{1/2} \cdot f \circ \psi^{-1} )$ if $f \in {\cal D}(M)$ and $\mbox{supp}(f) \subset U$.  These functionals $u$ are named {\bf distributions on} $(M,g)$. 

An open neighbourhood $V$ of $k_0\in\mathbb R^m$ is called {\bf conic} if $k\in V$ implies $\lambda k\in V$ for all $\lambda >0$. 

\begin{definition}
	\label{def_WaveFrontSet}
	Let $u\in{\cal D}^\prime(U)$, $U$ being an open subset of $\mathbb{R}^m$.
	A point $(x_0,k_0)\in U\times (\mathbb R^m\setminus\{0\})$ is called {\bf regular directed} if there exists $f\in{\cal D}(U)\doteq\mathcal{D}(\mathbb{R}^m)|_U$ with $f(x_0)\neq 0$ such that, for every $n\in\mathbb N$, there exists a constant $C_n\geq 0$ such that
	$$|\widehat{fu}(k)|\le C_n (1+|k|)^{-n}$$
	for all $k$ in an open conic neighbourhood of $k_0$. The {\bf wavefront set} $WF(u)$, of $u \in {\cal D}'(U)$
	is the complement in $U \times (\mathbb R^m\setminus\{0\})$ of the set of all regular directed points of $u$.
\end{definition}

\noindent We observe that, in the preceding definition, one should interpret $U\times(\mathbb{R}^m\setminus\{0\})$ as $T^*U\setminus\{0\}$. Let $\Gamma\subset T^*U\setminus\{0\}$ be a {\bf cone}, that is, if $(x,k) \in \Gamma$, then $(x,\lambda k) \in \Gamma$  for all $\lambda>0$. If $\Gamma$  is closed {\em in   the topology of} $T^*U\setminus 0$, we call
$${\cal D}_\Gamma'  \stackrel {\mbox{\scriptsize  def}} {=} \{ u \in {\cal D}'(U)\:|\: WF(u) \subset \Gamma\}\:.$$ 

\noindent Following \cite{Hormander}, 

\begin{definition} 
	If $u_j \in {\cal D}_\Gamma'(U)$ is a sequence and $u \in {\cal D}_\Gamma'(U)$, we write $u_j \to u$ in ${\cal D}_{\Gamma}'(U)$ if
	\begin{enumerate}
		\item $u_j \to u$ weakly in ${\cal D}'(U)$ as $j \to +\infty$,
		\item $\sup_j \sup_V |\xi|^N |\widehat{\phi u_j}(p)| <\infty$, $N=1,2, \ldots$, if $\phi \in {\cal D}(U)$ and $V \subset T^*U$ is any  closed cone, whose projection on $U$ is $\textrm{supp}(\phi)$,
		such that $\Gamma \cap V  = \emptyset$. 
	\end{enumerate}
	In this case, we say that $u_j$  converges to $u$ in the {\bf H\"ormander pseudotopology}. 
\end{definition}
Test functions are dense even with respect to this notion of convergence \cite{Hormander}.
\begin{proposition}
	If $u \in {\cal D}'_\Gamma(U)$, there exists a sequence $u_j \in {\cal D}(U)$ such that $u_j \to u$ in $ {\cal D}'_\Gamma(U)$.
\end{proposition}

\noindent A notion related to that of wavefront set is the following: We call $x\in U$ a {\bf regular point} of a distribution $u\in{\cal D}^\prime(U)$ if there exists an open neighborhood $O\subset U$ of $x$ such that $\langle u, f \rangle = \langle h_u, f\rangle$ for some $h_u\in {\cal D}(U)$ and every $f \in {\cal D}(U)$ supported in $O$. The closure of the complement of the set of regular points is called the {\bf singular support} of $u$.

\begin{theorem}	\label{thm_PropertiesWavefront} Let $u\in{\cal D}^\prime(U)$, where $U\subset \mathbb R^m$ is open and non-empty. The following statements hold true:
\begin{enumerate}
\item  $u\in C^\infty(U)$ if and only if $WF(u)=\emptyset$. More precisely, the {\em singular support} of $u$ is the projection of $WF(u)$ on $\mathbb R^m$.
	
\item  If $P$ is a partial differential operator on $U$  with smooth coefficients: $$WF(Pu)\subseteq WF(u)\,.$$
	
\item  Let  $O\subset\mathbb R^m$ be an open set and let $\chi:O\to U$ be a diffeomorphism. The pull-back $\chi^*u \in {\mathcal D}'(O)$ of $u$ defined by $\chi^*u(f)=u(\chi_*f)$ for all $f\in{\cal D}(O)$ fulfils
	$$WF(\chi^*u)=\chi^*WF(u)  \stackrel {\mbox{\scriptsize  def}} {=} \left\{(\chi^{-1}(x),\chi^*k)\;|\;(x,k)\in WF(u)\right\}\,,$$
	where $\chi^*k$ denotes the pull-back of $\chi$ in the sense of {\em cotangent vectors}.
\item Let $O\subset \mathbb R^n$ be an open set and $v \in {\cal D}'(O)$, then $WF(u\otimes v)$ is included in
	$$(WF(u) \times WF(v))\cup (({\rm supp} u \times \{0\})\times WF(v)) \cup 
	(WF(u)\times ({\rm supp} v \times \{0\})) \:.$$
	
\item Let $O\subset \mathbb R^n$,  $K \in {\cal D}'(U\times O)$ and   $f \in {\cal D}(O)$, then
	$$WF({\cal K}f) \subset \{(x, p)\in TU \setminus 0 \:|\:  (x,y,p, 0) \in WF(K) \:\:\mbox{for some}\: \:y \in {\rm supp}(f) \}\:,$$ 
	where  ${\cal K} :{\cal D}(O) \mapsto {\cal D}'(U)$ is the continuous linear map associated to $K$
	in view of Schwartz kernel theorem.
\end{enumerate}
\end{theorem}

From the third property we infer that the wavefront set transforms covariantly under diffeomorphisms as a subset of $T^*U$, with $U$ an open subset of $\mathbb R^m$. Therefore the definition of $WF$ can be extended to distributions on a generic manifold $M$ simply by glueing together wavefront sets in different coordinate patches of $M$ with the help of a partition of unity. As a result, for $u\in{\cal D}^\prime(M)$, $WF(u)\subset T^*M\setminus\{0\}$. Also the notion of convergence in the H\"ormander pseudotopology  as well as the statements of theorem \ref{thm_PropertiesWavefront} extend to this more general scenario.

To conclude this very short survey,  we report on a few remarkable results concerning  (a) the theorem of  {\em propagation of singularities}, (b) the {\em product of distributions}, (c) {\em composition of kernels}. These all play a key role in our construction of Hadamard states.

\begin{remark}\label{remarkps}$\null$
	
	{\bf (1)} If $M$ is a manifold and $P$ is thereon a differential operator of order $m\geq 1$ which reads in local coordinates $P= \sum_{|\alpha| \leq m} a_\alpha(x) \partial^\alpha$ (it is assumed that $a_\alpha\neq 0$ for some $\alpha$ with $|\alpha|=m$), where $a$ is a {\em multi-index} \cite{Hormander}, and $a_\alpha$ are smooth coefficients,  then the polynomial $\sigma_P(x,k) = \sum_{|\alpha| =  m} a_\alpha(x) (ik)^\alpha$ is called the {\bf principal symbol} of $P$. It is possible to prove that $(x, \xi) \mapsto \sigma_P(x,k)$ determines a  well defined function on $T^*M$ which, in general is complex valued.   The {\bf characteristic set} of $P$, indicated by
	$char(P) \subset T^*M\setminus\{0\}$, denotes the set of zeros of $\sigma_P$ made of {\em non-vanishing} covectors.
	The principal symbol  $\sigma_P$ can be used as a {\em Hamiltonian function} on $T^*M$ and the maximal solutions of Hamilton equations define the {\bf local flow} of $\sigma_P$ on $T^*M$.   
	
	{\bf (2)} 
	The principal symbol of the Klein-Gordon operator is  $-g^{\mu\nu}(x) k_\mu k_\nu$. If $M$ is a Lorentzian manifold and $P$ is a {\bf normally hyperbolic operator}, namely its principal symbol is the same as the one of tje Klein-Gordon operator, then the integral curves of the local flow of $\sigma_P$ are nothing but the lift to $T^*M$
	of the  {\em geodesics} of the metric $g$ parametrized by an {\em affine parameter}.
	Finally, we call  $char(P) = \{(x,k) \in T^*M \setminus\{0\} \:|\: g^{\mu\nu}(x) k_\mu k_\nu =0\}$  
\end{remark}

\begin{theorem}\label{tps}
	Let $P$ be a differential operator on a manifold $M$ whose principal symbol is real valued. If $u,f \in {\cal D}'(M)$ are such that $Pu=f$ then the following facts hold true:
	\begin{enumerate}
		\item  $WF(u) \subset char(P) \cup WF(f)$,
	
	\item $WF(u)\setminus WF(f)$ is invariant under the local flow of $\sigma_P$ on  $T^*M \setminus WF(f)$.  
\end{enumerate}
\end{theorem}

Let us conclude with the celebrated H\"ormander sufficient criterion to define the {\em product of distributions}, see  \cite[Th. 8.2.10]{Hormander}. In the following, if $\Gamma_1, \Gamma_2 \subset T^*M\setminus\{0\} $ are closed cones,
$$\Gamma_1+ \Gamma_2  \stackrel {\mbox{\scriptsize  def}} {=} \left\{(x, k_1+k_2) \subset T^*M \;|\;(x,k_1)\in \Gamma_1,\; (x,k_2)\in \Gamma_2\:\: \mbox{for some $x\in M$}\right\}.$$

\begin{theorem}[Product of distributions]\label{teoprod}
	Consider  a pair of closed cones $\Gamma_1, \Gamma_2  \subset  T^*M \setminus\{0\}$. If 
	$$\Gamma_1 + \Gamma_2 \not \ni (x, 0) \quad \mbox{for all $x \in M$,}$$ 
	then there exists a unique bilinear  map,  the {\bf product} of $u_1$ and $u_2$,
	$${\cal D}'_{\Gamma_1} \times  {\cal D}'_{\Gamma_2} \ni (u_1,u_2) \mapsto u_1u_2 \in {\cal D}'(M),$$ such that
	\begin{enumerate}
		\item it reduces to the standard pointwise product  if $u_1, u_2 \in {\cal D}(M)$,
	
	\item it is jointly sequentially continuous in the H\"ormander pseudotopology: If $u_j^{(n)} \to u_j$
	in $D_{\Gamma_j}(M)$ for $j=1,2$ then $u^{(n)}_1u^{(n)}_2 \to u_1u_2$ in ${\cal D}_{\Gamma}(M)$, where  $\Gamma$ is a closed cone in $T^*M\setminus\{0\}$ defined as  $\Gamma   \stackrel {\mbox{\scriptsize  def}} {=}  \Gamma_1\cup \Gamma_2\cup \left(\Gamma_1\oplus\Gamma_2\right)$. 
	\end{enumerate}
	In particular the following bound always holds if the above product is defined:  
	\begin{eqnarray} WF(u_1u_2)\subset  \Gamma_1\cup \Gamma_2\cup \left(\Gamma_1+\Gamma_2\right)\,. \label{boundWF}\end{eqnarray}
\end{theorem}

We present just one last theorem concerning the composition of distributional kernels which is analogous to \cite[Th. 8.2.14]{Hormander}. Let $X,Y$ be generic smooth manifolds. If $K \in {\cal D}'(X \times Y)$, the continuous map  associated to $K$ by the Schwartz kernel theorem  will be denoted  by ${\cal K}: {\cal D}(Y) \to {\cal D}'(X)$. We call
\begin{align}WF(K)_X &  \stackrel {\mbox{\scriptsize  def}} {=} \{(x,k) \:|\: (x,y, k,0) \in WF(K) \quad \mbox{for some $y\in Y$}\}\nonumber\:, \\
WF(K)_Y &  \stackrel {\mbox{\scriptsize  def}} {=} \{(y,k^\prime) \:|\: (x,y,0,k^\prime) \in WF(K) \quad \mbox{for some $x\in X$}\}\nonumber\:, \\
WF'(K) &  \stackrel {\mbox{\scriptsize  def}} {=} \{(x,y,k,k^\prime) \:|\: (x,y,k,-k^\prime) \in WF(K)\}\nonumber\:, \\
WF'(K)_Y &  \stackrel {\mbox{\scriptsize  def}} {=} \{(y,k^\prime) \:|\: (x,y,0,-k^\prime) \in WF(K) \quad \mbox{for some $x\in X$}\}\:.\nonumber 
\end{align}

\begin{theorem}\label{thm:kern}
	Consider three smooth manifolds $X, Y, Z$ and let $K_1 \in {\cal D}'(X\times Y)$, $K_2\in {\cal D}'(Y \times Z)$.
	If $WF'(K_1)_Y \cap WF(K_2)_Y = \emptyset$  and the projection $${\rm supp} K_2 \ni (y,z) \mapsto z \in Z$$ is proper, then the composition ${\cal K}_1 \circ {\cal K}_2$ is well defined, giving rise to $K\in {\cal D}'(X,Z)$, and it reduces to the standard one when the kernel are smooth. It holds also (the symbol $\circ$ denoting the composition of relations)
	\begin{multline} 
	WF'(K) \subset  WF'(K_1) \circ WF'(K_2)
	\cup (WF(K_1)_X \times Z \times \{0\}) \\ 
	{} \cup  (X \times \{0\} \times WF'(K_2)_Z) \:.
	\end{multline} 
\end{theorem}

For our purposes, the above excursus on the notion and on the properties of the wavefront set of a distribution is both necessary and useful in order to report the most important result proven in \cite{Radzikowski,Radzikowski2} which complements the local definition of Hadamard states with a global one:

\begin{theorem}\label{Th:global_to_local}
	Let $\omega:\mathcal{A}^{obs}(M)\to\mathbb{C}$ be a quasi-free state for the algebra of observables of a real, Klein-Gordon field on a four-dimensional globally hyperbolic spacetime $(M,g)$, whose associated two-point function identifies a distribution $\omega_2\in\mathcal{D}^\prime(M\times M)$ ((3) Remark \ref{Rem:obvious}). The following statements are equivalent:
	\begin{enumerate}
		\item $\omega$ is of Hadamard form in the sense of Definition \ref{Def:local_Hadamard_form},
		\item the wavefront set of $\omega_2$ is the following:
		\begin{equation}\label{eq:WF}
		WF(\omega_2)=\left\{(x,y;k_x,k_y)\in T^*(M\times M)\setminus\{0\}\;|\; (x,k_x)\sim(y,-k_y),k_x
		\triangleright 0\right\}, 
		\end{equation}
		where $\{0\}$ indicates the zero-section. The symbol $(x,k_x)\sim(y,k_y)$ stands for the existence of a null geodesic connecting $x$ and $y$ with cotangent vectors
		respectively $k_x$ and $k_y$, whereas $k_x\triangleright 0$ means that $k_x\neq 0$ is  future-directed.		
	\end{enumerate}
\end{theorem}

We observe that we asked for $(M,g)$ to be four-dimensional. This is only due to the fact that we want to make a connection with \eqref{eq:Hadamard_parametrix} which is valid for $\dim M=4$. Theorem \ref{Th:global_to_local} can be adapted to any value of $\dim M\geq 2$ provided that a suitable counterpart of \eqref{eq:Hadamard_parametrix} is used.

From the general properties of the wave-front set and of the state on $\mathcal{A}^{obs}(M)$, it turns out that two Hadamard states must thus differ only by a smooth term, hence yielding a global counterpart of Definition \ref{Def:local_Hadamard_form}. In addition, although at first sight it might look otherwise, it is much easier to verify \eqref{eq:WF} in concrete cases rather then using \eqref{eq:Hadamard_parametrix}. 

The first question, which can be asked in view of Definition \ref{Def:local_Hadamard_form} and of Theorem \ref{Th:global_to_local} is whether one can guarantee that Hadamard states do exist. A positive answer to this question relies on a deformation argument and it was first presented in \cite{FNW}. This result is unfortunately of limited applicability since it does not offer a concrete mean to construct Hadamard states, nor it guarantees invariance under the action of the background isometries. For many years one of the main difficulties in using the Hadamard condition was indeed the lack of explicit examples, barring the well-known Poincar\'e vacuum and Bunch-Davies state. In the past twenty years several studies have been made to bypass this hurdle, especially when the background metric possessed a large isometry group which allows for a construction of quasi-free states via a mode expansion. Limiting ourselves to a real, Klein-Gordon field a few and non exhaustive list of results of this line of investigation are the following:
\begin{itemize}
	\item If $(M,g)$ is a static spacetime with $\dim M\geq 2$, then the ground as well as the KMS states constructed with respect to the underlying timelike Killing field is of Hadamard form, as proven in \cite{Sahlmann:2000fh},
	\item A further interesting class of Hadamard states in general Friedmann-Robertson-Walker are the {\it states of low energy} constructed in \cite{Olbermann:2007gn}. These states minimise the energy density integrated in time with a compactly supported weight function. This construction has been generalised to encompass almost equilibrium states in \cite{Kusku:2008zz} and expanding spacetimes with less symmetry in \cite{Them:2013uka}. Noteworthy in this context is the thorough analysis of the mode expansion presented in \cite{Avetisyan:2012wq}. 
	\item Still in the context of cosmological spacetimes, another interesting construction of Hadamard states is due to Brum and Fredenhagen \cite{Brum:2013bia}, a work which modifies the so-called SJ states \cite{Afshordi:2012jf} to make them compatible with the Hadamard criterion  \cite{Fewster:2012ew,Fewster:2013lqa}. 
\end{itemize}

\noindent In the following sections we will be discussing a different constructive criterion to identify quasi-free Hadamard states, which relies on the possibility to find an injective map from the algebra of observables for a real, Klein-Gordon field into a suitable counterpart living on a codimension $1$ null submanifold.  This method was further extended in \cite{Gerard:2014hla} where, by solving a suitable characteristic initial value problem, it was shown that the injection mentioned above is actually a bijection once the bulk and the boundary symplectic spaces are chosen carefully. As a by-product of this choice, it descends that the state, that we will be constructing, is pure.
An alternative analysis for the case of asymptotically Minkowski spacetimes can be found in \cite{VasyWrochna}. Therein the construction of Hadamard states from data on a characteristic hypersurfaces in the asymptotically de Sitter scenario is also discussed.

In the context of cosmological spacetimes, the mode expansion methods mentioned above were used by Parker \cite{Parker} to introduce the so called adiabatic states. Such approach has been adapted by Junker \cite{Junker} and by Junker and Schroe \cite{Junker:2001gx} in order to have a state of Hadamard form. Similar ideas have been successively developed in \cite{Gerard:2012wb,Gerard:2016jqj}, see also \cite{Vasy:2016rum}.

\subsection{Natural Gaussian states for the boundary algebra}

Up to now we have established a criterion to select Gaussian states of physical relevance for a real, scalar field theory on a globally hyperbolic spacetime. Our next goal is to provide a concrete mean to construct some  of these states in the models considered in Chapter \ref{Ch:qft-null-surfaces}. The underlying idea is based on the following rather straightforward definition

\begin{definition}\label{Def:induced_state}
Let $\mathcal{A}$ and $\mathcal{A}^\prime$ be two unital $^*$-algebras and let $\iota:\mathcal{A}\to\mathcal{A}^\prime$ be an injective $^*$-homomorphism. For any algebraic state $\omega:\mathcal{A}^\prime\to\mathbb{C}$ we call {\bf induced (or pulled-back) state} the linear map $\iota^*\omega:\mathcal{A}\to\mathbb{C}$ such that, for all $a\in\mathcal{A}$
$$(\iota^*\omega)(a)\doteq\omega(\iota(a)).$$
\end{definition}

\begin{remark} The definition itself entails that positivity and the normalization of $\iota^*\omega$ are traded directly from those of $\omega$.
\end{remark}
From our point of view, Definition \ref{Def:induced_state} is particularly interesting since we have already established the existence of an injective $^*$-isomorphism between suitable algebras of observables, either $\mathcal{A}^{obs}(M)$ or $\mathcal{A}^{obs}_0(M)$ and a counterpart defined at future or past null infinity, respectively $\mathcal{A}_c(\Im)$ and $\mathcal{A}(\Im)$. Our main idea is to start from the latter case trying to identify a distinguished state on both $\mathcal{A}_c(\Im)$ and $\mathcal{A}(\Im)$ by exploiting the huge symmetry group of these algebras. Subsequently we will investigate the properties of the pulled-back state built according to Definition \ref{Def:induced_state}, in particular the Hadamard condition and the invariance under the action of the bulk isometries. As in the preceding chapter, the case when the background is Schwarzschild spacetime will be treated separately.

To start with, our first goal is the identification of a quasi-free state in the sense of Definition \ref{def:quasifree_state} for both $\mathcal{A}(\Im)$ and for $\mathcal{A}_c(\Im)$. These are characterized in Definition \ref{Def:boundary_algebra} as soon as one identifies $\Im^+$ and $\Im^-$ to $\Im \doteq \mathbb{R}\times \mathbb{S}^2$, by fixing a Bondi coordinate frame 
$(u, s) \in \mathbb{R}\times \mathbb{S}^2$ -- where $u$ is  an affine parameter along the null geodesics forming 
$\Im^+$ and $\Im^-$ respectively. As discussed in Remark \ref{remarkIDIM}, different choices of physically admissible Bondi frames lead to isomorphic $^*$-algebras $\mathcal{A}(\Im)$ and $\mathcal{A}_c(\Im)$ respectively. The states we go to define enjoy the same invariance property as we  shall see shortly, so the construction is not affected by arbitrariness.
 We start from a preliminary result:

\begin{proposition}\label{prop:boundary_scalar_product}
	Referring to the real symplectic spaces $(\mathcal{S}(\Im), \sigma)$  and  $(\mathcal{S}_c(\Im), \sigma_c)$
	as in Proposition \ref{propSSc} the following facts hold:\\
{\bf (a)} The map $\mu_\Im:\mathcal{S}(\Im)\times \mathcal{S}(\Im)\to\mathbb{R}$ 
	\begin{equation}\label{eq:boundary_scalar_product}
		\mu_\Im(\psi,\psi^\prime)\doteq \mbox{Re}\left[ \int\limits_{\mathbb{R^+}\times\mathbb{S}^2} 2k\overline{\widehat{\psi}(k,s)}\,\widehat{\psi}^\prime(k,s) \: dk\: d\mu_{\mathbb{S}^2}(s)\right]\quad 
		\forall \psi,\psi^\prime\in\mathcal{S}(\Im)\:, 
	\end{equation}
	 is a real scalar product on $\mathcal{S}(\Im)$, 
	where the symbol $\;\widehat{}\;$ identifies the Fourier-Plancherel transform (\ref{FTeq}) along $u$. Here $s$ is a short cut to indicate the angular coordinates on $\mathbb{S}^2$, {\it i.e.} $s=(\theta,\varphi)$. \\
{\bf (b)}	The same result holds true replacing $\mathcal{S}(\Im)$ with $\mathcal{S}_c(\Im)$ and defining a real scalar product $\mu_{\Im c}$.
\end{proposition}

\begin{proof}
	To prove that $\mu$ is well-defined it suffices to observe that the integrand in \eqref{eq:boundary_scalar_product} is integrable in view of \eqref{eq:boundary_functions} and \eqref{eq:boundary_space_Im}. In addition $\mu$ is bilinear and symmetric.
	$\mu_\Im(\psi,\psi) \geq 0$ by construction and $\mu_\Im(\psi,\psi) =0$ implies $\widehat{\psi}(k,s) =0$ almost everywhere if $k\geq 0$. However, 
	 exploting the fact that, since $\psi$ is real, $\widehat{\psi}(-k,s) =\overline{\widehat{\psi}(k,s)}$,
	 we also have that $\widehat{\psi}(k,s)=0$ almost everywhere so that $\psi=0$.\qed
\end{proof}

\noindent The next step consists of adapting Proposition \ref{Prop:One_particle_Hilbert} to the case in hand. 

\begin{proposition}\label{prop:boundary_one_particle_space}
	Let $\mathcal{H}_\Im=L^2(\mathbb{R}^+\times\mathbb{S}^2,2k dk\,d\mu_{\mathbb{S}^2})$ and let $K_\Im:\mathcal{S}(\Im)\to\mathcal{H}_\Im$ be the $\mathbb{R}$-linear map
	$$ K_\Im :  \mathcal{S}(\Im)\ni\psi\mapsto \widehat{\psi}|_{\mathbb{R}^+\times\mathbb{S}^2} \in \mathcal{H}_\Im\:.$$ Then 
	\begin{enumerate}
		\item $\overline{K_\Im(\mathcal{S}(\Im))}=\mathcal{H}_\Im$,
		\item for every $\psi,\psi^\prime\in\mathcal{S}(\Im)$, 
		\begin{equation}\label{eq:sympl_to_scalar_boundary}
		\sigma_\Im(\psi,\psi^\prime)= 2\mbox{Im}\left[\left\langle K_\Im\psi| K_\Im\psi^\prime\right\rangle_{\mathcal{H}_\Im}\right],
		\end{equation}
		\noindent where $\sigma_\Im$ is the symplectic form \eqref{eq:boundary_symplectic_form}, while $\langle\cdot  |\cdot \rangle_{\mathcal{H}_\Im}$ is the inner product of $\mathcal{H}_\Im$.
	\end{enumerate}
The analogous result holds true if one replaces $\mathcal{S}(\Im)$ with $\mathcal{S}_c(\Im)$, defining $\mathcal{H}_{\Im_c} \doteq \mathcal{H}_\Im$ and
indicating by $K_{\Im_c}$ the map corresponding to $K_\Im$. 
\end{proposition}
\noindent The proof of this proposition has been given in \cite[Prop. 2.2]{Dappiaggi:2008dk} for the case of $\mathcal{S}_c(\Im)$ while the one for $\mathcal{S}(\Im)$ can be found in \cite{Dappiaggi:2005ci}. 
\begin{remark}
 $\mu$ and $\sigma_{\Im}$ satisfy (\ref{ineqs}) as trivial consequence of Cauchy-Schwartz inequality 
for  $\langle\cdot  |\cdot \rangle_{\mathcal{H}_\Im}$. Similarly 1. above implies a. in 1. of 
Proposition \ref{Prop:One_particle_Hilbert}
\end{remark}
The last step consists of observing that the construction of a quasi-free state as discussed in Proposition \ref{Prop:characterization_quasifree_states} and following Definition \ref{def:quasifree_state} can be used also in this scenario as suggested in (1) Remark \ref{Rem:obvious}. This justifies the following definition:

\begin{definition}\label{Def:boundary_quasi_free_state}
	Let $\mathcal{A}(\Im)$ be the algebra built out of $\mathcal{S}(\Im)$ as per Definition \ref{Def:boundary_algebra}. Then we call $\omega_\Im:\mathcal{A}(\Im)\to\mathbb{C}$ the quasi-free state whose associated two-point function reads:
	\begin{equation}\label{eq:boundary_2pt_function}
	\omega_{2,\Im}(\psi,\psi^\prime)=\mu_\Im(\psi,\psi^\prime)+\frac{i}{2}\sigma_{\Im}(\psi,\psi^\prime),\quad\forall\psi,\psi^\prime\in\mathcal{S}(\Im)
	\end{equation}
	where $\mu_\Im$ is the inner product \eqref{eq:boundary_scalar_product} while $\sigma_{\Im}$ is the symplectic form \eqref{eq:symplectic_form}. The analogous  definition applies replacing $\mathcal{A}(\Im)$ with $\mathcal{A}_c(\Im)$, the algebra constructed out of $\mathcal{S}_c(\Im)$ and obtaining the quasifree state $\omega_{\Im_c}$.
\end{definition}
A direct calculation using \eqref{eq:boundary_2pt_function} and the explicit expressions \eqref{eq:boundary_scalar_product} and \eqref{eq:symplectic_form} shows the following notable consequence:

\begin{lemma}\label{lem:boundary_integral_kernel}
	Let $\omega_\Im$ be the quasi-free state of $\mathcal{A}(\Im)$ as per Definition \ref{Def:boundary_quasi_free_state}. The associated two-point function $\omega_{2,\Im}$ identifies an element of $\mathcal{D}^\prime(\Im\times\Im)$ whose integral kernel reads:
	\begin{equation}\label{eq:notable_2pt_function}
	\omega_{2,\Im}(x,x^\prime)=\lim_{\epsilon\to 0^+}\frac{\delta(s,s^\prime)}{(u-u^\prime-i\epsilon)^2},
	\end{equation}
	where $x=(u,s)$, $x^\prime=(u^\prime,s^\prime)$, $\delta(s,s')$ is the standard Dirac delta on the unit-radius sphere $\mathbb{S}^2$ referred to the standard rotational-invariant measure (with total area $4\pi$) thereon, and the limit has to be interpreted in a distributional sense.
\end{lemma}

\noindent An important question which remains to be answered is why one should select as state at the boundary the one of Definition \ref{Def:boundary_quasi_free_state}. The reason lies in the following uniqueness results. Since they are strongly tied to the boundary symmetry group which is different when one considers asymptotically flat spacetimes, hence $\mathcal{A}(\Im)$, or cosmological spacetimes, hence $\mathcal{A}_c(\Im)$, two separate statements are necessary.
 The proof of next theorem can be found 
in \cite[Th.1.1,Th.3.1]{Moretti} and \cite[Th.3.2]{Moretti2}, while that of the subsequent  statement appears in \cite[Prop. 4.1,Th.4.1]{Dappiaggi:2007mx},

\begin{theorem}\label{Th:uniqueness_asymptotically_flat}
	The quasifree state $\omega_\Im:\mathcal{A}(\Im)\to\mathbb{C}$ in Definition \ref{Def:boundary_quasi_free_state} enjoys the following properties.

{\bf (a)} $\omega_\Im$ is invariant under  the representation $\alpha$ of $G_{BMS}$, which is built out of the
$^*$-automorphisms as in Proposition \ref{prop:isomorphism_boundary_isometries}. In particular,  $\alpha$ is unitarily implementable in the GNS representation of $\omega_\Im$.

{\bf (b)} If $\{\alpha_t\}_{t \in \mathbb{R}}$ represents the one-parameter subgroup of $G_{BMS}$ generated by $\partial_u$, for an arbitrary choice of a Bondi frame on $\Im^+$ (also different to that used to initially identify $\Im^+$ to $\Im \doteq \mathbb{R} \times \mathbb{S}^2$), 
the unitary group which implements $\{\alpha_t\}_{t \in \mathbb{R}}$ leaving fixed the cyclic GNS vector of the GNS triple associated to $\omega_\Im$ is strongly continuous with nonnegative generator.

{\bf (c)} Fixing a Bondi frame on $\Im^+$ as in (b), there is a unique quasifree state $\omega :\mathcal{A}(\Im)\to\mathbb{C} $ such that 
the unitary group which implements $\{\alpha_t\}_{t \in \mathbb{R}}$ leaving fixed the cyclic GNS vector of the GNS construction  associated to $\omega$ is strongly continuous with nonnegative generator. It holds $\omega =\omega_\Im$.
\end{theorem}

\begin{theorem}\label{Th:uniqueness_cosmological}
	The quasifree state $\omega_{\Im_c}:\mathcal{A}_c(\Im)\to\mathbb{C}$ in Definition \ref{Def:boundary_quasi_free_state} enjoys the following properties.

{\bf (a)} $\omega_{\Im_c}$ is invariant under  the representation $\widetilde{\alpha}$
of the group  $SG_{\Im^-}$ in terms of
$^*$-automorphisms  as in Proposition \ref{prop:isomorphism_boundary_isometries},
so that, in particular,  $\widetilde{\alpha}$ is unitarily implementable in the GNS representation of $\omega_{\Im_c}$. 

{\bf (b)} If $\{\widetilde{\alpha}_t\}_{t \in \mathbb{R}}$ represents the one-parameter subgroup of $SG_{\Im^-}$ generated by $\partial_u$, for an arbitrary choice of a Bondi frame on $\Im^-$ (also different to that used to initially identify $\Im^+$ to $\Im \doteq \mathbb{R} \times \mathbb{S}^2$ but in agreement with (\ref{quasih})), 
the unitary group which implements $\{\widetilde{\alpha}_t\}_{t \in \mathbb{R}}$ leaving fixed the cyclic GNS vector of the GNS triple associated to $\omega_\Im$ is strongly continuous with nonnegative generator.

{\bf (c)} Fixing a Bondi frame on $\Im^+$ as in (b), there exists a unique quasifree state $\omega :\mathcal{A}(\Im)\to\mathbb{C} $ such that 
the unitary group which implements $\{\widetilde{\alpha}_t\}_{t \in \mathbb{R}}$ leaving fixed the cyclic GNS vector of the GNS construction associated to $\omega$ is strongly continuous with nonnegative generator. It holds $\omega =\omega_{\Im_c}$.
\end{theorem}

\section{The Bulk-To-Boundary Correspondence for States}\label{se:pullback}

In the previous section we have identified a state both for the $^*$-algebra $\mathcal{A}(\Im)$ and $\mathcal{A}_c(\Im)$ which enjoys a uniqueness property with respect to the action of the underlying symmetry group. Our next goal is to make use of Definition \ref{Def:induced_state} to start from this state inducing a counterpart in the bulk, which is at the same time of Hadamard form and invariant under the action of all bulk isometries. As in Section \ref{Sec:bulk-to-boundary}, it is more convenient to split the discussion in two parts, one for asymptotically flat spacetimes and one for the cosmological scenario. A third  last part will concern relevant states on Schwarzschild spacetime where the boundary is more complicated.

\subsection{Asmptotically Flat Spactimes} Let us consider $(M,g)$ to be a four-dimensional globally hyperbolic spacetime, which is asymptotically flat at null infinity and admits future time infinity as per \eqref{Rem:alternative_asymptotically_flat}.  In addition we assume also that there exists an open subset $V$ of the unphysical spacetime $(\widetilde{M},\widetilde{g})$ such that $\overline{J^-(\Im^-;\widetilde{M})\cap M}\subset V$ and $(V,\widetilde{g}|_V)$ is  globally hyperbolic. (Recall that, per convention, we are omitting the symbol of the conformal embedding from $M$ to $\widetilde{M}$.)
On top of $(M,g)$ we consider a real scalar field $\phi:M\to\mathbb{R}$, whose dynamics is ruled by \eqref{eq:conformally_coupled}. To such system we can associate the algebra of observables $\mathcal{A}_0(M)$ built as in Definition \ref{def:algebra_of_observables}. In \eqref{eq:bulk_to_boundary_homomorphism} we have introduced the injective $^*$-homomorphism $\iota_{\Im}$ from $\mathcal{A}^{obs}_0(M)$ to $\mathcal{A}(\Im)$, the latter being the $^*$-algebra from Definition \ref{Def:boundary_algebra} built out of \eqref{eq:boundary_functions}.
 We finally  consider the algebraic, quasi-free state $\omega_\Im$ over $\mathcal{A}(\Im)$ built in Definition \ref{Def:boundary_quasi_free_state} and, by applying Definition \ref{Def:induced_state} we introduce the pulled-back state $\omega_M:\mathcal{A}^{obs}_0(M)\to\mathbb{C}$
\begin{equation}\label{eq:bulk_state_asymptotically_flat}
  \omega_M(a)\doteq \iota_{\Im}^*\omega_\Im\:.
\end{equation}
It is immediate to realize that $\omega_M$ inherits the property of being quasi-free and, starting from $\omega_{\Im}$ as in Definition \ref{Def:boundary_quasi_free_state}, the associated two-point function in the sense of \eqref{eq:correlation_functions} reads
\begin{equation}\label{eq:bulk_two_point_asymptotically_flat}
\omega_{2,M}(f,f^\prime)=\omega_{\Im}([\Gamma_{\Im}([f])]\otimes [\Gamma_{\Im}([f^\prime])]),
\end{equation}
where $\Gamma_{\Im}$ is defined in Theorem \ref{Th:bulk_to_boundary}, while  $f,f^\prime\in C^\infty_0(M)$. The next theorem further characterizes the properties of $\omega_M$:

\begin{theorem}\label{Th:Hadamard_form_Asymptotically_Flat}
	The state $\omega_M:\mathcal{A}^{obs}_0(M)\to\mathbb{C}$  defined in \eqref{eq:bulk_state_asymptotically_flat} is of {\bf Hadamard form} in the sense of Definition \ref{Def:local_Hadamard_form}.
\end{theorem}
This theorem has been proven in \cite{Moretti2} and, since it is one of the main results, that we report in this work, we shall give an outline of the proof, highlighting the main points. 

\begin{sproof}
In view of the equivalence proven in Theorem \ref{Th:global_to_local}, it is sufficient to check that the wave front set of the two-point function is of the form \eqref{eq:WF}. 
In order to prove the inclusion 
\begin{equation}\label{eq:WFsupset}
WF(\omega_{2,M})\supset \left\{(x,y;k_x,k_y)\in T^*(M\times M)\setminus\{0\}\;|\; (x,k_x)\sim(y,-k_y),k_x
\triangleright 0\right\}
\end{equation}
we notice that the antisymmetric part of $\omega_{2,M}$ is proportional to the causal propagator $G$, namely
\[
G(f,f^\prime) = -i \omega_{2,M}(f,f^\prime) +i\widetilde{\omega}_{2,M}(f,f^\prime),
\] 
where $\widetilde{\omega}_{2,M}(f,f^\prime)=\omega_{2,M}(f^\prime, f)$ for all $f,f^\prime\in C^\infty_0(M)$. Furthermore, on account of the properties of the wave front set, it holds
\begin{equation}\label{eq:WFprove1}
WF(G) \subset  WF(\omega_{2,M}) \cup WF(\widetilde{\omega}_{2,M}).
\end{equation}
Notice that the wave front set of the causal propagator is 
\begin{equation}\label{eq:WFprove2}
WF(G) = \left\{(x,y;k_x,k_y)\in T^*(M\times M)\setminus\{0\}\;|\; (x,k_x)\sim(y,-k_y)\right\}
\end{equation}
and at the same time 
\begin{gather*}
WF(\widetilde{\omega}_{2,M}) = \left\{(x,y;k_x,k_y)\in T^*(M\times M)\setminus\{0\} \;|\; (x^\prime,y^\prime;k_{x^\prime},k_{y^\prime})\in WF(\omega_{2,M})\right.\\
\left.\textrm{with}\;(x^\prime,k_{x^\prime})=(y,k_y)\;\textrm{and}\;(y^\prime,k_{y^\prime})=(x,k_x)\right\}.
\end{gather*}
Hence, if \eqref{eq:WFsupset} does not hold, we get a contradiction between \eqref{eq:WFprove1} and \eqref{eq:WFprove2}. In order to prove the inclusion 
\begin{equation}\label{eq:WFsubset}
WF(\omega_{2,M})\subset \left\{(x,y;k_x,k_y)\in T^*(M\times M)\setminus\{0\}\;|\; (x,k_x)\sim(y,-k_y),\; k_x
\triangleright 0\right\}
\end{equation}
we notice that
\[
\omega_{2,M}(f,f^\prime) = \omega_{2,\Im}(\widetilde{G}(f),\widetilde{G}(f^\prime))
\]
where $\widetilde{G}(f) = \left.G(f)\right|_\Im$. Hence $\omega_{2,M}$ can be read as the distribution obtained composing $\omega_{2,\Im}$ with the map 
$\widetilde{G}\otimes \widetilde{G}$. The desired inclusion descends by using Theorem \ref{thm:kern}.
In order to follow this strategy successfully, one needs to evaluate the wave front set of $\omega_{2,\Im}$ and that of
$\widetilde{G}\otimes \widetilde{G}$ knowing \eqref{eq:WFprove2}. In tackling this problem one can observe in addition that only the singularities of $\omega_{2,\Im}$ present in a compact region on $\Im\times \Im$ can influence the wave front set of $\left.\omega_{2,M}\right|_{O\times O}$ where $O$ is an open set of $M$. This is a consequence of the fact that the null geodesics emanating from $O$ intersect $\Im$ in a compact region $C$. Hence, considering a partition of unity $\chi$ adapted to $C$ it holds that   
\[
WF(\left.\omega_{2,M}\right|_{O\times O})  =  WF( \omega_{2,\Im} \circ \chi \widetilde{G} \otimes \chi\widetilde{G} )
\]
At this stage Theorem \ref{thm:kern} yields that $WF(\left.\omega_{2,M}\right|_{O\times O})$ satisfies \eqref{eq:WFsubset}. Since $O$ is generic the sought result descends. \qed
\end{sproof}

\noindent Before considering the next case, two comments on the relevance of $\omega_M$ are in due course  \cite{Moretti,Moretti2}:
\begin{enumerate}
	\item Consider any complete Killing field $\xi$ of $(M,g)$, whose unique extension to $\Im^+$ as per Proposition \ref{Prop:Extension_Isom} is indicated with $\widetilde{\xi}$. Let $\chi^\xi_t$ and $\widetilde{\chi}^{\widetilde{\xi}}_t$ be respectively the one-parameter group of isometries associated to $\xi$ and the one-parameter subgroup of $G_{BMS}$ associated to $\widetilde{\xi}$, $t\in\mathbb{R}$ and consider the representations of these symmetries in terms of $^*$-automorphisms on the respective algebras as in 
	Proposition \ref{prop:isometries_bulk_to_boundary_homomorphism}. As a direct consequence of Proposition \ref{prop:isometries_bulk_to_boundary_homomorphism}, we can conclude that, for every $a\in\mathcal{A}_0^{obs}(M)$, 
	$$\omega_M(\alpha_{\chi^\xi_t}(a))=\omega_\Im(\iota_\Im(\alpha_{\chi^\xi_t}(a)))=\omega_\Im(\alpha_{\widetilde{\chi}_t^{\widetilde{\xi}}}(\iota_\Im(a)),$$
	where we used \eqref{eq:bulk_state_asymptotically_flat} and, in the last equality, \eqref{eq:interplay_algebras}. Applying Theorem \ref{Th:uniqueness_asymptotically_flat}, it descends, for all $a\in\mathcal{A}_0^{obs}(M)$ 
	\begin{equation}\label{eq:state_invariance_asymptotically_flat}
	\omega_M(\alpha_{\chi^\xi_t}(a))=\omega_\Im(\iota_\Im(a))=\omega_M(a),
	\end{equation}
	which entails that {\em $\omega_M$ is always invariant under the action of all bulk isometries}.
	\item As a further consequence of the last remark, we can apply our construction to Minkowski spacetime.  Hence, in view of \eqref{eq:state_invariance_asymptotically_flat}, the state $\omega_M$ must coincide with the Gaussian {\em Poincar\'e}-invariant  state which is the unique, ground state, invariant under all  orthochronous proper Poincar\'e isometries.
\end{enumerate} 
\subsection{Cosmological Spacetimes} Let us next consider $(M,g_{FRW})$, a four-dimensional, simply connected Friedamann-Robertson-Walker spacetime with flat spatial sections, fulfilling the hypotheses of Theorem \ref{theorem1}. This is globally hyperbolic and it can be extended to a larger, globally hyperbolic spacetime $(\widetilde{M},\widetilde{g})$ so that $\partial M\subset\widetilde{M}$ is the cosmological horizon of $M$ as per Definition \ref{defexp}. On top of $(M,g_{FRW})$ we consider a generic Klein Gordon field $\phi:M\to\mathbb{R}$ fulfilling \eqref{eq:dynamics}. The mass and the coupling to the scalar curvature are constrained in such a way that the hypotheses of Theorem \ref{th:cosmological_bulk_to_boundary_compatibility} are met. 
The associated $^*$-algebra of observables is $\mathcal{A}^{obs}(M)$ as per Definition \ref{def:algebra_of_observables} and there exists an injective $^*$-isomorphism, defined in \eqref{eq:cosmological_bulk_to_boundary_homomorphism} $\iota_c:\mathcal{A}^{obs}(M)\to\mathcal{A}_c(\Im)$. The latter is the $*$-algebra from Definition \ref{Def:boundary_algebra} built out of \eqref{eq:cosmological_boundary_functions}. 
We finally  consider the algebraic, quasi-free state $\omega_\Im$ on $\mathcal{A}_c(\Im^-)$ built in Definition \ref{Def:boundary_quasi_free_state} and, by applying Definition \ref{Def:induced_state} we introduce the pulled-back state $\omega_{c,M}:\mathcal{A}^{obs}(M)\to\mathbb{C}$
\begin{equation}\label{eq:bulk_state_cosmological} \omega_{c,M} \doteq\iota_c^*\omega_\Im\:.
\end{equation}
It is immediate to realize that $\omega_{c,M}$ inherits the property of being quasi-free and, starting from $\omega_{\Im}$ as in Definition \ref{Def:boundary_quasi_free_state}, the associated two-point function in the sense of \eqref{eq:correlation_functions} reads
\begin{equation}\label{eq:bulk_two_point_cosmological}
\omega_{2,c}(f,f^\prime)=\omega_{\Im}([\Gamma_c([f])][\Gamma_c([f^\prime])]),
\end{equation}
where $\Gamma_c$ is defined in \eqref{eq:cosmological_bulk_to_boundary}, while  $f,f^\prime\in C^\infty_0(M)$. The next theorem further characterizes the properties of $\omega_{c,M}$:

\begin{theorem}\label{Th:Hadamard_form_Cosmological}
	The state $\omega_{c,M}:\mathcal{A}^{obs}(M)\to\mathbb{C}$  defined in \eqref{eq:bulk_state_cosmological} is of {\bf Hadamard form} in the sense of Definition \ref{Def:local_Hadamard_form}.
\end{theorem}
This theorem has been proven in \cite{Dappiaggi:2008dk}, to which we refer for the details of the proof. 

\vskip .2cm

\noindent To conclude this section, two comments  \cite{Dappiaggi:2008dk}  on the relevance of $\omega_{c,M}$ are in due course:
\begin{enumerate}
	\item Consider any complete Killing field $\xi$ of $(M,g)$, preserving $\Im^-$ as per Proposition \ref{togroup}, whose unique extension to $\Im^-$ is indicated with $\widetilde{\xi}$. Let $\chi^\xi_t$ and $\widetilde{\chi}^{\widetilde{\xi}}_t$ be respectively the bulk and $\Im^-$, one-parameter group of isometries associated to $\xi$ and $\widetilde{\xi}$, $t\in\mathbb{R}$ and consider the corresponding 
representations in terms of $^*$-algebras introduced in 	 Proposition \ref{prop:cosmological_isometries_bulk_to_boundary_homomorphism}.
 As a direct consequence we can conclude that, for every $a\in\mathcal{A}^{obs}(M)$, 
	$$\omega_{c,M}(\widetilde{\alpha}_{\chi^\xi_t}(a))=\omega_\Im(\iota_c(\widetilde{\alpha}_{\chi^\xi_t}(a)))=\omega_\Im(\widetilde{\alpha}_{\widetilde{\chi}_t^{\widetilde{\xi}}}(\iota_c(a)),$$
	where we used \eqref{eq:bulk_state_cosmological} and, in the last equality, \eqref{eq:cosmological_interplay_algebras}. Applying Theorem \ref{Th:uniqueness_cosmological}, it descends, for all $a\in\mathcal{A}^{obs}(M)$ 
	\begin{equation}\label{eq:state_invariance_cosmological}
	\omega_{c,M}(\widetilde{\alpha}_{\chi^\xi_t}(a))=\omega_\Im(\iota_c(a))=\omega_{c,M}(a),
	\end{equation}
which entails that $\omega_{c,M}$ is invariant under the action of all bulk isometries, preserving $\Im^-$ in the sense of Proposition \ref{togroup}.
\item As a further consequence of the last remark, we can apply our construction to the cosmological de-Sitter spacetime, that is the scale factor in \eqref{metric} is $a(\tau)=-\frac{1}{H\tau}$, $H>0$ and $\tau\in(-\infty,0)$. In this case all spacetime isometries preserve the cosmological horizon $\Im^-$. Hence, in view of \eqref{eq:state_invariance_cosmological}, the state $\omega_{c,M}$ must coincide with the {\em Bunch-Davies} state \cite{BD} which is the unique quasifree ground state, invariant under all de Sitter isometries.
\end{enumerate}

\subsection{The Unruh state on Schwarzschild spacetime}\label{Sec:Unruh_State}

As next step of our investigation, we consider $\mathcal{M}$, the physical portion of Kruskal spacetime, an asymptotically flat spacetime at spatial infinity which needs to be analysed separately exactly as in Section \ref{Sec:Schwarzschild}. We recall that, on top of \eqref{Schw}, we consider a massless, real scalar field whose dynamics is thus ruled by \eqref{eq:conformally_coupled} and whose associated algebra of observables is $\mathcal{A}^{obs}_0(M)$. The latter is built as per Definition \ref{def:algebra_of_observables}.
On account of the conformal structure of the manifold under consideration, one cannot encode the information of $\mathcal{A}^{obs}_0(M)$ in a single ancillary counterpart living on a codimension $1$ submanifold. As shown in Proposition \ref{prop:Schwarzschild_bulk_to_boundary} and \eqref{eq:Schwarzschild_*_homomorphism}, there exists an injective $*$-homomorphism $\iota_{\mathcal{H},\Im}:\mathcal{A}^{obs}_0(M)\to\mathcal{A}_{\Im^-,\mathcal{H}}$ where the latter is the $^*$-algebra $\mathcal{A}(\Im^-) \otimes \mathcal{A}(\mathcal{H})$ as per  \eqref{eq:Schwarzschil_full_boundary_algebra}, \eqref{eq:boundary_space_H} and \eqref{eq:boundary_space_Im}.

In order to use the procedure of Definition \ref{Def:induced_state} defining a state for $\mathcal{A}^{obs}_0(M)$, we need to assign a suitable counterpart to both $\mathcal{A}(\mathcal{H})$ and $\mathcal{A}(\Im^-)$. We start from the latter and we follow again Definition \ref{Def:boundary_quasi_free_state} and Lemma \ref{lem:boundary_integral_kernel} replacing $\mathcal{A}(\Im)$ with $\mathcal{A}(\Im^-)$ generated by $\mathcal{S}(\Im^-)$ as in \eqref{eq:boundary_space_Im}. 

\vskip .3cm

\noindent\textbf{\em The state on $\Im^-$:} Let us endow $\Im^-$ with the coordinates $(v,\theta,\varphi)$, $v$ being defined in \eqref{two} and let us consider the quasi-free state $\omega_{\Im^-}$ whose two-point function reads as in \eqref{eq:notable_2pt_function} with $u$ replaced by $v$. Before proceeding one should observe that, in view of \eqref{eq:boundary_space_Im}, one cannot apply Proposition \ref{prop:boundary_one_particle_space} slavishly since the decay rate of the elements of $\mathcal{S}(\Im^-)$ are not sufficient to apply the Fourier-Plancherel transform. Hence it is unclear why \eqref{eq:notable_2pt_function} is reliable and, in turn, which is the associated one-particle Hilbert space obtained out of the GNS construction of $\omega_{\Im^-}$. 

To overcome this hurdle, consider an auxiliary coordinate $x$ whose domain is $\mathbb{R}$ and such that $x=\sqrt{v}$ if $v\geq 0$, while $x=-\sqrt{-v}$ if $v\leq 0$. We define the space $H^1(\Im^-)_x$ as the collection of functions $\psi\in L^2(\mathbb{R}\times\mathbb{S}^2,\,dx\,d\mu_{\mathbb{S}^2})$ together with their first distributional derivative along the $x$-direction. The following proposition has been proven in \cite{Dappiaggi:2009fx}:

\begin{proposition}\label{Prop:Aux_Schwarzschild}
	Let $\mathbb{R}^*_-\doteq (-\infty,0)$ and let  $\lambda_{\Im^-}:C^\infty_0(\Im^-;\mathbb{C})\times C^\infty_0(\Im^-;\mathbb{C})\to\mathbb{C}$ be the Hermitian inner product:
	\begin{equation}\label{eq:Hermitian_inner_product_Im}
	\lambda_{\Im^-}(f_1,f_2)=-\frac{1}{\pi}\lim_{\epsilon\to 0^+}\int\limits_{\mathbb{R}\times\mathbb{R}\times\mathbb{S}^2}\frac{f_1(v,\theta,\varphi)\overline{f_2(v^\prime,\theta,\varphi)}}{(v-v^\prime-i\epsilon)^2} dv\,dv^\prime\,d\mu_{\mathbb{S}^2}\:;
	\end{equation}
	where $f_1,f_2\in C^\infty_0(\Im^-;\mathbb{C})$. Then the following holds:
	
	\begin{enumerate}
	\item Every sequence $\{\psi_n\}_{n\in \mathbb{N}} \subset C_0^\infty(\mathbb{R}^*_- \times \mathbb{S}^2; \mathbb{R})$ such 
	$\psi_n \to \psi$ as $n\to +\infty$ in  $H^1(\Im^-)_x$
	is of Cauchy type in $\overline{\left(C_0^\infty(\Im^-; \mathbb{C}), \lambda_{\Im^-} \right)}$, the Hilbert space completion of $C^\infty_0(\Im^-;\mathbb{C})$ with respect to $\lambda_{\Im^-}$.
	
	\item There is $\{\psi_n\}_{n\in \mathbb{N}} \subset C_0^\infty(\mathbb{R}^*_- \times \mathbb{S}^2; \mathbb{R})$ such 
	$\psi_n \to \psi$ as $n\to +\infty$ in $H^1(\Im^-)_x$ and, 
	if $\{\psi'_n\}_{n\in \mathbb{N}} \subset C_0^\infty(\mathbb{R}^*_- \times \mathbb{S}^2; \mathbb{R})$ converges to the same $\psi$
	in  $H^1(\Im^-)_x$, then $\psi'_n-\psi_n\to 0$ in 
	$\overline{\left(C_0^\infty(\Im^-; \mathbb{C}), \lambda_{\Im^-} \right)}$.
	
	\item If $\widehat{\psi}_+(k,\theta,\phi) \doteq \mathcal{F}(\psi)|_{\{k\geq 0, (\theta,\phi) \in 
		\mathbb{S}^2\}}(k,\theta,\phi)$ 
	denotes the $v$-Fourier transform  of
	$\psi \in  C_0^\infty(\Im^-; \mathbb{C})$
	restricted to $k\in  \mathbb{R}_+$, the map 
	$$C_0^\infty(\Im^-; \mathbb{C}) \ni \psi \mapsto \widehat{\psi}_+(k,\theta,\phi) \in 
	L^2(\mathbb{R}_+\times \mathbb{S}^2, 2kdk d\mu_{\mathbb{S}^2})\doteq H_{\Im^-}$$ is isometric, extending uniquely, per continuity, to a Hilbert space isomorphism
	\begin{equation}\label{Fu} 
	F_{(v)}:  \overline{\left(C_0^\infty(\Im^-; \mathbb{C}), \lambda_{\Im^-} \right)} \to
	H_{\Im^-}\:.
	\end{equation}
	
	\item If one replaces $\mathbb{C}$ with $\mathbb{R}$:
	\begin{equation} \label{scriaggdens}
	\overline{F_{(v)}\left(C_0^\infty(\Im^-; \mathbb{R}) \right)} = H_{\Im^-}\:. 
	\end{equation}
	\end{enumerate}

\end{proposition}

\noindent The next lemma further elaborates on the consequences of this last proposition:

\begin{lemma}\label{Lem:further}
 Let $\psi \in\mathcal{S}(\Im^-)$ be such that $\textrm{supp}(\psi) \subset \mathbb{R}^*_- \times \mathbb{S}^2$. Then $\psi$ can be naturally identified with a corresponding element of $\overline{\left(C_0^\infty(\Im^-; \mathbb{C}), \lambda_{\Im^-} \right)}$,
	indicated with the same symbol. In addition, 
	\begin{equation}\label{FumF}
	F_{(v)}|_{\mathcal{S}(\Im^-)} = \Theta(h)\cdot \mathcal{F}|_{\mathcal{S}(\Im^-)}\:, 
	\end{equation}
	and, for $\psi,\psi'\in \mathcal{S}(\Im^-)$,
	\begin{equation}
	\lambda_{\Im^-}(\psi,\psi') = \int_{\mathbb{R}_+\times \mathbb{S}^2} \overline{\mathcal{F}(\psi')}\left( \mathcal{F}(\psi)(h,s)+\overline{\mathcal{F}(\psi)(h,s)}\right)\: 2hdh d\mu_{\mathbb{S}^2}(s)\:,\label{antilin}
	\end{equation}
	where, $\Theta(h)=0$ if $h\leq 0$ and $\Theta(h)=1$ otherwise. Here
	$\mathcal{F} : L^2(\mathbb{R} \times \mathbb{S}^2, dx d\mu_{\mathbb{S}^2}) \to L^2(\mathbb{R} \times \mathbb{S}^2, dk d\mu_{\mathbb{S}^2})$
	is the $x$-Fourier-Plancherel transform\footnote{We employ the symbol $h$ for the coordinate conjugate to $x$, to disambiguate it from the one conjugate to $v$, which is indicated with $k$.}.
\end{lemma}
\noindent Following closely Proposition \ref{Prop:One_particle_Hilbert}, the next step is the construction of a suitable $\mathbb{R}$-linear  map $K_{\Im^-}:\mathcal{S}(\Im^-)\to H_{\Im^-}$.
Let $\eta\equiv\eta(v)\in C^\infty(\mathbb{R}^*_- \times \mathbb{S}^2)$ be a smooth, non negative function such that 
$\eta=1$ for all $v$ smaller than an arbitrary, but fixed $v_0$. 
Decompose any $\psi\in\mathcal{S}(\Im^-)$ as 
\begin{equation}\label{dec0scri}
\psi = \psi_0 + \psi_-\:, \mbox{where $\psi_0 = (1-\eta)\psi$} 
\end{equation}
It holds $\psi_0 \in C_0^\infty(\Im^-)$ and
$\textrm{supp}(\psi_-) \subset \mathbb{R}^*_-\times \mathbb{S}^2$, where $\mathbb{R}^*_-$ is referred to the coordinate $v$ on $\mathbb{R}$. 
In addition, following \eqref{FumF}, let
\begin{equation}
K_{\Im^-}(\psi) \doteq F_{(v)}(\psi_0) +  F_{(v)}(\psi_-)\:, \quad \forall \psi \in \mathcal{S}(\Im^-)\:,
\label{mapidscri}
\end{equation}
where, $\psi_-\in\overline{\left(C_0^\infty(\Im^-; \mathbb{C}), 
	\lambda_{\Im^-} \right)}$ on account of Proposition \ref{Prop:Aux_Schwarzschild}. 
The $\mathbb{R}$-linear map $K_\Im : \mathcal{S}(\Im^-) \to H_\Im^-$ is continuous when the domain is equipped with the norm 
\begin{equation}\label{normeScri}
\| \psi \|_{\Im^-}^\eta = \|  \psi_- \|_{H^1(\Im^-)_x} + \| \psi_0  \|_{H^1(\Im^-)_v}
\end{equation}
where
$\|  \cdot  \|_{H^1(\Im^-)_x}$ and $\|  \cdot  \|_{H^1(\Im^-)_v}$ are the norms of the Sobolev spaces 
$H^1(\Im^-)_x$ and $H^1(\Im^-)_v$ respectively, the latter hence with respect to the $v$-coordinate. 
Let us remark that different $\eta$ and $\eta'$ produce equivalent norms $\|\cdot\|^\eta_{\Im^-}$ and $\|\cdot\|^{\eta'}_{\Im^-}$. We shall drop, therefore the index $\eta$. 

\begin{proposition}\label{Prop:well_posed_K_Im}
The linear map $K_{\Im^-} : \mathcal{S}(\Im^-) \to H_{\Im^-}$ in (\ref{mapidscri}) enjoys the following properties:
\begin{enumerate}
\item It is independent from the chosen decomposition 
	(\ref{dec0scri}) for a fixed $\psi \in \mathcal{S}(\Im^-)$,
	
\item It reduces to $F_{(v)}$ when restricting to $C_0^\infty(\Im^-; \mathbb{R})$,
	
\item It satisfies: 
	$$\sigma_{\Im}(\psi,\psi') = 2 Im \langle K_{\Im^-}(\psi), K_{\Im^-}(\psi') \rangle_{H_{\Im^-}}\:, 
	\quad \mbox{if $\psi,\psi' \in \mathcal{S}(\Im^-)$,}
	$$
where $\sigma_{\Im}$ is the symplectic form \eqref{eq:boundary_symplectic_form} for $\Im^-$,
	
\item It is injective and it holds $\overline{K_{\Im^-}(\mathcal{S}(\Im^-))} = H_{\Im^-}$
	
\item It is continuous with respect to the norm $\|\cdot\|_{\Im^-}$ defined in \eqref{normeScri} for every choice of the function $\eta$. Consequently, there exists a constant $C>0$ such that:
	$$
	|\langle K_{\Im^-}(\psi), K_{\Im^-}(\psi') \rangle_{H_{\Im^-}}|\leq  C^2 \| \psi \|_{\Im^-} \cdot \| \psi' \|_{\Im^-} \;
	\quad \mbox{if $\psi,\psi' \in \mathcal{S}(\Im^-)$.}
	$$
\end{enumerate} 	
\end{proposition}
\noindent We can revert now our attention to the quasi-free state on the algebra constructed in Definition \ref{Def:boundary_algebra}, starting from \eqref{eq:boundary_space_Im} and showing that it enjoys several notable properties. In the following we report the main ones. The proof of this proposition can be found in \cite[Th. 3.2]{Dappiaggi:2009fx}:

\begin{proposition}\label{Prop:State_on_past_null_infinity}
	Let $\omega_{\Im^-}:\mathcal{A}(\Im^-)\to\mathbb{C}$ be the quasi-free state constructed in Definition \ref{Def:boundary_algebra} starting from \eqref{eq:boundary_space_Im}. It holds that
	\begin{enumerate}
		\item  The pair $(H_{\Im^-},K_{\Im^-})$ is the one-particle structure for $\omega_{\Im^-}$, which is uniquely determined by the requirement that its two-point function coincides with the right-hand side of 
		(\ref{eq:Hermitian_inner_product_Im}) under the restriction to $C_0^\infty(\Im^-; \mathbb{R})$.
		\item $\omega_{\Im^-}$ is invariant under the one-parameter group of $^*$-automorphisms generated by the Killing vectors of $\mathbb{S}^2$ and of $X|_{\Im^-}$, $X$ coinciding with $\partial_t$ in both $\mathcal{W}$ and $\mathcal{B}$.
	\end{enumerate}
\end{proposition}

\vskip .3cm

\noindent\textbf{\em The state on the horizon $\mathcal{H}$:} In order to define a quasi-free state on $\mathcal{A}(\mathcal{H})$ the situation is very similar to the preceding case both with respect to the problem that we face and to the techniques that we employ. In addition this case has been already thoroughly investigated in \cite{KW}.

As starting point, let us consider $\mathcal{H}$ endowed with the coordinates $(U,\theta,\varphi)$, $U$ being defined in \eqref{UV}. We endow  $C^\infty_0(\mathcal{H};\mathbb{C})$ with the Hermitian inner product 
\begin{equation}\label{KW-inner_product}
\lambda_{\mathcal{H}}(f,f^\prime)=-\frac{4M^2}{\pi}\lim_{\epsilon\to 0^+} \int\limits_{\mathbb{R}\times\mathbb{R}\times\mathbb{S}^2}\frac{f(U,\theta,\varphi)\overline{f^\prime(U^\prime,\theta,\varphi)}}{(U-U^\prime-i\epsilon)^2}dU\,dU^\prime\,d\mu_{\mathbb{S}^2}(\theta,\varphi)\:,
\end{equation}
where $f,f^\prime\in C^\infty_0(\mathcal{H};\mathbb{C})$, while $M$ is the mass of Schwarzschild black hole. This allows to construct per completion a Hilbert space whose property are listed in the next proposition proven in \cite{Dappiaggi:2009fx}:

\begin{proposition}\label{prop:horizon_Hilbert_space}
	Let $\overline{\left(C_0^\infty(\mathcal{H}; \mathbb{C}), \lambda_{\mathcal{H}} \right)}$ be the Hilbert completion
	of the complex vector space $C_0^\infty(\mathcal{H}; \mathbb{C})$ with respect to \eqref{KW-inner_product}. Let 
	$\widehat{\psi}_+\doteq\mathcal{F}(\psi)|_{\{K\geq 0,(\theta,\phi) \in \mathbb{S}^2\}}$ be the restriction to positive values of $K$ of the $U$-Fourier transform of $\psi \in C_0^\infty(\mathcal{H}; \mathbb{C})$. Then
	
	\begin{enumerate}
		\item The linear map 
	$$C_0^\infty(\mathcal{H}; \mathbb{C}) \ni \psi \mapsto \widehat{\psi}_+(K,\theta,\phi) \in 
	L^2(\mathbb{R}_+\times \mathbb{S}^2, 2KdK r_S^2d\mu_{\mathbb{S}^2}) \doteq H_{\mathcal{H}}$$ is isometric and uniquely 
	extends, by linearity and continuity, to a Hilbert space isomorphism of 
	$$
	F_{(U)}:  \overline{\left(C_0^\infty(\mathcal{H}; \mathbb{C}), \lambda_{KW} \right)} \to
	H_{\mathcal{H}}\:.
	$$
	
	\item If one switches to $\mathbb{R}$ in place of $\mathbb{C}$
	$$
	\overline{F_{(U)}\left(C_0^\infty(\mathcal{H}; \mathbb{R}) \right)} = H_{\mathcal{H}}\:. 
	$$
\end{enumerate}
\end{proposition}
\noindent As in the previous case we cannot embed in $H_{\mathcal{H}}$ the space $\mathcal{S}(\mathcal{H})$, defined in \eqref{eq:boundary_space_H} since the decay rate in $U$ is too slow to apply directly the Fourier-Plancherel transform. Hence, let us split $\mathcal{H}$ as $\mathcal{H}^+\cup B\cup\mathcal{H}^-$ where $\mathcal{H}^\pm$ stands for the collection of points on the horizon with $U>0$ (+) and with $U<0$ (-) while $B$ is the bifurcation surface at $U=0$. It is convenient to introduce on $\mathcal{H}^+$ the coordinate $u\doteq 4M\ln(U)$ and on $\mathcal{H}^-$, the coordinate $u\doteq 4M\ln(-U)$ where the same symbol is used since no confusion can arise. Let $H^1(\mathcal{H}^\pm)_u$ denote the space of complex valued functions on $\mathcal{H}^\pm$ which are, together with their first distributional derivative along $u$, square-integrable with respect to the measure $du\,d\mu_{\mathbb{S}^2}(\theta,\varphi)$. In addition let us define:
\begin{gather}
 \mathcal{S}(\mathcal{H}^\pm) \doteq \left\{\psi \in C^\infty(\mathbb{R} \times \mathbb{S}^2; \mathbb{R})\:\left| 
\: \psi(u) = 0 \;\mbox{in a neighbourhood of $B$ and }\right.\right. \notag \\
\exists\; C_\psi, C^\prime_\psi\geq 0\notag
\mbox{  with  }
  |\psi(u,\theta,\phi)| < \frac{C_\psi}{1+|u|} \:,\notag\\
\left. \left|\partial_u\psi(u,\theta,\phi)\right|<
\frac{C'_\psi}{1+|u|}\:,  (u, \theta,\phi) \in \mathbb{R}\times \mathbb{S}^2  \right\}\:,\label{SHpm}
\end{gather}  
where $C_\psi,C^\prime_{\psi}\geq 0$. One can realize by direct inspection that $\mathcal{S}(\mathcal{H}^\pm)\subset H^1(\mathcal{H}^\pm)$, which leads to the following proposition:

\begin{proposition}\label{Prop:properties_horizon}
Let $d\mu(k)\doteq 8M^2\frac{k e^{4\pi M k}}{e^{4\pi M k}-e^{-4\pi M k}}  dk$. The following holds true:
\begin{enumerate}
	\item If $\widetilde{\psi} = (\mathcal{F}(\psi))(k,\theta,\phi) = \widetilde{\psi}(k,\theta,\phi)$ denotes the
$u$-Fourier transform of either $\psi \in C_0^\infty(\mathcal{H}^+; \mathbb{C})$ or  $\psi \in C_0^\infty(\mathcal{H}^-; \mathbb{C})$
the maps 
$$C_0^\infty(\mathcal{H}^\pm; \mathbb{C}) \ni \psi \mapsto \widetilde{\psi} \in 
L^2(\mathbb{R} \times \mathbb{S}^2, d\mu(k) d\mu_{\mathbb{S}^2})$$ are isometric when $C_0^\infty(\mathcal{H}^\pm; \mathbb{C})$
is equipped with the scalar product $\lambda_{\mathcal{H}}$. It uniquely 
extends, per continuity, to the Hilbert space isomorphisms: 
\begin{equation} \label{Fv}
F^{(\pm)}_{(u)} : \overline{C_0^\infty(\mathcal{H}^\pm; \mathbb{C})}\to L^2(\mathbb{R} \times \mathbb{S}^2, d\mu(k) d\mu_{\mathbb{S}^2})\:, 
\end{equation} 
where 
$\overline{C_0^\infty(\mathcal{H}^\pm; \mathbb{C})}$ are viewed as Hilbert subspaces of 
$\overline{\left(C_0^\infty(\mathcal{H}; \mathbb{C}), \lambda_{\mathcal{H}} \right)}$.

\item The spaces $\mathcal{S}(\mathcal{H}^\pm)$  are naturally identified with
real subspaces of $\overline{C_0^\infty(\mathcal{H}; \mathbb{C})}$ as follows:
If either $\{\psi_n\}_{n\in \mathbb{N}}, \{\psi'_n\}_{n\in \mathbb{N}} \subset C_0^\infty(\mathcal{H}^+; \mathbb{R})$
or $\{\psi_n\}_{n\in \mathbb{N}}, \{\psi'_n\}_{n\in \mathbb{N}} \subset C_0^\infty(\mathcal{H}^-; \mathbb{R})$ and, according to the case,
both sequences 
$\{\psi_n\}_{n\in \mathbb{N}}, \{\psi'_n\}_{n\in \mathbb{N}}$
converge to the same $\psi \in \mathcal{S}(\mathcal{H}^\pm)$ in $H^1(\mathcal{H}^\pm)$, 
then  both sequences are of Cauchy type in
$\overline{\left(C_0^\infty(\mathcal{H}; \mathbb{C}), \lambda_{KW} \right)}$ 
and 
$\psi_n-\psi'_n \to 0$ in $\overline{\left(C_0^\infty(\mathcal{H}; \mathbb{C}), \lambda_{KW} \right)}$. \\
The subsequent identification of $\mathcal{S}(\mathcal{H}^\pm)$  with
real subspaces of $\overline{C_0^\infty(\mathcal{H}; \mathbb{C})}$ is such that:
\begin{equation}
F^{(\pm)}_{(u)}|_{\mathcal{S}(\mathcal{H}^\pm)} = \mathcal{F}|_{\mathcal{S}(\mathcal{H}^\pm)}\:,\label{FF}
\end{equation}
where 
$\mathcal{F}:  L^2(\mathbb{R} \times \mathbb{S}^2, du d\mu_{\mathbb{S}^2}) \to L^2(\mathbb{R} \times \mathbb{S}^2, dk d\mu_{\mathbb{S}^2})$
stands for the standard $u$-Fourier-Plancherel transform.
\end{enumerate}
\end{proposition}

\noindent Following closely Proposition \ref{Prop:One_particle_Hilbert}, the next step is the construction of a suitable map $K_{\mathcal{H}}:\mathcal{S}(\mathcal{H})\to H_{\mathcal{H}}$. Let us thus consider $\chi\in C^\infty_0(\mathcal{H})$, so to decompose every $\psi\in\mathcal{S}(\mathcal{H})$ as
\begin{equation}\label{+_-_0}
\psi=\psi_++\psi_0+\psi_-,
\end{equation}
where $\psi_0\doteq\chi\psi$, while $\psi^\pm=(1-\chi)\psi|_{\mathcal{H}^\pm}$. We can now define the map 
\begin{equation}\label{K_on_horizon}
K_{\mathcal{H}}:\mathcal{S}(\mathcal{H})\mapsto H_{\mathcal{H}},\quad\psi\mapsto K_{\mathcal{H}}\psi\doteq F_{(u)}\psi_+ +F_{(U)}\psi_0+ F_{(u)}\psi_-,
\end{equation}
where the map $F_{(U)}$ is well-defined in view of item $2.$ in Proposition \ref{Prop:properties_horizon}. Furthermore we prove that, $K_{\mathcal{H}} : \mathcal{S}(\mathcal{H}) \to H_{\mathcal{H}}$ is continuous if
we read  $\mathcal{S}(\mathcal{H})$ as a normed space equipped with the norm 
\begin{equation}\label{norme}
\| \psi \|_{\mathcal{H}}^\chi = \| (1-\chi) \psi \|_{H^1(\mathcal{H}^-)_u} + \| \chi \psi  \|_{H^1(\mathcal{H})_U}
+
\|  (1-\chi)\psi \|_{H^1(\mathcal{H}^+)_u}
\end{equation}
where
$\|  \cdot  \|_{H^1(\mathcal{H}^\pm)_u}$ and $\|  \cdot  \|_{H^1(\mathcal{H})_U}$ are the norms of the Sobolev spaces $H^1(\mathcal{H}^\pm)_u$ and
$H^1(\mathcal{H})_U$ respectively. Observe that, $\|\cdot\|^\chi_{\mathcal{H}}$ and $\|\cdot\|^{\chi'}_{\mathcal{H}}$, defined with respect of different 
decompositions generated by $\chi$ and $\chi'$, are equivalent. Due to such  equivalence, we will often write $\| \psi \|_{\mathcal{H}}$ in place of $\| \psi \|_{\mathcal{H}}^\chi$. The following proposition characterizes completely the properties of the map $(\mathcal{S}(\mathcal{H}),H_{\mathcal{H}})$ whose proof is in \cite{Dappiaggi:2009fx}:

\begin{proposition}
	The $\mathbb{R}$-linear map $K_{\mathcal{H}} : \mathcal{S}(\mathcal{H}) \to H_{\mathcal{H}}$ in \eqref{K_on_horizon} enjoys the following properties:
	
	\begin{enumerate}
		\item It is independent from the choice of the function $\chi$ used in the decomposition \eqref{+_-_0} of $\psi \in \mathcal{S}(\mathcal{H})$;
	
	\item It reduces to $F_{(U)}$ when restricting to $C_0^\infty(\mathcal{H}; \mathbb{R})$;
	
	\item It satisfies 
	\begin{equation}\label{8anni_fa_fu_chiamata_ultimabastarda1}
	\sigma_{\mathcal{H}}(\psi,\psi') = 2 Im \langle K_{\mathcal{H}}(\psi), K_{\mathcal{H}}(\psi')
	\rangle_{H_{\mathcal{H}}}\:, \quad \mbox{if $\psi,\psi' \in \mathcal{S}(\mathcal{H})$;}	
	\end{equation}
	
	\item It is injective and it holds $\overline{K_{\mathcal{H}}(\mathcal{S}(\mathcal{H}))} = H_{\mathcal{H}}$;
	
	\item It is continuous with respect to the norm $\|\cdot\|_{\mathcal{H}}$ defined in \eqref{norme} for every 
	choice of the function  $\chi$. Consequently, there exists $C>0$ such that
	$$
	|\langle K_{\mathcal{H}}(\psi), K_{\mathcal{H}}(\psi')
	\rangle_{H_{\mathcal{H}}}| \leq  C^2 \| \psi \|_{\mathcal{H}} \cdot \| \psi' \|_{\mathcal{H}} \;
	\quad \mbox{if $\psi,\psi' \in \mathcal{S}(\mathcal{H})$.}
	$$ 
\end{enumerate}
\end{proposition}

\noindent We can revert now our attention to the quasi-free state on the algebra constructed in Definition \ref{Def:boundary_algebra} starting from \eqref{eq:boundary_space_H} showing that it enjoys several notable properties. In the following we report the main ones. The proof of this proposition can be found in \cite[Th. 3.1]{Dappiaggi:2009fx}:

\begin{proposition}\label{Prop:state_on_the_horizon}
Let $\omega_{\mathcal{H}}:\mathcal{A}(\mathcal{H})\to\mathbb{C}$ be the quasi-free state constructed in Definition \ref{Def:boundary_quasi_free_state} starting from \eqref{eq:boundary_space_H}. It holds that:
\begin{enumerate}
	\item The pair $(H_{\mathcal{H}},K_{\mathcal{H}})$ is the one-particle structure for $\omega_{\mathcal{H}}$, which is uniquely individuated by the requirement that its two-point function coincides to the right-hand side of 
(\ref{KW-inner_product}) under restriction to $C_0^\infty(\mathcal{H}; \mathbb{R})$.

\item The state $\omega_{\mathcal{H}}$ is invariant under the natural action of the one-parameter group 
of $*$-automorphisms generated by $X|_{\mathcal{H}}$ and of those generated by the Killing vectors of $\mathbb{S}^2$. Here $X$ coincides with $\partial_t$ in both $\mathcal{W}$ and $\mathcal{B}$.

\item The restriction of $\omega_{\mathcal{H}}$ to $\mathcal{A}(\mathcal{S}(\mathcal{H}^\pm))$, the algebra built out of \eqref{SHpm}, is a quasifree state
$\omega^{\beta_H}_{\mathcal{H}^\pm}$  individuated by the one particle structure
$(H^{\beta_H}_{\mathcal{H}^\pm},K^{\beta_H}_{\mathcal{H}^\pm})$
$$
H^{\beta_H}_{\mathcal{H}^\pm} \doteq L^2(\mathbb{R}\times\mathbb{S}^2, d\mu(k)  d\mu_{\mathbb{S}^2})\quad\textrm{and}\quad \mbox{$K^{\beta_H}_{\mathcal{H}^
		\pm} \doteq \mathcal{F}|_{\mathcal{S}(\mathcal{H}^\pm)} = F^{(\pm)}_{(u)}|_{\mathcal{S}(\mathcal{H}^\pm
		)}$.}
$$

\item The states $\omega^{\beta_H}_{\mathcal{H}^\pm}$ satisfy the KMS condition with respect to the one-parameter group of $^*$-automorphisms generated by, respectively, $\mp X|_{\mathcal{H}}$, with Hawking inverse temperature
$\beta_H = 8\pi M$.
\end{enumerate}
\end{proposition}

\vskip .3cm

\noindent\textbf{\em The Unruh state:} We can combine together the states defined at $\Im^-$ and at the horizon recalling that the full boundary $*$-algebra of observables is $\mathcal{A}_{\Im^-,\mathcal{H},}$, defined in \eqref{eq:Schwarzschil_full_boundary_algebra}. In other words we set $\omega_{\Im^-,\mathcal{H}}$ as map
\begin{equation}\label{eq:Schwarzschild_boundary_state}
\omega_{\Im^-,\mathcal{H}}: a\otimes b\mapsto\omega_{\mathcal{H}}(a)\omega_{\Im^-}(b),
\end{equation}
where $a\in\mathcal{A}(\mathcal{H})$, $b\in\mathcal{A}(\Im^-)$ while $\omega_{\mathcal{H}}$ and $\omega_{\Im^-}$ are defined respectively in Proposition \ref{Prop:state_on_the_horizon} and in Proposition \ref{Prop:State_on_past_null_infinity}. We can use \eqref{eq:Schwarzschild_boundary_state} in combination with the injection \eqref{eq:Schwarzschild_*_homomorphism} and with Definition \ref{Def:induced_state} to obtain the following:

\begin{definition}\label{Def:Unruh_State}
	Let $\mathcal{A}^{obs}_0(\mathcal{M})$ be the $^*$-algebra of observables for a massless, conformally coupled real scalar field in $\mathcal{M}$, the physical region of the Kruskal extension of Schwarzschild spacetime. We call {\bf Unruh state} $\omega_U$ on $\mathcal{A}^{obs}_0(\mathcal{M})$ the unique  linear extension of 
	$$\omega_U: \mathcal{A}^{obs}_0(\mathcal{M}) \ni a\mapsto \omega_{\Im^-,\mathcal{H}}(\iota_{\Im^-,\mathcal{H}}(a)),$$
	where $\omega_{\Im^-,\mathcal{H}}$ is the state \eqref{eq:Schwarzschild_boundary_state} while $\iota_{\Im^-,\mathcal{H}}$ is the injective $^*$-homomorphism \eqref{eq:Schwarzschild_*_homomorphism}.
\end{definition}

\noindent The Unruh state inherits several properties for the building blocks at the boundary, especially the property of being quasi-free. 
The most important properties of $\omega_U$ are summarized in the next theorem which recollects the main results of \cite{Dappiaggi:2009fx} where all proofs can be found:

\begin{theorem}\label{Th:Properties_Unruh_State}
	Let $\omega_U:\mathcal{A}^{obs}_0(\mathcal{M})\to\mathbb{C}$ be the Unruh state as per Definition \ref{Def:Unruh_State}. It enjoys the following properties:
	\begin{enumerate}
		\item It is invariant under the action on $\mathcal{A}^{obs}_0(\mathcal{M})$ of all spacetime isometries,
		\item $\omega_U$ is a Hadamard state in the sense of Definition \ref{Def:local_Hadamard_form}.
	\end{enumerate}
\end{theorem}
\noindent We comment only that the proof of this theorem concerns mainly item {\em 2.} since the first one holds true almost automatically per construction. The only point, which needs to be addressed, is a
  generalization of Proposition \ref{prop:isometries_bulk_to_boundary_homomorphism} to the case in hand. 
From a physical point of view $\omega_U$ represents a state which becomes ground at past 
null infinity, while its restriction to $\mathcal{H}^\pm$ is a thermal state at the Hawking temperature, as per Proposition \ref{Prop:state_on_the_horizon}. These properties justify calling $\omega_U$ the
 Unruh state. It is worth mentioning that $\omega_U$ has manifold applications. On the one hand, 
 together with its smooth perturbations, it is believed to be the natural candidate to be used in the description of the gravitational collapse of a spherically symmetric state.
  On the other hand it represents the key ingredient to apply the construction of Fredenhagen and Haag, which proves the appearance of Hawking radiation at future 
  null infinity $\Im^+$, \cite{Fredenhagen:1989kr}. \\ Another state related with Hawking radiation phenomenon is the thermal equilibrium state called {\em Hartle-Hawking-Israel} 
  state which is defined in the whole Kruskal manifold. Its uniqueness was
   established by Kay and Wald in \cite{KW} assuming it satisfies the Hadamard property. Only recently Sanders proved the existence of such state (and more generally, the analogue for spacetimes with bifurcate Killing horizons assuming some further symmetry hypothesis) showing also that it is of Hadamard type \cite{Ko}. Noteworthy are also the results of Brum and Jor\'as \cite{BrumJoras} who have used a bulk to boundary procedure to construct an Hadamard state in the Schwarzschild-de Sitter spacetime.

\section{Some further applications}

Goal of this last section is to discuss some of the physical properties of the states constructed with the pullback procedure presented in Section \ref{se:pullback}.
In particular, we shall see that, whenever it is available, the mode decomposition of the two-point functions of the states pulled back from the null boundary has a distinctive character.

\subsection{Asymptotic vacua and scattering states}

Starting from the observation that the state $\omega_{M}$ for a massless real scalar field on Minkowski spacetime, constructed in \eqref{eq:bulk_state_asymptotically_flat} coincides with the Poincar\'e vacuum, we elaborate further how, more generally, $\omega_{M}$, as discussed in section \ref{se:pullback} can be understood as an asymptotic vacuum state. 
To this end, let us focus once more on the four dimensional Minkowski spacetime $\mathbb{M}^4\equiv(\mathbb{R}^4,\eta)$. In addition we consider the standard reference frame for which the metric takes its canonical form
 \[
\eta = - dt\otimes dt + \sum_{i=1}^3 dx^i\otimes dx^i\:.
 \]
Let $f\in C^\infty_0(\mathbb{R})$ and let $\phi_f=G_0(f)$ where $G_0$ is the causal propagator for the massless Klein-Gordon operator $P_0$ as per Proposition \ref{prop:causal_propagator}. The associated mode decomposition reads
\begin{equation}\label{eq:mode-mink}
\phi_f(t,\vec{x}) = \frac{1}{(2\pi)^{\frac{3}{2}}} \int\limits_{\mathbb{R}^3} \left( e^{-i |\vec{k}| t + i\vec{k}\cdot\vec{x}} \widetilde{\phi}_+(\vec{k}) +e^{i |\vec{k}| t - i\vec{k}\cdot\vec{x}} \overline{\widetilde{\phi}_+(\vec{k})} \right) \frac{d^3\vec{k} }{\sqrt{2|\vec{k}|}} 
\end{equation}
where $k=|\vec{k}|$. In order to relate $\phi_f$ with its value on future null infinity $\Im^+$, we can make use of the injection map $\Gamma_{\Im}$ introduced in Theorem \ref{Th:bulk_to_boundary}. As an alternative procedure we can take directly a limit towards $\Im^+$ by considering the standard spherical coordinates 
\[
\vec{x} = (r\sin\theta \cos \varphi, r \sin\theta\sin\varphi,r\cos\theta)
\]
and then defining the  null coordinates $u\doteq t-r$, $v\doteq t+r$. It holds that 
$$ 
(\Gamma_{\Im}\phi_f)(u,\theta,\varphi) = \lim_{v\to +\infty} \frac{v}{2}\phi(u,v,\theta,\varphi) \:.
$$
To unveil an interplay between this limit and the mode decomposition in \eqref{eq:mode-mink}, we can use \cite[Lemma 4.1]{Dappiaggi:2005ci}, which ensures that 
\begin{gather*}
(\Gamma_{\Im}\phi_f)(u,\theta,\varphi)=
\frac{1}{\sqrt{2\pi}}
\int_0^\infty \left( e^{-iku} \,  \widetilde{\phi}_+(k,\theta,\varphi) + e^{iku} \,   \overline{\widetilde{\phi}_+(k,\theta,\varphi)} \right) \sqrt{\frac{k}{2}}\, dk\;
\end{gather*}
where $\widetilde{\phi}_+(k,\theta,\varphi)$ is nothing but $\widetilde{\phi}_+$ as in \eqref{eq:mode-mink}, though rewritten in spherical coordinates and with $k=|\vec{k}|$. In an analogous way, the relation \eqref{eq:bulk_two_point_asymptotically_flat} among the bulk and the boundary two-point functions can be rewritten  more explicitly as
\begin{gather*}
\omega_{2,\mathbb{M}^4}(f,f^\prime)=\langle\phi_+|\phi^\prime_+\rangle_{\mathbb{M}^4}=\\
= \int\limits_{\mathbb{R}^+\times \mathbb{S}^2}  \; k^2\: \:
 \overline{\widetilde{\phi_+}(\vec{k})} \widetilde{\phi^\prime_+}(\vec{k}) \:dk  d\mu_{\mathbb{S}^2}
 =\omega_{\Im}([\Gamma_{\Im}([f])][\Gamma_{\Im}([f^\prime])])\:.
\end{gather*}
where both $\widetilde{\phi}$ and $\widetilde{\phi^\prime}$ can be obtained from $\phi_f$ and $\phi_{f^\prime}$ inverting \eqref{eq:mode-mink}. Here $\mu_{\mathbb{S}^2}$ is the standard measure on the unit $2-$sphere.
The first equality in the previous chain of equations is the standard relation between the two-point function and a scalar product defined intrinsically on $\Sigma_0$, here taken to be the hypersurface in $\mathbb{M}^4$ at $t=0$. More precisely $\widetilde{\phi}_+$ can be equivalently read as the spatial Fourier transform of the initial data of $\phi_f$ on $\Sigma_0$.

Observe that, consistently with the idea of considering a ground state, the negative frequencies in \eqref{eq:mode-mink} do not play any role in the previous product. If we now repeat this analysis on an asymptotically flat spacetime $M$ as per Definition \ref{Def:asymptotically_flat}, we notice that in analogy to the Minkowski case, the negative frequency part with respect to the coordinate associated to the vector field $n$ on $\Im^+$ are suppressed in $\omega_{2,\Im}$.
This statement holds true also in a neighbourhood of $\Im^+$ in the unphysical spacetime $\widetilde{M}$, provided that the Bondi coordinates are meaningful. In other words
$$ 
(\Gamma_{\Im}\phi_f)(u,\theta,\varphi) = \lim_{\Omega\to 0} \frac{1}{\Omega}\phi_f(u,\Omega,\theta,\varphi) \:.
$$
Hence, by continuity, we infer that in the bulk two-point function, computed as a scalar product on a Cauchy surfaces $\Sigma$, the negative frequency part 
tends to vanish in the limit where $\Sigma$ tends to $\Im^+$. For this reason the state we are considering can be understood as an asymptotic vacuum and its two-point function $\omega_M(x,x^\prime)$ approximate more and more the Poincar\'e vacuum when both points $x,x^\prime$ tend to $\Im^+$.

\subsection{Applications to cosmology}

Let us here consider instead a four-dimensional,  Friedmann-Robertson-Walker spacetime $(M,g_{FRW})$ with flat simply-connected spatial sections, whose metric is given in \eqref{metric} with $\kappa=0$.
We recall that, also in that case, any solution $\phi$ of the Klein-Gordon equation with compactly supported, smooth, initial data can be decomposed in modes as
\[
\label{bosondecomposition}
\phi(\tau,\vec{x})= \int\limits_{\mathbb{R}^3}
\left( \phi_{\vec{k}}(\tau,\vec{x}) \widetilde{\phi}(\vec{k})      + \overline{\phi_{\vec{k}}(\tau,\vec{x}) \widetilde{\phi}(\vec{k})}
\right)
d^3k,
\]
where $\tau$ is the conformal time, while 
\[
\phi_{\vec{k}}(\tau,\vec{x}) = \frac{ T_{k}(\tau)   e^{i\vec{k}\vec{x}}}{(2\pi)^{3/2}a(\tau)}.
\]
Furthermore, $T_k(\tau)$ is a complex solution of the differential equation
\begin{equation}\label{eq:tkmodes}
(\partial^2_\tau +k^2+a^2m^2)T_{k} = 0\,,
\end{equation}
which abides to the following normalisation condition
\[
T_k\partial_\tau\overline{T_k} - \overline{T_k}\partial_\tau T_k \equiv i\;.
\]
We are also assuming that the functions $k\mapsto T_k(\tau)$ and $k\mapsto \partial_\tau T_k(\tau)$
are both polynomially bounded for large $k$ uniformly in $\tau$ and they are $L^2_{loc}([0,+\infty),kdk)$ for every $\tau \in \mathbb{R}$.
As shown in \cite{Dappiaggi:2007mx, Pinamonti:2011} the modes $T_k$ fulfilling the above assumptions can be concretely constructed in a large class of spacetimes by means of a converging perturbative/Dyson series, built out of the modes for the massless theory, which are known explicitly. Uniqueness is not guaranteed and the available freedom can be used to decompose the two-point function $\omega_2$ of any pure, quasi-free state on $(M,g_{FRW})$ which is invariant under the action of the three-dimensional Euclidean group, classifying the isometries of the metric at any, fixed value of $\tau$. Following the pioneering work of \cite{LuRo}, in which the work of Parker  \cite{Parker} about adiabatic states is made rigorous, 
(see also \cite{Junker, Junker:2001gx, Olbermann:2007gn, Avetisyan:2012wq}), 
it holds that 
\begin{equation}\label{eq:two-pt-luders-roberts}
\omega_2(x,x^\prime):= \frac{1}{(2\pi)^3} \int\limits_{\mathbb{R}^3}   \frac{\overline{T_k}(\tau)}{a(\tau)}\frac{T_k(\tau^\prime)}{a(\tau^\prime)} e^{i{\bf k}\cdot(\vec{x}^\prime-\vec{x})}     d\vec{k} \;,
\end{equation}
for a given function $T_k$, solution of \eqref{eq:tkmodes}. 
Notice that,  $(M,g_{{FRW}})$ being conformally flat, if the limit $\tau\to-\infty$ can be taken in $M$, then the past boundary of $M$ is conformally related to $\Im^-$, part of the conformal boundary of Minkowski spacetime. In this case, 
a particularly notable choice for the modes $T_k$ can be made imposing the following initial conditions 
\begin{equation}\label{eq:initial_conditions_modes}
\lim_{\tau\to-\infty} e^{ik\tau}T_k(\tau) = \frac{1}{\sqrt{2k}} \;, \qquad
\lim_{\tau\to-\infty} e^{ik\tau}\partial_\tau T_k(\tau) = -i\sqrt{\frac{k}{2}} \;.
\end{equation}
Under suitable geometric hypotheses we can make use once more of a bulk-boundary procedure. We obtain via pull-back as per \eqref{eq:bulk_two_point_cosmological} a two-point function which coincides with $\omega_2$ in  \eqref{eq:two-pt-luders-roberts}, built out of modes $T_k$ which satisfy the initial conditions \eqref{eq:initial_conditions_modes}.
This is indeed the case when $(M,g_{FRW})$ is asymptotically de Sitter in the limit $\tau\to-\infty$ as discussed in  \cite{Dappiaggi:2007mx} or when
 $(M,g_{FRW})$ possesses a null initial singularity as $\tau\to-\infty$, see {\em e.g.} \cite{Pinamonti:2011}.  Hence, combining this observation with the results of Theorem \ref{Th:Hadamard_form_Cosmological} for the case of asymptotically de Sitter spacetimes or Th.3.1 in \cite{Pinamonti:2011} for the case of spacetimes with null initial singularities, it holds that 
the quasi-free state constructed with the two-point function \eqref{eq:two-pt-luders-roberts} with the initial conditions \eqref{eq:initial_conditions_modes} for the modes if of Hadamard form. We shall now discuss applications of this kind of states in the literature.

\subsubsection{Quantum Fluctuations}
It is strongly believed that, after the Big Bang, the Universe experienced a phase of rapid expansion, known as inflation. The spacetime during that phase is best modelled by the flat patch of a de Sitter manifold with a very large curvature if compared to the present one. It is furthermore believed that the tiny quantum fluctuations of the matter-gravity coupled system have been inflated to become classical ripples in the spacetime. These perturbations can be detected in the cosmic background radiation. 

A careful analysis of this mechanism has been accounted for by Mukhanov, Feldman and Brandenberger \cite{Mukhanov}, see also the recent analysis \cite{Hack:2014} for the case of a scalar Klein-Gordon field. In the latter, the scalar modes of the matter-gravity fluctuations are described by a single scalar field called Mukhanov-Sasaki variable $\Psi$. Furthermore, at the first order in perturbation theory,  $\Psi$ is decoupled from the tensor modes and it satisfies a free linear wave-like equation over the background. More precisely, it holds
\begin{equation}\label{eq:mukhanov-sasaki}
\left( \Box  -\frac{R}{6}  + \frac{\overline{z}''}{a^2 \overline{z}}\right)\Psi=0, \qquad \overline{z} \doteq \frac{\phi'}{H}
\end{equation}
where $R$ is the scalar curvature of the background Friedmann-Robertson-Walker metric, while $H$ is its Hubble parameter. Furthermore, $\overline{z}''$ is the second derivative of $\overline{z}$ with respect to the conformal time, while $\phi$ is a background solution of the Klein-Gordon equation. Typically, the quantum state in which these fluctuations are originated is assumed to be the Bunch-Davies state \cite{BD}. 

However, it is an over idealization to assume that the background is modelled exactly by the expanding flat patch of the de Sitter spacetime and, at the same time, in \eqref{eq:mukhanov-sasaki}, it is present a time dependent potential. For these reasons, the selection of the state 
used in the analysis is done assuming that its two-point functions tends to the  Bunch Davies one in the limit $\tau\to-\infty$. This procedure, advocated in the literature, see e.g. \cite{Mukhanov} shares several common features with the construction outlined in Section \ref{se:pullback}. To wit one should consider a two-point function $\omega_2$ as the one given in \eqref{eq:two-pt-luders-roberts} with the requirements \eqref{eq:initial_conditions_modes}, aimed at choosing the modes to be a solution of \eqref{eq:tkmodes} with a time dependent square mass equal to $-{\overline{z}''}/{(a^2 \overline{z})}$ descending from  \eqref{eq:mukhanov-sasaki}. The bulk-to-boundary procedure can be adapted to discuss also this situation and, on account of Theorem \ref{Th:Hadamard_form_Cosmological}, the ensuing state is of Hadamard form. In turn this entails that all quantum fluctuations are finite.  

\subsubsection{Semiclassical backreaction}
The back-reaction of a free scalar quantum fields on the classical curvature was thoroughly analysed by Anderson, see in particular \cite{Anderson3, Anderson4}. 
In these papers, the backreaction is taken into account by means of the semiclassical Einstein equations which in natural units are
\begin{equation}\label{eq:semiclassical}
G_{\mu\nu}(x)=8\pi \langle T_{\mu\nu}(x)\rangle_\omega
\end{equation}
where $G_{\mu\nu}$ is the Einstein tensor of the spacetime $(M,g)$ and $\langle T_{\mu\nu}\rangle_\omega$ is the expectation value of
the (renormalized \cite{Moretti:2001qh,Hack:2012qf}) stress-energy tensor operator $T_{\mu\nu}(x)$ associated to the quantum scalar field theory in a state $\omega$.
In the work of Anderson, the state of the scalar field is assumed to be quasi-free and, moreover, its two-point function is taken to be of the from 
\eqref{eq:two-pt-luders-roberts}. The modes $T_k$ are chosen prescribing their asymptotic behavior as \eqref{eq:initial_conditions_modes}. The ensuing state can be considered a bona-fide asymptotic vacuum in the past.  

More recently in \cite{Pinamonti:2011}, see {\it e.g.} \cite{PinamontiSiemssen} it is shown that the problem of existence of solutions of the semiclassical Einstein equations on a cosmological spacetime can be addressed with fixed point methods once a large class of spacetime with a uniform and coherent choice for the states can be selected. Furthermore, the latter need to be such that $\langle T_{\mu\nu}(x)\rangle_\omega$ is meaningful. 
Even on homogeneous and isotropic spacetimes like $(M,g_{FRW})$ such selection is a daunting task, especially if one aims at analysing the properties of the state on a single Cauchy surface in such a way that the Hadamard condition holds ture.   

The bulk to boundary construction of states described in Section \ref{se:pullback} yields a practical tool to circumvent this problem. In particular in \cite{Pinamonti:2011} this method has been applied to assign asymptotic vacuum states to a class of cosmological spacetimes which possess a past null boundary, see also \cite{Hack:2010}. The chosen state is quasi-free state and the two-point function is of the form \eqref{eq:two-pt-luders-roberts}. The modes $T_k$ are chosen once more in such a way that the initial condition at $\tau\to\infty$ are those in \eqref{eq:initial_conditions_modes}. As discussed above, an application of the bulk to boundary methods described so far allows to prove that the ensuing states are of Hadamard form, see e.g. Theorem 3.1 in \cite{Pinamonti:2011}.

\subsubsection{Approximate KMS state in cosmology}

Up to now in this section we have considered only states which are approximate vacuum states at their boundaries. However, the bulk to boundary techniques discussed so far can be employed to construct states which have different asymptotic behaviour. In a few words, the pull-back of boundary states whose two-point function differs from  $\omega_{2,\Im}$ in \eqref{eq:notable_2pt_function} by a smooth two-point function induces states in the bulk which are again of Hadamard form, as one can infer adapting the proof of Theorem \ref{Th:Hadamard_form_Cosmological}. 
Similar ideas have been used in \cite{DHP} to construct states in cosmological spacetimes which can be interpreted as approximated KMS states near their boundaries. On the past null boundary $\Im^-$, equipped with an asymptotic tangent null field $n$, $\omega_{2,\Im}$ in \eqref{eq:notable_2pt_function}
is constructed as the two-point function of the vacuum with respect to the rigid translations generated by $n$.
In other words, the boundary two-point function $\omega_{2,\Im}$ can be obtained out of the scalar product $\mu$ \eqref{eq:boundary_scalar_product}, which we recall here being
\[
%
		\mu(\psi,\psi^\prime)= \mbox{Re} \int\limits_{\mathbb{R}\times\mathbb{S}^2}2k\Theta(k)\overline{\widehat{\psi}(k,\theta,\varphi)}\,\widehat{\psi}^\prime(k,\theta,\varphi) \:dkd\mu_{\mathbb{S}^2}(\theta,\varphi)\:.
\]
In order to opt for a different class of states, for example those appearing  
as thermal with respect to the translations generated by $n$, we have to alter the form of $\mu$. In particular, if we fix the inverse temperature $\beta$, it holfs
\[
\mu_\beta(\psi,\psi^\prime)= \mbox{Re} \int\limits_{\mathbb{R}\times\mathbb{S}^2}\frac{2k}{1-e^{-\beta k}}\;   \overline{{\widehat\psi}(k,\theta,\varphi)} \widehat\psi^\prime(k,\theta,\varphi)   \;  dk\,d\mu_{\mathbb{S}^2}(\theta,\varphi).
\]
The quasi-free state $\omega_\Im^\beta$ for $\mathcal{A}(\Im)$ constructed out of the two-point function defined out of $\mu_\beta$ satisfies the following properties whose validity is proven in \cite[Prop. 4.1]{DHP}:
\begin{itemize}
	\item[a)] $\omega_\Im^\beta$ is invariant under rotations and translations generated by $n$; 
	\item[b)] $\omega_\Im^\beta$ is a KMS (Kubo-Martin-Schwinger) state at inverse temperature $\beta$ with respect to the translations generated by $n$;  
	\item[c)] In the limit $\beta\to\infty$, $\omega_\Im^\beta$ converges weakly to $\omega_\Im$.
\end{itemize}
We recall here that states that satisfy the KMS condition can be seen as describing thermodynamic equilibrium. The analysis of this condition and of its implications can be found in \cite{Bratteli:1979tw,Bratteli:1996xq}.  
The pull-back of this state by means of $\iota_c$ as in \eqref{eq:bulk_state_cosmological} yields
\[
\omega^\beta_{M}:\mathcal{A}^{obs}(M)\to\mathbb{C},\quad \omega^\beta_{M} \doteq\omega^\beta_\Im\circ \iota_c\:.
\]
This state is per construction quasi-free and, at the level of two-point function, in analogy to \eqref{eq:bulk_two_point_cosmological}, it holds that 
\begin{equation}\label{eq:two-point-bulk-thermal}
\omega^\beta_{2,M}(f,f^\prime) = \omega^\beta_\Im([\Gamma_\Im([f])][\Gamma_\Im([f^\prime])]).
\end{equation}
We have the following result which is proven in \cite[Th. 4.1]{DHP}
\begin{theorem}
The state $\omega^M_\beta:\mathcal{A}(M)\to\mathbb{C}$ is of Hadamard form.
\end{theorem}
We notice here that the two-point function \eqref{eq:two-point-bulk-thermal} admits the following mode decomposition 
\begin{equation}\label{eq:pullbacktwo-point-function}
\omega_{2,\beta}(\tau_x,\vec{x},\tau_y,\vec{y})
=
\frac{1}{(2\pi)^3 a(\tau_x)a(\tau_y)}
\int\limits_{\mathbb{R}^3}
\left(  \frac{\overline{T_k}(\tau_x) T_k(\tau_y)}{1-e^{-\beta k}}   +
\frac{{T_k}(\tau_x) \overline{T_k}(\tau_y)}{e^{\beta k}-1}              \right)  e^{-i\vec{k}{(\vec{x}-\vec{y})}}   \;d^3k\;,
\end{equation}
where $T_k(\tau)$ are the modes which satisfy \eqref{eq:tkmodes} and the initial conditions \eqref{eq:initial_conditions_modes}.
In the case of a conformally coupled massless theory the state $\omega^M_\beta$ constructed with the bulk-to-boundary procedure discussed so far coincides with a conformal KMS state, namely the one which satisfies the KMS condition with respect to the one parameter group of $^*-$automorphism generated by the conformal Killing vector field
\[
K = a(t)\frac{\partial}{\partial t}.
\] 
In the case of a massive scalar field theory, the state obtained by pull-back under the bulk-to-boundary map is no longer of equilibrium with respect to a given time evolution. Despite of this fact the spectrum of the two-point function, as seen in \eqref{eq:pullbacktwo-point-function} is very close to the one of an equilibrium state in flat spacetime. For this reason they are called approximate equilibrium states. They can hence be used to describe massive particles in an almost thermal equilibrium.

\subsection{Hadamard states and Hawking radiation}

In \cite{Fredenhagen:1989kr} Fredenhagen and Haag have shown that the renown Hawking radiation can be obtained under mild assumptions on the state which describes a scalar quantum field in a spacetime which experiences a spherically symmetric collapse to a Schwarzschild black hole. 
In particular, one of these assumptions requires that the state is Hadamard in the neighbourhood of the future horizon in region $\mathcal{B}\cup\mathcal{W}$ in Figure \ref{fig2}. In addition, towards $\Im^-$, the state must be an asymptotic vacuum. 
In the late eighties, no example of such kind of states was given. Nowadays we know that the Unruh state satisfies all these properties and in particular it is an Hadamard state as seen in Theorem \ref{Th:Properties_Unruh_State}.

More recently in \cite{MP12}, the key ingredients for obtaining a state which describes the appearance of the Hawking radiation have been obtained analysing only certain local properties of the state near the future horizon. In particular, the thermal nature in the exact counterpart at the level of the correlation functions for pairs of points of the manifold which are separated by a Killing horizon. Furthermore, these correlations allow to estimate the tunneling probability of particles across the horizon. This idea was originally devised by Parikh and Wilczek in \cite{ParikhWilczek}. 
The analysis in \cite{MP12} is valid in a region $\mathcal{O}$ which contains a portion of a bifurcated Killing horizon $\mathcal{H}$.
Furthermore, the two-point function of the quantum state is assumed to satisfy a condition weaker than the Hadamard one
\[
\omega_2(x,x') = \lim_{\epsilon\to0} \frac{U(x,x')}{\sigma_\epsilon(x,x')}+ w_\epsilon(x,x')\:, \label{ultra}
\]
where $U$ is proportional to the van Vleck-Morette determinant, see {\it e.g.} \eqref{eq:Hadamard_parametrix}, $\sigma_\epsilon$ is the regularized halved, squared geodesic distance between the points $x$ and $x'$, while $w_\epsilon$ is a distribution which is less singular then $\sigma_\epsilon^{-1}$.
If $\mathcal{O}$ is sufficiently small, it is possible to parametrize its points with a coordinate patch $(V,U,x_3,x_4)$ adapted to the horizon in the following way: The horizon $\mathcal{H}$ can be seen as a congruence of null-geodesics, $U$ is an affine parameter along the null geodesics forming the horizon, $(x_3,x_4)$ are spatial coordinates on $\mathcal{H}$, while $U$ is fixed in such a way that the metric $g$ takes the following form 
\begin{equation}\label{eq:metric-near-horizon}
\left.g\right|_{{\cal H}} = - \frac{1}{2}dU \otimes dV -  \frac{1}{2}dV \otimes dU + \sum_{i,j=3}^4 h_{ij}(x_3, x_4) dx^i\otimes  dx^j 
\end{equation}
where the metric $h$ is the one induced by $g$ on the spatial sections of $\mathcal{H}$. This coordinate system is exactly the Bondi frame introduced in Section \ref{Sec:Example} for the case of the horizons of Schwarzschild spacetime. 
We can now analyse the scaling properties of the two-point function $\omega_2(x,x')$ when both points $x,x'$ tend to the horizon $\mathcal{H}$. 
In particular, we evaluate
\[
\lim_{\lambda\to 0^+} \omega_2\left(f_\lambda, f'_\lambda\right)  
\]
where, for $\lambda >0$,
\[
f_\lambda(V,U,x_3,x_4) \doteq \frac{1}{\lambda}f \left(\frac{V}{\lambda},U,x_3,x_4\right) \;, \quad 
f'_\lambda(V,U,x_3,x_4) \doteq \frac{1}{\lambda}f' \left(\frac{V}{\lambda},U,x_3,x_4\right)
\]
where the compactly supported functions $f,f'$ are centered around $x,x'$ and do not posses zero modes. In other words
\[
f = \frac{\partial F}{\partial V}\;, \quad  f' = \frac{\partial F'}{\partial V}\:, \quad \mbox{for given $F,F' \in C_0^\infty(\mathcal{O})$.}\label{f2}
\] 
As proven in \cite[Th. 3.1]{MP12} it holds that 
\begin{equation}
\lim_{\lambda\to 0^+} \omega_2\left(f_\lambda, f'_\lambda\right)  = \lim_{\epsilon\to0^+}  -\frac{1}{16\pi} \int\limits_{\mathbb{R}^4\times {\cal B}}  \frac{F(V,U,s) F'(V',U',s)}{(V- V'-i\epsilon)^2}  dU dV dU' dV' d\mu(s)\:.
\label{correlations}
\end{equation}
where $s$ is a compact notation for $(x_3,x_4)$ and $\mu(s)$ is the measure induced by $h$ on $\mathcal{H}$. The limit obtained so far is exactly the two-point function of the state $\omega_{\Im}$ introduced in Definition \ref{Def:boundary_quasi_free_state}.
We outline a few physical properties of this limit in two notable cases. 

When the supports of $f_\lambda$ and $f'_\lambda$ lie in $J^-{(\mathcal{H})}\cap \mathcal{O}$ it holds that 
\[
\lim_{\lambda\to 0 }  \omega(\Phi(f_\lambda)\Phi(f'_\lambda)) = \frac{1}{32} \int_{\mathbb{R}^2\times {\cal B}} \left(\int_{-\infty}^\infty 
 \frac{\overline{\widehat{F}(E,U,x)}  \widehat{F'}(E,U',x) }{1-e^{-\beta_H E}} E dE \right) dU dU' 
d\mu(x) \;,
\]
where $\beta_H={2\pi}/{\kappa}$ is the {\em inverse Hawking temperature}, while $\widehat{F}$ is the Fourier Plancherel transform along $\tau=v =\log(V)/\kappa$ where $\kappa$ is the surface gravity. The thermal nature of the resulting limit can be inferred form the Bose factor $({1-e^{-\beta_H E}})^{-1}$.

When the support of $f_\lambda$ lies in $J^+{(\mathcal{H})}\cap \mathcal{O}$ while that of $f'_\lambda$ lies in $J^-{(\mathcal{H})}\cap \mathcal{O}$ we are considering the correlations between elements in the inner and in the outer regions. In this case, up to a normalization, $|\omega(\Phi(f_\lambda)\Phi(f'_\lambda))|^2$ can now be interpreted as a tunneling probability through the horizon. It holds that 
\[
\lim_{\lambda\to 0 }   \omega_2(f_\lambda,f'_\lambda)) = \frac{1}{16} \int_{\mathbb{R}^2\times {\cal B}} \left(\int_{-\infty}^\infty 
 \frac{\overline{\hat{F}(E,U,s)}  \hat{F'}(E,U',s)}{\sinh(\beta_H E/2)}  EdE \right) dU dU' 
d\mu(s) \:.\label{tunnel} 
\]
Considering wave packets concentrated around a high value of the energy $E_0$ we estimate the tunneling probability as
$$
\lim_{\lambda\to 0}|\omega(\Phi(f_\lambda)\Phi(f'_\lambda))|^2 \sim \mbox{const.}\:  E_0^2 \: e^{-\beta_H E_0}\:,
$$
in agreement with the ideas in \cite{ParikhWilczek}. 
We see in this procedure that these properties can be obtained in any state which is close to the Unruh state provided that it is of Hadamard form.

%
%
%
\chapter[Wick Polynomials]{Wick Polynomials and Extended Observables Algebras}
\label{Ch:wicks} 

Goal of this chapter is to discuss the interplays between the bulk-to-boundary correspondence outlined in the previous chapters and the extension of the algebra of observables to include also the Wick polynomials. 
In order to keep at bay both the length of the discussion and the intrinsic complication of the notation, we focus on the case of a non empty double cone $\mathcal{D}=\mathcal{D}(q,p)=J^{+}(p)\cap J^{-}(q)$, where $p,q$ are two arbitrary points of a globally hyperbolic spacetime $M$. Observe that  $\mathcal{C}=\{ J^{+}(p) \setminus I^+(p)\} \cap \mathcal{D}$ forms a null surface which can be interpreted as the past boundary of $\mathcal{D}$. In this framework $\mathcal{D}$ plays the role of the bulk spacetime, on top which we consider observables to be related with a suitable counterpart living intrinsically on $\mathcal{C}$. Notice that, in the case of asymptotically flat spacetimes endowed with $\Im^-$, past null infinity, then $M$ can be identified with the unphysical spacetime, $\mathcal{C}$ with $\Im^-$ itself, while $\mathcal{D}$ is nothing but the physical spacetime, up to a conformal transformation.

\section{Extended algebra of fields}

As a preliminary step we discuss how one can extend the algebra of observables $\mathcal{A}^{obs}(M)$ for a generic, real, massive scalar field on a globally hyperbolic spacetime, to include local, nonlinear fields. We shall follow a construction which is suited to work on a generic curved background; for further details refer to \cite{Hollands:2001nf, Hollands:2001fb} and \cite{BFK}. 

Recall that, in view of Definition \ref{def:algebra_of_observables},  $\mathcal{A}^{obs}(M)$ is realized as the quotient of the universal tensor algebra \eqref{eq:universal_tensor_algebra} with respect to the ideal generated by the commutation relations. 
Hence, a generic element $[F]\in\mathcal{A}^{obs}(M)$ is an equivalence class whose representative $F$ can be expressed as a direct sum 
\[
F = \bigoplus_n{F_n}, \qquad F_n \in \mathbb{C}\otimes(\mathcal{E}^{obs}(M))^{\otimes_s n} 
\]
where $\otimes_s$ indicates the symmetric tensor product. Observe that, since $\mathcal{E}^{obs}(M)=\frac{C^\infty_{0}(M)}{P[C^\infty_{0}(M)]}$, each element $F_n$ represents in turn a (tensor product of) equivalence classes. For later convenience and with a slight abuse of notation, we will indicate with $F_n$ the choice of a representative in the tensor products of equivalence classes. In other words $F_n\in (C^\infty_0(M;\mathbb{C}))^{\otimes_s n}$.

Our goal is to extend the algebra of observables so to include also the Wick polynomials. These are described in terms of a pointwise product of fields, which, in term of algebra elements, entails the necessity to account for elements $F = \bigoplus_n{F_n}$ whose components $F_n$ are distributions supported on the diagonal 
$D_n=\{(x,\dots, x)\in M^n, x\in M\}$, $M$ being the underlying spacetime. In general the extension of the algebra product of $\mathcal{A}^{obs}(M)$ to two of such elements gives rise to divergences. The way out of this quandary consists of finding a different representation for $\mathcal{A}^{obs}(M)$ in which such problem is not observed 

In order to better display the necessary steps to reach this goal, we represent the elements of $\mathcal{A}^{obs}(M)$ as functionals over real smooth field configurations as in \cite{Brunetti:2009qc}; see also \cite {Brunetti:2009pn, Rejzner}. 
More precisely, for a given $F=\{F_n\}_n\in\mathcal{A}^{obs}(M)$ and for every $\varphi\in C^\infty(M)$, we set
\begin{equation}
  \label{eq:functional} 
  F ( \varphi ) \doteq \sum_{n = 0}^{\infty} \frac{1}{n!} \langle
  F_{n} , \varphi^{\otimes n} \rangle_n \text{.}
\end{equation}
We stress that we are adopting the same symbol $F$ to indicate different, albeit closely related mathematical objects. In \eqref{eq:functional} $F$ stands for a functional whose building blocks, the coefficients $F_n$ can be constructed via the functional (Gateaux) derivatives evaluated in $\varphi=0$. This entails also that we are considering only smooth functionals in the sense that all functional derivatives exist and each $F$ possesses only a finite number of such derivatives. In other words it holds
\begin{equation*}
  F^{(n)} ( \varphi ) (h^{\otimes n}) =
  \left. \frac{d^{n}}{d\lambda^{n}} F ( \varphi + \lambda h )
  \right\vert_{\lambda = 0} \text{,} \quad \forall h \in C^{\infty} (
  M ) \text{,}
\end{equation*}
and in particular
\begin{equation}\label{eq:functional-derivatives-components}
F^{(n)} ( 0 ) (h^{\otimes n}) = \langle F_n, h^{\otimes n} \rangle_n
\end{equation}
where $\langle \cdot,\cdot \rangle_n$ indicates the standard pairing in $M^n$ between smooth and smooth, compactly supported functions.

In this picture the product in $\mathcal{A}^{obs}$ is given by the following $\star-$product
\begin{equation}
  \label{eq:star-product}
  ( F \star F^\prime ) ( \varphi ) = \sum_{n = 0}^{\infty} \frac{i^{n}}{2^{n}
    n!} \big\langle F^{(n)} ( \varphi ) , G^{\otimes n} {F^\prime}^{(n)} (
  \varphi ) \big\rangle \text{,} \qquad \forall F \text{,\thinspace} F^\prime
  \in \mathcal{A}^{obs} ( M ) \text{,}
\end{equation}
where $G\in \mathcal{D}'(M\times M)$ is the causal propagator introduced in Proposition \ref{prop:causal_propagator}. 
Consider now $F_1,F_2$ so that
\[
F_i=\int_M d\mu_g\, f_i \varphi^,\qquad f_i\in C^\infty_0(M),
\]
where $d\mu_g$ stands for the metric induced measure. 
If we apply \eqref{eq:star-product}, we obtain a term $\langle f_1,G^2 f_2\rangle$ which diverges because the pointwise product of $G(x,x^\prime)$ with itself is ill-defined. To bypass this hurdle we deform the $\star$-product so to replace $\frac{i}{2}G$ in \eqref{eq:star-product} with the Hadamard bi-distribution $H$ \eqref{eq:Hadamard_parametrix}. The product obtained in this way is denoted with $\star_H$ and the procedure used is actually a formal deformation of the algebra realized by the $*-$isomorphism $\alpha_{H-iG/2}:(\mathcal{A}^{obs},\star)\to(\mathcal{A}^{obs},\star_{H})$ whose action on $F\in\mathcal{A}^{obs}$ reads
\begin{equation}\label{eq:homeomorphisms}
    \alpha_{H-iG/2} ( F )  \doteq \sum_{n = 0}^{\infty} \frac{1}{n!}
    \left\langle \left(H-\frac{i}{2}G\right)^{\otimes n} , F^{(2n)} \right\rangle \text{.}
\end{equation}
The elements $(\mathcal{A}^{obs},\star_{H})$ can be understood as normal ordered fields. Furthermore, since 
only the symmetric part of the product is changed, the commutator of two linear fields is left untouched by $\alpha_{H}$.

\bigskip
Wick ordered fields can now be accounted for in $(\mathcal{A}^{obs},\star_{H})$ without introducing unwanted divergences:

\begin{definition}
Let $\mathcal{A}_e^{obs}$ be the set of functionals which are smooth, compactly supported and microcausal, that is
\begin{equation}
  \label{eq:inters}
  \text{WF} ( F_{n} ) \cap \big\{ \big( M \times
  \overline{V}^{\thinspace +} \big)^{n} \cup \big( M \times
  \overline{V}^{\thinspace -} \big)^{n} \big\}  = \emptyset \text{,}
 \end{equation}
where $\{p\}\times \overline{V}^{+/-}$ is the set of all future/past pointing causal contangent vectors in $T^*_pM$. Then $\mathcal{A}_e^{obs}$ forms a $*-$algebra if equipped with $\star_H$ as product while the involution is nothing but complex conjugation. 
\end{definition}

Since the choice of the Hadamard parametrix is not unique, it is important to stress that different choices for $\star_H$ yield algebras which are $*-$isomorphic. Given two parametrices $H,H^\prime$, the $*-$isomorphism is realized by 
\begin{equation}
  \label{eq:isomorphism}
  \begin{split}
    \mathfrak{i}_{H^{\prime} , H} &= \alpha_{H^{\prime}} \circ
    \alpha_{H}^{-1} \text{,} \\
    \alpha_{H} ( F ) & \doteq \sum_{n = 0}^{\infty} \frac{1}{n!}
    \big\langle H^{\otimes n} , F^{(2n)} \big\rangle \text{.}
  \end{split}
\end{equation}
We conclude the section by focusing shortly on ordinary Wick powers. To start with, we observe that Wick polynomials as
\[
F = \int f \phi^n d\mu_g 
\]
are contained in this extended algebra. In addition the requirement for the functionals to be microcausal is necessary to ensure that the $\star_H$-product among such polynomials is well defined. A detailed proof of this fact which makes extensive use of microlocal techniques can be given along the lines of Theorem 2.1 in \cite{Hollands:2001nf}. 

Here we recall that the proof according to which the product $F\star_H F^\prime$ is well-posed can be obtained applying notable results of H\"ormander. In particular Theorem 8.2.9 in \cite{Hormander} is used to estimate the wave-front set of $F^{(n)}\otimes {F^\prime}^{(n)}$ and of $H^{\otimes n}$ while Theorem 8.2.10 in \cite{Hormander} is used to multiply $F^{(n)}\otimes {F^\prime}^{(n)}$ with $H^{\otimes n}$.

\section{Extension of the algebra on the boundary}

In this section we outline the construction of the extended algebra of observables on the boundary. To this end we consider a double cone $\mathcal{D}(q,p)=J^{+}(p)\cap J^{-}(q)$, chosen to be such that its past null boundary $\mathcal{C}$ can be parametrized with the following coordinates
\[
  \mathcal{C} = \big\{ ( V , \theta , \varphi ) \in \mathbb{R}
  \times \mathbb{S}^{2} \medspace \vert \medspace V \in I \subset
  \mathbb{R} \text{,} \thickspace ( \theta , \varphi ) \in
  \mathbb{S}^{2} \big\} \text{,}
\] 
where $V$ is an affine parameter along the null geodesics forming $\mathcal{C}$ while $\theta , \varphi$ are transverse coordinates.
Furthermore, the open interval $I$ is chosen in such a way that the tip of the cone is not contained in $\mathcal{C}$. The metric thereon induced by that of $M$ turns out to be conformally equivalent to 
\begin{equation}
\label{eq:conformal-metric-cone}
\left.\tilde{g}^2\right|_\mathcal{H}=0\otimes dV+d\mathbb{S}^2(\theta,\varphi),
\end{equation}
where $d\mathbb{S}^2(\theta,\varphi)$ is the standard metric of the unit $2$-sphere. The construction of $\mathcal{A}(\mathcal{C})$ the algebra of observables on $\mathcal{C}$, can be performed along the lines of the one on null infinity discussed in section \ref{sec:observables-on-null-infinity}. 
As next step we need to analyse the singular structure of the two-point function of the state $\omega_\Im$ introduced in Lemma \ref{lem:boundary_integral_kernel}. This bi-distribution and its interplay with the boundary commutator function will take the role played in the previous section by the Hadamard bidistribution. In particular, the extension, we are looking for, needs to be compatible with the Hadamard regularization and furthermore, it must include the Wick polynomials constructed in the bulk and projected on the boundary by an application of time-slice axiom see {\it e.g.} \cite{ChilianFredenhagen}. 

To start with we recall the form of the wave front set of the boundary state, {\it i.e.}, as discussed and proven in \cite{Moretti}, its two-point function $\omega_{2,\mathcal{C}}$ has an integral kernel of the same form of \eqref{eq:notable_2pt_function}. Accordingly
\begin{equation}
  \label{eq:WFomega}
  \text{WF} (\omega_{2,\mathcal{C}}) \subseteq A \cup B \text{,}
\end{equation}
where
\begin{multline}
  \label{eq:WFA}
  A = \big\{ \big( ( V , \theta , \varphi , \zeta_{V} ,
  \zeta_{\theta}, \zeta_{\varphi} ) , ( V^{\prime} , \theta^{\prime} ,
  \varphi^{\prime} , \zeta_{V^{\prime}} , \zeta_{\theta^{\prime}} ,
  \zeta_{\varphi^{\prime}} ) \big) \in ( T^{\ast} \mathscr{C} )^{2}
  \setminus \{0\} \thickspace \vert \\
  V = V^{\prime} \text{,~} \theta = \theta^{\prime} \text{,~} \varphi
  = \varphi^{\prime} \text{,~} 0 < \zeta_{V} = -\zeta_{V^{\prime}}
  \text{,~} \zeta_{\theta} = -\zeta_{\theta^{\prime}} \text{,~}
  \zeta_{\varphi} = -\zeta_{\varphi^{\prime}} \big\}
\end{multline}
and
\begin{multline}
  \label{eq:WFB}
  B = \big\{ \big( ( V , \theta ,\varphi , \zeta_{V} , \zeta_{\theta}
  , \zeta_{\varphi} ) , ( V^{\prime} , \theta^{\prime} ,
  \varphi^{\prime} , \zeta_{V^{\prime}} , \zeta_{\theta^{\prime}} ,
  \zeta_{\varphi^{\prime}} ) \big) \in ( T^{\ast} \mathscr{C} )^{2}
  \setminus \{0\} \thickspace \vert \\
  \theta = \theta^{\prime} \text{,~} \varphi = \varphi^{\prime}
  \text{,~} \zeta_{V} = \zeta_{V^{\prime}} = 0 \text{,~}
  \zeta_{\theta} = -\zeta_{\theta^{\prime}} \text{,~} \zeta_{\varphi}
  = -\zeta_{\varphi^{\prime}} \big\} \text{.}
\end{multline}
To introduce the extension of $\mathcal{A}(\mathcal{C})$, we start by discussing the regularity for the components of the algebra of observables decomposed as in \eqref{eq:functional}. Following \cite{DPP},
\begin{definition}
  \label{Def:testdistributions}
  We call $\mathcal{F}^{n}$ the set of elements $F^{\prime}_{n} \in
  \mathcal{D}^{\prime} ( \mathscr{C}^{n};\mathbb{C})$ that fulfil the
  following properties:
  \begin{enumerate}
  \item \textbf{Compactness}: Each $F^{\prime}_{n}$ is compact towards
    the future, \textit{i.e.}, the support of $F^{\prime}_{n}$ is
    contained in a compact subset of $\mathscr{C}^{n} \sim (
    \mathbb{R} \times \mathbb{S}^{2} )^{n}$.
  \item \textbf{Causal non-monotonic singular directions}: The wavefront set of $F^{\prime}_{n}$ contains only causal non-monotonic directions:
    \begin{equation}
      \label{eq:WFS} 
      \text{WF} ( F^{\prime}_{n} ) \subseteq W_{n} \doteq \big\{ ( x , \zeta
      ) \in ( T^{\ast} \mathscr{C})^{n} \setminus \{0\} \medspace
      \vert \medspace ( x , \zeta ) \not\in \overline{V}_{n}^{+} \cup
      \overline{V}_{n}^{-} \text{,~} ( x , \zeta ) \not\in S_{n}
      \big\} \text{,}
    \end{equation}
    where $( x , \zeta ) \equiv ( x_{1} , \dotsc , x_{n} , \zeta_{1} ,
    \dotsc , \zeta_{n} ) \in \overline{V}^{+}_{n}$ if, employing the
    standard coordinates on $\mathscr{C}$, for all $i = 1$, \dots,
    $n$, $( \zeta_{i} )_{V} > 0$ or $\zeta_{i}$ vanishes.  The
    subscript $V$ here refers to the component along the $V$-direction
    on $\mathscr{C}$.  Analogously, we say $( x , \zeta ) \in
    \overline{V}^{-}_{n}$ if every $( \zeta_{i} )_{V} < 0$ or
    $\zeta_{i}$ vanishes.  Furthermore, $( x , \zeta ) \in S_{n}$ if
    there exists an index $i$ such that, simultaneously, $\zeta_{i}
    \neq 0$ and $( \zeta_{i} )_{V} = 0$.
  \item \textbf{Smoothness Condition}: The distribution
    $F^{\prime}_{n}$ can be factorised into the tensor product of a
    smooth function and an element of $\mathcal{F}^{n-1}$ when
    localised in a neighbourhood of $V = 0$. In other words there
    exists a compact set $\mathcal{O} \subset \mathscr{C}$ such
    that, if $\Theta \in C^{\infty}_{0} ( \mathscr{C} )$ so that
    it is equal to $1$ on $\mathcal{O}$ and $\Theta^{\prime} \doteq 1
    - \Theta$, then for every multi-index $P$ in $\{ 1 , \dotsc , n
    \}$ and for every $i \leqslant n$,
    \begin{equation}
      \label{eq:factorisation}
      f \doteq \tilde{F}^{\prime}_{n} ( u_{x_{P_{i+1}} , \dotsc ,
        x_{P_{n}}} ) \thinspace \Theta^{\prime}_{x_{P_{1}}}
      \negthinspace \dotsm \thinspace \Theta^{\prime}_{x_{P_{i}}} \in
      C^{\infty} ( \mathscr{C}^{i};\mathbb{C} ) \text{,}
    \end{equation}
    where $\tilde{F}^{\prime}_{n} : C^{\infty}_{0} (
    \mathscr{C}^{n-i-1}) \to \mathcal{D}^{\prime} (
    \mathscr{C}^{i} )$ is the unique map from $C^{\infty}_{0} (
    \mathscr{C}^{n-i-1})$ to $\mathcal{D}^{\prime} (
    \mathscr{C}^{i} )$ determined by $F_{n}^{\prime}$ using the
    Schwartz kernel theorem.  Furthermore, $u_{x_{P_{i+1}} , \dotsc ,
      x_{P_{n}}} \in C^{\infty}_{0} ( \mathscr{C}^{n-i} )$, and we
    have specified the integrated variables $x_{P_{i+1}}$, \dots,
    $x_{P_{n}}$.  For every $j \leqslant i$, $\partial_{V_{1}}
    \dotsm \partial_{V_{j}} f$ lies in $C^{\infty} (
    \mathscr{C}^{i};\mathbb{C} ) \cap L^{2} ( \mathscr{C}^{i} ,
    dV_{P_{1}} d\mathbb{S}_{P_{1}}^{2} \dotsm dV_{P_{i}}
    d\mathbb{S}_{P_{i}}^{2} ) \cap L^{\infty} ( \mathscr{C}^{i}
    )$, while the limit of $f$ as $V_{j}$ tends uniformly to $0$ vanishes uniformly in the other coordinates.
  \end{enumerate}
\end{definition}

\noindent Before proceeding in our investigation, we collect here a few remarks aimed at clarifying the three above conditions: 
\begin{remark} 
{\bf (1)} The compactness condition is not a restrictive one. Notice that, for any $\mathcal{O} \subset \mathcal{D} $
we have that $J^-(\mathcal{O}) \cap \mathcal{C}$ is compact towards the future. Furthermore, since causality holds in $\mathcal{D}$, observables supported in $ \mathcal{D}$  are not influenced by points of $\mathcal{C}$ which lie outside  $J^-(\mathcal{O})$. The bulk-to-boundary map, that we shall construct below, is an application of the time-slice axiom which respects causality. Hence we expect that the image of any observable supported in $\mathcal{O}$ will be mapped to a counterpart whose support lies in $J^-(\mathcal{O}) \cap \mathcal{C}$.\\
{\bf (2)} Notice that, local functionals on $\mathcal{C}$ have functional derivatives whose wavefront set does not respect the inclusion \eqref{eq:WFS}. This is due to the fact that purely spatial directions are present in that wave front set. On the one hand such directions are not compatible with the delta distribution along the angular coordinates present both in the symplectic form on the boundary and in \eqref{eq:notable_2pt_function}. On the other hand, we shall see that the image of local fields under the bulk-to-boundary map is not local on $\mathcal{C}$. These non-localities will make the functional derivatives compatible with \eqref{eq:notable_2pt_function}.\\
{\bf (3)} Huygens principle does not hold for a generic Klein-Gordon equation on curved spacetimes. For this reason we expect that the bulk-to-boundary map applied to an observable $F$ localised in $O\subset \mathcal{D}$ will be supported in $J^-(\mathcal{O}) \cap \mathcal{C}$. Moreover the is no reason why this observable should be supported outside the tip of the cone in $\mathcal{C}$. 
Despite of this fact, thanks to the propagation of singularity theorem the singularities present in the components $F_n$ of $F$ can never reach the tip of the cone. This is due to the fact that null curves stemming from the tip of the cone can never get to $O$ because the latter lies in its interior. In other words, the smoothness condition entails that the singularity present in the tip of the cone is harmless. 

\end{remark}

\noindent We are now ready to introduce the extension of the boundary algebra:

\begin{definition}
Let $\mathcal{F}^{n}_{s}$ be the subset of the totally symmetric
elements in $\mathcal{F}^{n}$ introduced in
Definition~\ref{Def:testdistributions}. The {\bf extended algebra of boundary observables} is defined as 
\begin{equation*}
  \mathscr{A}_{e} ( \mathscr{C} ) = \bigoplus_{n \geqslant 0}
  \mathcal{F}_{s}^{n} \text{,}
\end{equation*}
and its product is given by 
\begin{equation}
  \label{eq:starproductbound}
  \begin{split}
    & \star_{\omega_{\mathcal{C}}} : \mathscr{A}_{e} ( \mathscr{C}_{p}) \times
    \mathscr{A}_{e} ( \mathscr{C}_{p} ) \to \mathscr{A}_{e} (
    \mathscr{C}_{p} ) \text{,} \\
    & ( F \star_{\omega_{\mathcal{C}}} F^{\prime} ) ( \Phi ) =
    \sum_{n=0}^{\infty} \frac{1}{n!} \big\langle F^{ (n)} ( \Phi
    ) , {\omega_{2,\mathcal{C}}}^{\otimes n}\left( F^{\prime (n)} ( \Phi )\right) \big\rangle \text{,}
  \end{split}
\end{equation}
for all $F$, $F^{\prime} \in \mathscr{A}_{e} (
\mathscr{C} )$ and for all $\Phi \in C^{\infty} ( \mathscr{C}
)$. Here $\omega_{2,\mathcal{C}}$ coincides with \eqref{eq:notable_2pt_function}.
\end{definition}
The product $\star_\omega$ is well posed as established in \cite[Prop. 3.5]{DPP}. Once more the crucial point in this proof consists of showing that $\big\langle F^{ (n)} ( \Phi
    ) , {\omega_{2,\mathcal{C}}}^{n} F^{\prime (n)} ( \Phi ) \big\rangle$ is well defined. 
As for algebra of observables in the bulk, this can be obtained making use of microlocal techniques. In particular Theorem 8.2.9 in \cite{Hormander} is used to estimate the wave-front set of $F^{(n)}\otimes {F}^{\prime (n)}$ and of $\omega_{2,\mathcal{C}}^{\otimes n}$ while Theorem 8.2.10 in \cite{Hormander} is used to evaluate the pointwise product of $F^{(n)}\otimes {F^\prime}^{(n)}$ with $\omega_{2,\mathcal{C}}^{\otimes n}$.

As last point, we observe that the subalgebra of $\mathscr{A}_{e}(\mathcal{C})$ formed by those elements $F$ whose components $F_n$ are smooth functions is $*-$homeomorphic to $\mathcal{A}(\mathcal{C})$ defined in Section \ref{sec:observables-on-null-infinity}. This $*-$homomorphism can be understood as a deformation of the symmetric part of the product, in analogy to the counterpart which intertwines $(\mathcal{A}(M),\star)$ with $(\mathcal{A}(M),\star_H)$ as per \eqref{eq:homeomorphisms}. 
Its explicit realization is the following: 
\begin{equation}\label{eq:homeomorphismsb}
    \alpha_{\omega_{2,\mathcal{C}}-iB/2} ( F )  \doteq \sum_{n = 0}^{\infty} \frac{1}{n!}
    \left\langle \left(-\frac{i}{2}B\right)^{\otimes n} , F^{(2n)} \right\rangle \text{.}
\end{equation}
which acts on any $F\in \mathscr{A}_{e}(\mathcal{C})$. Furthermore, $B$ is the integral kernel of the symplectic form $\sigma_\Im$ which is intrinsically defined on $\mathcal{C}$ in \eqref{eq:boundary_symplectic_form}. For completeness we recall that its integral kernel is proportional to the antisymmetric part of $\omega_{2,\mathcal{C}}$ and it reads
\begin{equation}\label{eq:integralkernel}
  B \big( ( V , \theta , \varphi ) , (
  V^{\prime} , \theta^{\prime} , \varphi^{\prime} ) \big) \doteq -
  \frac{\partial^{2}}{\partial V \partial V^{\prime}} \text{sign} ( V -
  V^{\prime}) \medspace \delta ( \theta , \theta^{\prime} ) \text{,}
\end{equation}
where $\text{sign}(x)$ is the sign function.

\section{Interplay Between Algebras and States on
  $\mathscr{D}$ and on $\mathscr{C}$}

Goal of this section is to show that, in full analogy with \eqref{eq:bulk_to_boundary_homomorphism} , there exists an injective $*-$homomorphism between $\mathcal{A}_e^{obs}(\mathcal{D})$ and $\mathcal{A}_e(\mathcal{C})$.
This problem is tackled in \cite[Th. 3.6]{DPP} of which we summarize the content:
\begin{theorem}\label{The:Pi1}
  Consider a double cone $\mathscr{D}$ in $M$ and its past boundary $\mathcal{C}$. Call $\Pi :
  \mathscr{A}_{e}^{obs} ( \mathscr{D} ) \to \mathscr{A}_{e} (
  \mathscr{C} )$ be the linear map whose action on the components $\{F_n\}_n$ of $F\in \mathscr{A}_{e}^{obs} ( \mathscr{D} )$ is
  \begin{equation}
    \label{eq:restriction}
    \Pi_{n} ( F_{n} ) \doteq \sqrt[4]{|{g_{AB}}|_{1}} \dotsm
    \sqrt[4]{|g_{AB}|}_{n} \medspace G^{\otimes n} ( F_{n} )
    \big\vert_{\mathscr{C}^{n}} \text{,}
  \end{equation}
  where $G$ is the causal propagator introduced in Proposition \ref{prop:causal_propagator}, while
  $\vert_{\mathscr{C}}$ denotes the restriction on
  $\mathscr{C}$ and 
  $\sqrt[4]{|g_{AB}|}_{i}$ the square root of the determinant of the metric induced on the spatial sections of the $i-$th cone $\mathcal{C}_i$ by the metric of $\mathcal{D}$. Let $\widehat{\Pi}_{n}$ be the integral kernel of $\Pi_{n}$. Then, the following properties hold true:
  \renewcommand{\labelenumi}{\arabic{enumi})}
  \begin{enumerate}
  \item $\widehat{\Pi}_{n}=\otimes^{n} \widehat{\Pi}_{1}$ and is an element of
    $\mathcal{D}^{\prime} \big( ( \mathscr{C} \times
    \mathscr{D} )^{n} \big)$.  The wave front set of $\widehat{\Pi}_{n}$
    satisfies
    \begin{equation}
      \label{WFPin}
      \text{WF} ( \widehat{\Pi}_{n} ) \subset ( \text{WF} ( \widehat{\Pi}_{1} ) \cup \{ 0
      \} )^{n} \setminus \{ 0 \} \text{.}
    \end{equation}
    Furthermore, if $( x , \zeta_{x} ; y , \zeta_{y} ) \in \text{WF} (
    \hat{\Pi}_{1} )$, then:
  {\renewcommand{\labelenumi}{(\alph{enumi})}
    \begin{enumerate}
    \item neither $\zeta_{x}$ nor $\zeta_{y}$ vanish;
    \item $( \zeta_{x} )_{r} \neq 0$
    \item $( \zeta_{x} )_{r} \geqslant 0$ if and only if $-\zeta_{y}$
      is future directed.
    \end{enumerate}}
  \item $\Pi[\mathcal{A}_{e}^{obs} ( \mathcal{D} )]\subset \mathcal{A}_{e} ( \mathscr{C})$.
  \end{enumerate}
\end{theorem}

Notice that the map $\Pi$ introduced above is nothing but the extension to $\mathcal{F}^n_s$ of $\Gamma_\Im$ defined in Theorem \ref{Th:bulk_to_boundary}. Furthermore, the factors $\sqrt[4]{|g_{AB}|}$ in \eqref{eq:restriction} are the conformal factors of the conformal transformation which maps the induced metric on $\mathcal{C}$ to $\tilde{g}$ as in \eqref{eq:conformal-metric-cone}.

Before analysing the action of the map $\Pi$ on $\mathcal{A}_{e}^{obs} ( \mathcal{D} )$, we show how $\Pi$ can be used to compute the pull-back
of the quasi-free boundary state $\omega_\mathcal{C}$, built out of $\omega_{2,\mathcal{C}}$, whose integral kernel is \eqref{eq:notable_2pt_function}.  The next proposition characterizes the singular structure of the bulk state built out of $\omega_{\mathcal{C}}$:

\begin{proposition} \label{Pro:Hadamard}
  Under the assumptions of Theorem~\ref{The:Pi1}
  \begin{equation}
    \label{eq:pull-back}
    H_{\omega} \doteq \Pi^{\ast} \omega_{2,\mathcal{C}}
  \end{equation}
  is a Hadamard bi-distribution constructed as the pull-back under \eqref{eq:restriction} of $\omega_{2,\mathcal{C}}$ as in \eqref{eq:notable_2pt_function}.
\end{proposition}

The next theorem summarizes the content \cite[Th.3.11]{DPP}, characterizing in particular the action of $\Pi$ both on the algebraic structures and on the boundary state. 

\begin{theorem}
  \label{The:Pi2}
  Under the assumptions of Theorem~\ref{The:Pi1}, it holds that 
  $\Pi$ induces an injective unit-preserving $^{\ast}$-homomorphism between
    the algebras $\big( \mathcal{A}_{e}^{obs} ( \mathcal{D} ) ,
    \star_{H_{\omega}} \big)$ and $\big( \mathscr{A}_{e} (\mathscr{C}) , \star_{\omega_{\mathcal{C}}}\big)$.
\end{theorem}

\begin{remark}
We notice that due to the form of $\Pi$ in \eqref{eq:restriction}, the $^{\ast}$-homomorphism discussed above does not preserve the locality of the elements. In other words when we apply $\Pi$ to a local field, namely to a functional $F$ supported on $\mathcal{D}$ whose functional derivatives are supported on the diagonal, we do not obtain a local field on $\mathcal{C}$. This is not at all surprising since the very same behaviour is present in the $*-$homomorphisms which relates extended algebras localised in sufficiently small neighbourhoods of two different surfaces via the time-slice axiom, see e.g. \cite{ChilianFredenhagen}. 
At the same time, we notice that local fields on $\mathcal{C}$ are not contained in $\mathcal{A}_2^{obs}(\mathcal{C})$ due to the form of their wavefront set. 
\end{remark}

In order to enlighten this features, we outline an example inspired by the discussion of \cite[Sec. 4.1]{DPP}. Let us consider a double cone $\mathcal{D}(p,q)$ in Minkowski spacetime $\mathbb{M}^4\equiv(\mathbb{R}^4,\eta)$. We equip $\mathbb{M}^4$ with the standard Minkowskian coordinated $(t,\vec{x})$ centred around $p$ so that $p$ corresponds to $(0,0)$ while $q$ to $(T,0)$ with $T>0$. Consider a massless, minimally coupled scalar field theory on $\mathbb{M}^4$ and consider the following local functional on $\mathcal{A}_2^{obs}(M)$
\[
\Phi^2(f)\doteq\int_0^T f(t)\phi^2(t,0) dt, \qquad f\in C^\infty_0((0,T))\;
\] 
which represents the Wick ordered squared scalar field, smeared along a timelike geodesic running through $p$ and $q$. Notice that the only non vanishing component of $\Phi^2(f)$ obtained as in \eqref{eq:functional-derivatives-components} is 
\[
(\Phi^2(f))^{(2)}(y,y^\prime) =  2 f(t(y))\delta_3(\vec{x}(y))\delta (y-y^\prime) 
\]
where $(t,\vec{x})$ are the coordinate functions of the chart introduced above, $\delta_3$ is the Dirac delta on $\mathbb{R}^3$ while  $\delta$ is the one supported on the diagonal of $M\times M$. 

We can characterize the action on $\Phi^2(f)$ of the map $\Pi$ given in Theorem \ref{The:Pi1}. To this end we recall that $\Pi$ is given in terms of the causal propagator $G$, which, in Minkowski spacetime, for a massless minimally coupled real scalar field, reads (see \cite{Friedlander:2010eqa} or
\cite{Poisson})
\begin{equation}
  \label{Minkp}
  G ( x , x^{\prime} ) \doteq - \frac{\delta ( t - t^{\prime}
    - |{\bf x} - {\bf x}^{\prime}| )}{4 \pi |{\bf x}
      - {\bf x}^{\prime}|} + \frac{\delta ( t - t^{\prime} +
    |{\bf x} - {\bf x}^{\prime}| )}{4 \pi |{\bf x} -
      {\bf x}^{\prime}|} \text{,}
\end{equation}
where $t$ is the time coordinate, while ${\bf x}$ is the three-dimensional spatial vector in Euclidean coordinates. A direct computation shows that, 
\begin{gather*}
\Pi(\Phi^2(f))(\phi)= \\
\int f(V-t)\delta(V'-V)\phi(V,\theta,\varphi)\phi(V',\theta',\varphi') V V' 
dVd\theta d\varphi dVd\theta' d\varphi'dt
\end{gather*}
where we have represented $\Pi(\Phi^2(f))$ as a functional over the field configuration $\phi$ on $\mathcal{C}$. Notice the manifest non local nature of $\Pi(\Phi^2(f))$. 

\backmatter
%
%
%

\printindex


\end{document}